\documentclass[a4paper,reqno]{amsart}

\usepackage[british]{babel}

\numberwithin{equation}{section}

\usepackage{amssymb}
\usepackage[protrusion=true,expansion=true]{microtype}	
\usepackage{amsmath,amsfonts,amsthm} 
\usepackage[pdftex]{graphicx}	
\usepackage{url}
\usepackage[title,titletoc,page]{appendix} 
\usepackage{listings} 
\usepackage{color}
\usepackage[table]{xcolor}
\usepackage{comment}
\usepackage{caption}
\usepackage{subcaption}
\usepackage{algorithm}
\usepackage{nicematrix}

\usepackage{fancyhdr}
\numberwithin{equation}{section}		
\numberwithin{figure}{section}			
\numberwithin{table}{section}				

\newcommand{\leqs}{\leqslant}
\newcommand{\geqs}{\geqslant}
\newcommand{\babs}[1]{\left|{#1}\right|}
\newcommand{\bnorm}[1]{\left|\left|{#1}\right|\right|}

\newcommand{\vect}[1]{\boldsymbol{\mathbf{#1}}}

\usepackage[left=3.5cm,right=3.5cm,top=3.5cm,bottom=3.5cm,footskip=1.5cm,headsep=1cm]{geometry}
\usepackage[%
pdfauthor={Habib Ammari, Silvio Barandun,Ping Liu},%
pdftitle={Applications of Chebyshev polynomials and Toeplitz theory to topological
metamaterials},
pdfkeywords={Chebyshev polynomials,Toeplitz,topological
metamaterials,k-Toeplitz,skin effect, generalised Brillouin zone}
]{hyperref}

\usepackage[T1]{fontenc} 
\usepackage{lmodern} 
\rmfamily 
\DeclareFontShape{T1}{lmr}{b}{sc}{<->ssub*cmr/bx/sc}{}
\DeclareFontShape{T1}{lmr}{bx}{sc}{<->ssub*cmr/bx/sc}{}
\usepackage{lipsum}
\usepackage{mathtools}
\usepackage{amsmath}
\usepackage{amssymb}
\usepackage{tikz}
\usetikzlibrary{matrix,positioning,decorations.pathreplacing,calc}
\usetikzlibrary{
decorations.text,%
decorations.markings,%
shadows}
\usepackage{adjustbox}
\usepackage{graphicx}

\usepackage{dsfont} 
\usepackage[format=hang]{caption}
\usepackage{subcaption}
\usepackage[normalem]{ulem}

\newcommand{\abs}[1]{\lvert#1\rvert}
\usepackage{bm}

\usepackage[
style=numeric,
sorting=nty,
maxnames=99,
maxalphanames=5,
natbib=true,
sortcites]{biblatex}

\DeclareNameAlias{default}{family-given}

\graphicspath{{figures/}}

\AtEveryBibitem{
 \clearfield{url}
 \clearfield{issn}
 \clearfield{isbn}
 \clearfield{urldate}
 
 \ifentrytype{book}{
  \clearfield{pages}}{%
 }
}

\newcommand{\nm}{\noalign{\smallskip}}
\newcommand{\ds}{\displaystyle}
\usepackage{xargs}                      
\usepackage[prependcaption,textsize=tiny,textwidth=3cm]{todonotes}
\newcommandx{\unsure}[2][1=]{\todo[linecolor=red,backgroundcolor=red!25,bordercolor=red,#1]{#2}}
\newcommandx{\change}[2][1=]{\todo[linecolor=blue,backgroundcolor=blue!25,bordercolor=blue,#1]{#2}}
\newcommandx{\info}[2][1=]{\todo[linecolor=OliveGreen,backgroundcolor=OliveGreen!25,bordercolor=OliveGreen,#1]{#2}}
\newcommandx{\improvement}[2][1=]{\todo[linecolor=black,backgroundcolor=black!25,bordercolor=black,#1]{#2}}
\newcommandx{\thiswillnotshow}[2][1=]{\todo[disable,#1]{#2}}

\usepackage{amsthm}
\usepackage[noabbrev,capitalize]{cleveref}
\crefname{proposition}{Proposition}{Propositions}
\crefname{equation}{}{}

\newtheorem{theorem}{Theorem}[section]
\newtheorem{lemma}[theorem]{Lemma}
\newtheorem{proposition}[theorem]{Proposition}
\newtheorem{corollary}[theorem]{Corollary}

\theoremstyle{definition}
\newtheorem{definition}[theorem]{Definition}

\newtheorem{remark}[theorem]{Remark}

\crefname{assumption}{Assumption}{Assumptions}
\crefname{definition}{Definition}{Definitions}
\crefname{corollary}{Corollary}{Corollaries}
\crefname{enumi}{item}{items}

\usepackage{fancyhdr}

\fancypagestyle{plain}{%
\fancyhead[C]{}
\fancyfoot[C]{}

}
\fancypagestyle{nosection}{%
\fancyhead[CE]{}
\fancyhead[CO]{}
\fancyfoot[CE,CO]{\thepage}
}

\DeclareMathOperator{\N}{\mathbb{N}}
\DeclareMathOperator{\Z}{\mathbb{Z}}

\DeclareMathOperator{\R}{\mathbb{R}}
\DeclareMathOperator{\C}{\mathbb{C}}

\renewcommand{\i}{\mathbf{i}}
\renewcommand{\tilde}{\widetilde}
\renewcommand{\hat}{\widehat}
\renewcommand{\bar}[1]{\overline{#1}}
\newcommand{\inv}{^{-1}}

\renewcommand{\qedsymbol}{$\blacksquare$}

\usepackage{fancyhdr}

\fancypagestyle{plain}{%
\fancyhead[C]{}
\fancyfoot[C]{}

}
\fancypagestyle{nosection}{%
\fancyhead[CE]{}
\fancyhead[CO]{}
\fancyfoot[CE,CO]{\thepage}
}

\DeclareMathOperator{\dtn}{\mathcal{T}}
\DeclareMathOperator{\diag}{diag}
\DeclareMathOperator{\BO}{\mathcal{O}}
\DeclareMathOperator{\exactcapmat}{\mathbf{C}}
\DeclareMathOperator{\capmat}{\mathnormal{C}}

\DeclareMathOperator{\capmatg}{\mathnormal{C}^\gamma}

\newcommand{\pri}{^\prime}
\newcommand{\prii}{^{\prime\prime}}
\renewcommand{\epsilon}{\varepsilon}
\DeclareMathOperator{\dd}{d\!}

\renewcommand{\i}{\mathbf{i}}
\renewcommand{\tilde}{\widetilde}
\renewcommand{\hat}{\widehat}

\DeclareMathOperator{\iL}{{\mathsf{L}}}
\DeclareMathOperator{\iR}{{\mathsf{R}}}
\DeclareMathOperator{\iLR}{{\mathsf{L},\mathsf{R}}}

\usepackage{tcolorbox}
\usetikzlibrary{tikzmark,calc}

\usepackage{enumerate}
\usepackage{upgreek}

\DeclareMathOperator{\interior}{int}

\newcommand{\cm}{C^\gamma}
\newcommand{\cg}{\mathnormal{C}^{\theta,\gamma}}

\newcommand*{\pnr}[1]{\frac{P_N(#1)}{P_{N-1}(#1)}}
\DeclareMathOperator\Arg{Arg}

\newcommand{\qp}{\alpha+\i\beta}
\newcommand{\qpconj}{\alpha'+\i\beta'}
\newcommand{\associateqp}{e^{-\i\spatialperiod(\alpha+\i\beta)}}
\newcommand{\spatialperiod}{\iL}

\newcommandx{\silvio}[2][1=]{\todo[linecolor=blue,backgroundcolor=blue!25,bordercolor=blue,#1]{Silvio: #2}}

\begin{document}
\title[Mathematical foundations of topological metamaterials]{Applications of Chebyshev polynomials and Toeplitz theory to topological
metamaterials}

\author[H. Ammari]{Habib Ammari}
\address{\parbox{\linewidth}{Habib Ammari\\
        ETH Z\"urich, Department of Mathematics, Rämistrasse 101, 8092 Z\"urich, Switzerland.}}
\email{habib.ammari@math.ethz.ch}
\thanks{}

\author[S. Barandun]{Silvio Barandun}
\address{\parbox{\linewidth}{Silvio Barandun\\
        ETH Z\"urich, Department of Mathematics, Rämistrasse 101, 8092 Z\"urich, Switzerland.}}
\email{silvio.barandun@sam.math.ethz.ch}

\author[P. Liu]{Ping Liu}
\address{\parbox{\linewidth}{Ping Liu\\
        School of Mathematical Sciences, Zhejiang University, No. 866, Yuhangtang Road, Hangzhou, 310058, China.}}
\email{pingliu@zju.edu.cn}

\begin{abstract}
We survey  the use of Chebyshev polynomials and Toeplitz theory for studying topological metamaterials. We consider both Hermitian and non-Hermitian systems of subwavelength resonators and provide a mathematical framework to explain some spectacular properties of metamaterials. 
\end{abstract}
\maketitle

\date{}

\bigskip

\par

\bigskip

\\
\tableofcontents

\section{Introduction}

Subwavelength resonators are materials with high-contrast material parameters and size much smaller than the operating wavelength. They are typical building blocks for large, complex structures, known as metamaterials \cite{kadic20193d}. These microstructures can exhibit a variety of exotic and useful properties, such as robust localisation and transport properties at subwavelength scales. The starting point for their mathematical analysis is that, in the subwavelength regime, their  eigenfrequencies and eigenmodes are described using a capacitance matrix formalism. This discrete formulation can be applied to a wide variety of settings, including both finite- and infinite-sized models and both Hermitian and non-Hermitian systems; see \cite{ammari.davies.ea2021Functional,cbms} and the references therein. 

In this survey, we specifically consider one dimensional  metamaterials consisting of subwavelength resonators, which allows for more explicit formulations and applications not only in the field of metamaterials, but also in condensed matter physics due to its analogies with quantum mechanical phenomena. This is due to the fact that, in contrast to higher dimensional cases, interactions between the subwavelength resonators in one dimension only imply the nearest neighbours. Thus, the capacitance matrix formalism in one dimension is analogous to the tight-binding approximation for quantum systems, while in three dimensions some correspondence holds only for dilute resonators \cite{francesco}.

We consider both Hermitian systems of subwavelength resonators where the material parameters are real positive and non-Hermitian ones. In subwavelength wave physics, non-Hermiticity can be obtained via either a reciprocal mechanism by adding gain and loss inside the resonators \cite{miri.alu2019Exceptional}, or via a non-reciprocal one by introducing a directional damping term, which is motivated by an imaginary gauge potential \cite{yokomizo.yoda.ea2022NonHermitian}. On the one hand, introducing gain and loss inside the resonators, represented by the imaginary parts of complex-valued material parameters, can create exceptional points. An exceptional point is a point in parameter space at which two or more eigenmodes coalesce \cite{heiss2012physics,miri.alu2019Exceptional,ammari.davies.ea2022Exceptional,ammari.zhang2015Superresolution}. On the other hand, for systems that are non-Hermitian due to non-reciprocity, the wave propagation is usually amplified in one direction and attenuated in the other. This inherent unidirectional dynamics is related to the non-Hermitian skin effect, which leads to the accumulation of modes at one edge of the structure \cite{hatano.nelson1996Localization,yokomizo.yoda.ea2022NonHermitian,rivero.feng.ea2022Imaginary}. 

Our aim here is to highlight the powerful mathematical tools used for studying systems of subwavelength resonators in one dimension which are based on Chebyshev polynomials and Toeplitz theory. We use these mathematical tools to 
(i) characterise quantitatively the localised (topologically protected) interface modes in systems of finitely many resonators, (ii) study the skin effect in a system of finitely many subwavelength resonators with a non-Hermitian imaginary gauge potential and prove the condensation of bulk eigenmodes at one of the edges
of the system, (iii) prove that all but two eigenmodes of a system of subwavelength resonators equipped with two kinds of non-Hermiticity --- an imaginary gauge potential and on-site gain and loss--pass through exceptional points and decouple, and (iv)  introduce the notion of generalised Brillouin zone to account for the unidirectional spatial decay of the eigenmodes in non-reciprocal systems and obtain correct spectral convergence properties of finite to corresponding infinite systems.

The survey is organised as follows. In \Cref{sec:theory}, we first recall monotonicity properties of Chebyshev polynomials. Then, we characterise the spectra of tridiagonal Toeplitz matrices with perturbations in the diagonal corners and show that their associated eigenvectors exhibit exponential decay. These matrices are crucial in the study of the one-dimensional monomer systems of subwavelength resonators \cite{ammari.barandun.ea2024Mathematical, ammari.davies.ea2021Functional, feppon.ammari2022Subwavelength, ammari.barandun.ea2023Perturbed}.  Finally, we consider tridiagonal $k$-Toeplitz matrices and illustrate the connection between the winding number of the eigenvalues of the symbol function and the exponential decay of the associated eigenvectors. 
In \Cref{sec:1d}, we consider one-dimensional systems of subwavelength resonators as our main physical reference. This choice follows the rigorous derivation from first principles of the governing mathematical model: the capacitance matrix.
We discuss both the Hermitian and the non-reciprocal cases. 
In \Cref{sec: interface modes}, we prove the existence of a spectral gap for defectless finite dimer structures and provide a direct relationship between eigenvalues being within the spectral gap and the localisation of their associated eigenmodes. Then,  we show the existence and uniqueness of an eigenvalue in the gap in  the defect structure, proving the existence of a unique localised interface mode. Our method, based on Chebyshev polynomials, is the first  to characterise quantitatively the localised interface modes in systems of finitely many resonators. \Cref{sec: skin effect} is devoted to the mathematical foundations of the non-Hermitian skin effect in systems of subwavelength resonators where a gauge potential is added inside the resonators to break Hermiticity.  We
prove the condensation of the eigenmodes at one edge of the structure and  show the robustness of such effect with respect to random imperfections in the system. It is worth emphasising that the effect of adding this gauge potential  only inside the subwavelength resonators would be negligible without exciting the structure's subwavelength eigenfrequencies.  In \cref{sec: PT symm EP}, we study the interplay of an imaginary gauge potential and on-site gain and loss. By tuning the gain-to-loss ratio, we prove that a system of subwavelength resonators changes from a phase with unbroken parity-time symmetry to a phase with broken parity-time symmetry. At the macroscopic level,  this is observed as a transition from symmetrical eigenmodes to condensated eigenmodes at one edge of the structure. In \Cref{sec: GBZ},  we show  how the spectral properties of a finite structure are associated with those of the corresponding semi-infinitely or infinitely periodic lattices and give explicit characterisations of how to extend the Floquet--Bloch theory in the Hermitian case to non-reciprocal settings. Finally, in the appendix we outline for completness the derivation of the capacitance matrix approximation to the eigenfrequencies and 
eigenmodes of a system of subwavelength resonators. 

\part{Theoretical foundations}
\section{Chebyshev polynomials and Toeplitz theory} \label{sec:theory}

\subsection{Chebyshev polynomials}
Chebyshev polynomials are two sequences of polynomials related to the cosine and sine functions, denoted respectively by $T_{n}(x)$ and $U_{n}(x)$. In particular, the Chebyshev polynomials of the first kind are obtained from the recurrence relation
\begin{equation}\label{equ:chebyshevpolynomial1}
    \begin{aligned}
        T_0(x)     & =1,                      \\
        T_1(x)     & =x,                      \\
        T_{n+1}(x) & =2 x T_n(x)-T_{n-1}(x) ,
    \end{aligned}
\end{equation}
and the Chebyshev polynomials of the second kind are obtained from the recurrence relation
\begin{equation}\label{equ:chebyshevpolynomial2}
    \begin{aligned}
        U_0(x)     & =1,                      \\
        U_1(x)     & =2x,                     \\
        U_{n+1}(x) & =2 x U_n(x)-U_{n-1}(x) .
    \end{aligned}
\end{equation}

A Chebyshev polynomial of either kind with degree $n$ has $n$ different simple roots, called Chebyshev roots, in the interval $[-1,1]$. The roots of $T_n(x)$ are
\[
    x_k=\cos \left(\frac{\pi(k+1 / 2)}{n}\right), \quad 0 \leqs k \leqs n-1,
\]
and the roots of $U_{n}(x)$ are
\begin{equation}\label{equ:rootsofchebyshevpolynomial1}
    x_k = \cos\left(\frac{k\pi}{n+1}\right),\quad 1 \leqs k \leqs n.
\end{equation}
It is also well-known that for $-1 \leqs x \leqs 1$ and any $n\geqs 0$, we have the upper bounds
\begin{equation}\label{equ:boundofchebyshev1}
    \left|T_n(x)\right| \leq\left|T_n(1)\right|=1, \quad \left|U_n(x)\right| \leq\left|U_n(1)\right|=n+1.
\end{equation}

Although both $T_n(x)$ and $U_n(x)$ are highly oscillating in the interval $[-1, 1]$, they have some monotonicity properties outside $[-1, 1]$.
\begin{proposition}\label{lemma: monotonicity_Chebyshev}
    Let $n\in\N$. Then,
    \begin{enumerate}
        \item[(i)] $T_{2n+1}(x)$ is strictly increasing for $x\in(-\infty, -1)\cup(1,+\infty)$. $T_{2n}(x)$ is strictly increasing for $x\in(1,+\infty)$ and strictly decreasing for $x\in(-\infty, -1)$.
        \item[(ii)] $U_{2n+1}(x)$ is strictly increasing for $x\in(-\infty, -1)\cup(1,+\infty)$. $U_{2n}(x)$ is strictly increasing for $x\in(1,+\infty)$ and strictly decreasing for $x\in(-\infty, -1)$.
        \item[(iii)] $\frac{U_{n-1}(x)}{U_n(x)}$ is strictly decreasing for $x\in(-\infty, -1)\cup(1,+\infty)$.
    \end{enumerate}
\end{proposition}
A proof of item (iii) of \cref{lemma: monotonicity_Chebyshev} can be found in \cite[Lemma 5.2]{ammari.barandun.ea2023Exponentially}.

We show in \cref{fig: cheby poly} the behaviour of the Chebyshev polynomials of first and second kinds and illustrate the opposite behaviors of $T_{n}(x)$ and $U_{n}(x)$ for $x$ inside and outside the interval $[-1,1]$. We will see in the subsequent sections that the two different outside and inside behaviors  correspond to the skin effect and interface modes, respectively. 

\begin{figure}[h]
    \centering
    \begin{subfigure}[t]{0.49\textwidth}
    \centering
\includegraphics[width=\textwidth]{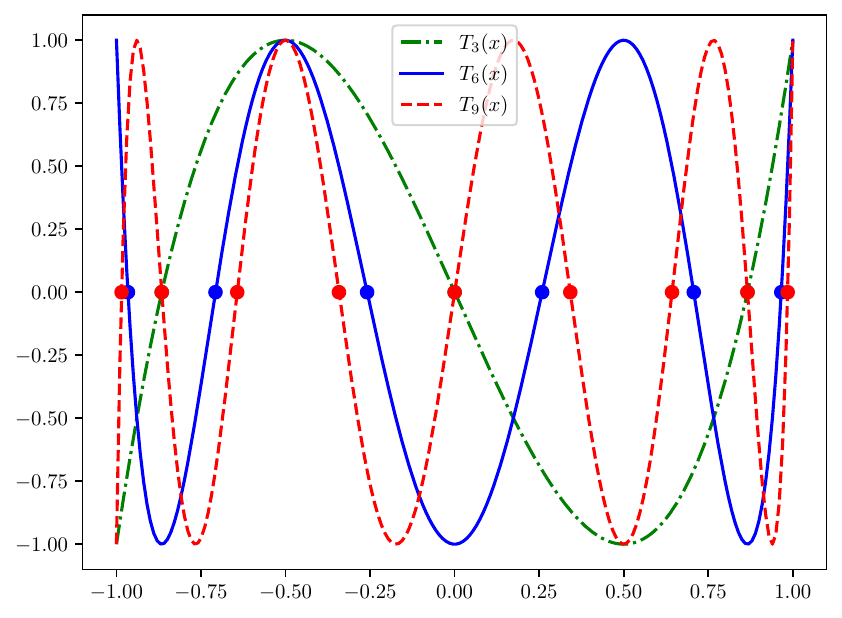}
    \caption{}
    \label{fig: chebyt small}
    \end{subfigure}\hfill
    \begin{subfigure}[t]{0.49\textwidth}
    \centering
    \includegraphics[width=\textwidth]{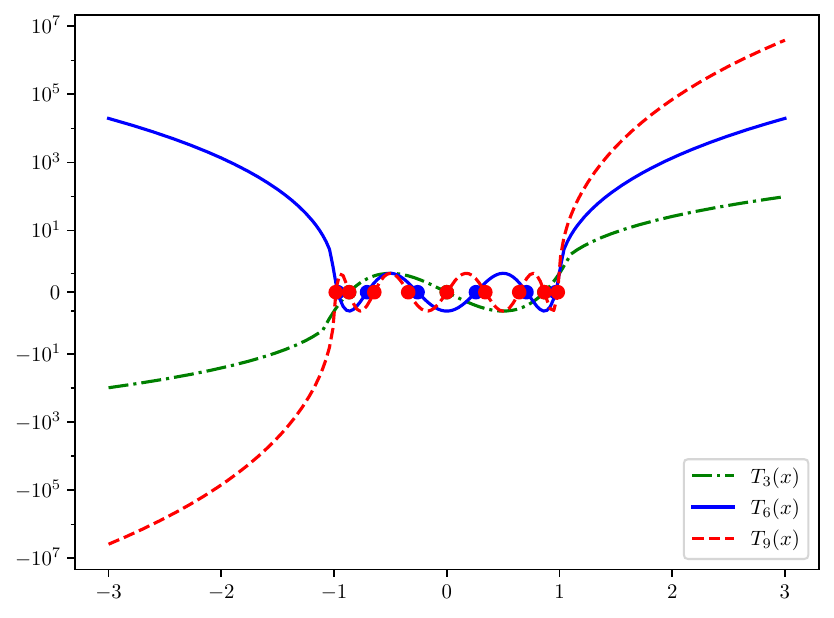}
    \caption{}
    \label{fig: chebyt wide}
    \end{subfigure}
    \\[5mm]
    \begin{subfigure}[t]{0.49\textwidth}
    \centering
    \includegraphics[width=\textwidth]{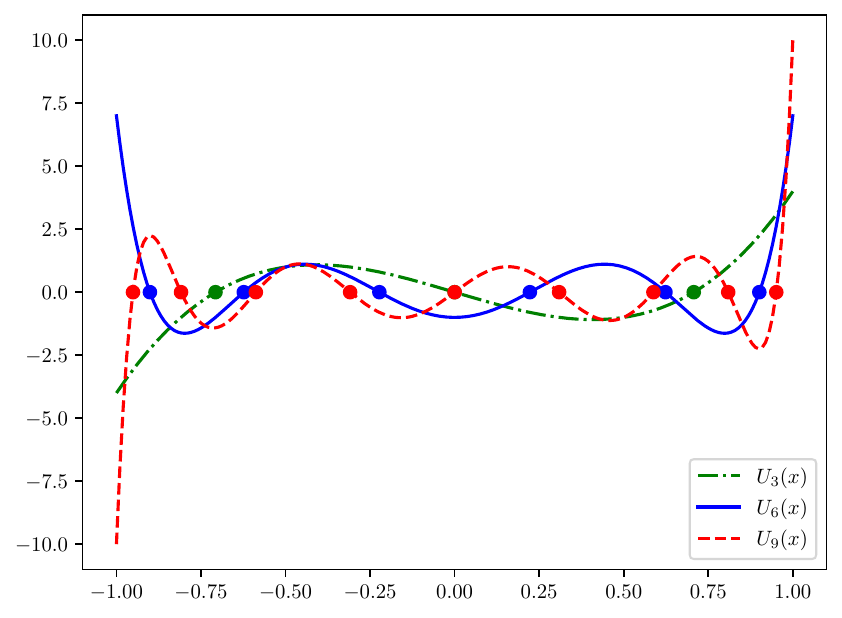}
    \caption{}
    \label{fig: chebyu small}
    \end{subfigure}\hfill
    \begin{subfigure}[t]{0.49\textwidth}
    \includegraphics[width=\textwidth]{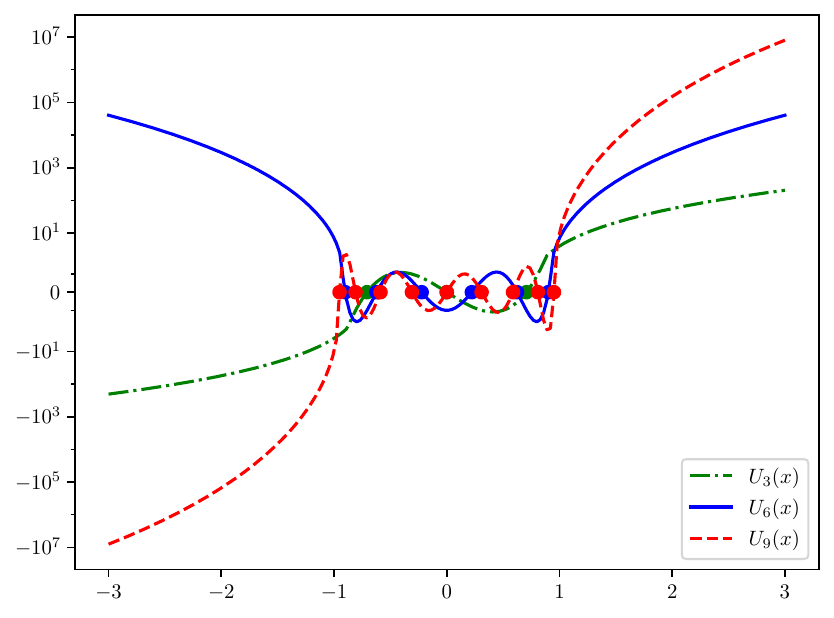}
    \caption{}
    \label{fig: chebyu wide}
    \end{subfigure}
    \caption{The Chebyshev polynomials of first (\cref{fig: chebyt small,fig: chebyt wide}) and second (\cref{fig: chebyu small,fig: chebyu wide}) kinds on a narrow interval (\cref{fig: chebyt small,fig: chebyu small}) where one observes the oscillating behaviour and on a wide interval (\cref{fig: chebyt wide,fig: chebyu wide}) where the monotone behaviour is visible (note the symmetrical log $y$-axis).
    }
    \label{fig: cheby poly}
\end{figure}

\subsection{Toeplitz matrices and operators}\label{sec: Toeplitz theory}

Consider a complex valued function  $f:\mathbb{T}^1\to \C$ on the unit circle $\mathbb{T}^1$ represented by its Fourier series
\begin{align}
    f(z)=\sum_{k\in\Z} a_k z^k.
\end{align}

The $m \times m$ \emph{Toeplitz matrix} associated to $f$ is the matrix whose entries are constant along diagonals and are given by the Fourier coefficients of $f$:
$$
    T_m(f)=\left(\begin{array}{ccccc}
            a_0     & a_{-1} &        & \cdots & a_{1-m} \\
            a_1     & a_0    &        &        & \vdots  \\
                    & \ddots & \ddots & \ddots &         \\
            \vdots  &        &        & a_0    & a_{-1}  \\
            a_{m-1} & \cdots &        & a_1    & a_0
        \end{array}\right).
$$
A semi-infinite matrix of the same form is known as the \emph{Toeplitz operator} associated to $f$ and is denoted by $T(f)$ while a (bi-)infinite matrix of the same form is known as the \emph{Laurent operator} associated to $f$ and is denoted by $L(f)$. In all three cases, we call $f$ the \emph{symbol} associated to the matrix or the operator. Moreover, the Toeplitz operators (matrices) with $a_{j}=0, j> l$ are called $l$-semibanded Toeplitz operators (matrices). The Toeplitz operators (matrices) with $a_{j}=0, |j|> l$ are called $l$-banded Toeplitz operators (matrices).

We assume that the symbol function $f:\mathbb{T}^1\to\C$ is $L^2$-integrable and continuous. This represents a certain loss of generality, but will be more than sufficient for the applications we are interested in.

The Toeplitz operator associated to $f$ as above is an operator on $\ell^2(\N)$, where $\ell^2(\N)$ is the set of square summable sequences. Similarly, the Laurent operator acts on $\ell^2(\Z)$. Note that the assumption that $f$ is continuous is sufficient to ensure that $T(f)$ is bounded. Often we shall fix $f$ and consider the family of Toeplitz matrices $\left\{T_m(f)\right\}$ of increasing dimensions obtained as $m \times m$ finite sections of the infinite matrix $T(f)$. 

Given a point $\lambda \in \mathbb{C} \backslash f(\mathbb{T}^1)$, we define the \emph{winding number of the symbol $f$}, $\operatorname{wind}(f, \lambda)$, to be the winding number of $f$ about $\lambda$ in the usual positive (counterclockwise) sense, that is,
$$
    \operatorname{wind}(f, \lambda) \coloneqq \frac{1}{2\pi \i}\int_{f(\mathbb{T}^1)} (\xi - \lambda)^{-1} \dd \xi.
$$
If $\lambda \in f(\mathbb{T}^1),$ then $\operatorname{wind}(f, \lambda)$ is undefined.

The spectrum of the Toeplitz operator $T(f)$, $$\sigma(T(f)):=\{ \lambda \in \mathbb{C}: (T(f) - \lambda) \text{ is not invertible} \},$$ is given by its symbol $f$ \cite{trefethen.embree2005Spectra, reichel.trefethen1992Eigenvalues}. In particular, the following proposition holds.

\begin{proposition}\label{thm:spectratoeplitz}
    Let $T(f)$ be a Toeplitz operator with continuous symbol $f$. Then $\sigma(T(f))$ is equal to $f(\mathbb{T}^1)$ together with all the points $\lambda$ enclosed by this curve with nonzero winding number $\operatorname{wind}(f, \lambda)$. In particular, for $l$-banded or $l$-semibanded Toeplitz operators $T(f)$, if $\operatorname{wind}(f, \lambda)<0$, then there exists an eigenvector $\bm x\in\ell^2$ of $T(f)$ associated to $\lambda$ and some $\rho<1$ such that
    \begin{equation}\label{eq: bound on eigenvector}
        \frac{\lvert \bm x_j\rvert}{\max_{i}\lvert \bm x_i\rvert} \leq Cj^{l-1} \rho^{j-1},\quad j\geq 1,
    \end{equation}
    where $C>0$ is a constant depending only on $\lambda$ and $f$.
\end{proposition}

Remark that \cref{thm:spectratoeplitz} characterises eigenvectors whose eigenvalues satisfy $$\operatorname{wind}(f, \lambda)<0.$$ Similar characterisation can formally be done with eigenvalues with $\operatorname{wind}(f, \lambda)>0$, but the corresponding eigenvectors would not be square summable.

 For Toeplitz matrices, similar results can be derived. We will need the concept of pseudospectrum.

 \begin{definition}
 Let $\epsilon>0$ and $M$ be a matrix or an operator on a normed space. Then we denote by $\sigma_\epsilon(M)$ the $\epsilon$-pseudospectrum of $M$ defined by
 $$\sigma_\epsilon(M):=\{ \lambda \in \mathbb{C}: 
 \Vert (M - \lambda)^{-1}\Vert > \epsilon^{-1} \}.$$
 Any value $\lambda$ with $\Vert (M - \lambda)^{-1}\Vert > \epsilon^{-1}$ is called $\epsilon$-pseudoeigenvalue. 
 \end{definition}
 The topic of pseudoeigenvalues is extensively discussed in \cite{trefethen.embree2005Spectra}. We recall from there the following result.
 \begin{lemma}
     Let $M$ be an operator with enough regularity (\emph{e.g.}, a matrix) and $\epsilon>0$. Then, the $\epsilon$-pseudospectrum of $M$ is defined equivalently as the set of $z\in\C$:
     \begin{enumerate}
         \item[(i)] $\Vert (M - \lambda)^{-1}\Vert > \epsilon^{-1}$;
         \item[(ii)] $z\in\sigma(M+E)$ for some $E$ with $\Vert E \Vert < \epsilon$;
         \item[(iii)] $\sigma \in \sigma(M)$ or $\Vert (z-M)u\Vert<\epsilon$ for some $u$ with $\Vert u\Vert = 1$.
     \end{enumerate}
 \end{lemma}
A vector $u$ as in item (iii) is called an $\epsilon$-pseudoeigenvector to $z$.
 
  Let  the tridiagonal Toeplitz matrix $T_N(f)$ be the finite, truncated, counterpart of $T(f)$ of size $N\times N$. Then the following proposition holds. 
\begin{proposition}\label{prop: specra of fam of toeplitz}
Consider a symbol $f$ of a banded or semi-banded Toeplitz matrix. Let $\lambda\in\C$ be such that $\operatorname{wind}(f,\lambda)\neq 0$. Then there exists $M>1$ so that, for $N$ sufficiently large, $\lambda$ is an $M^{-N}$-pseudoeigenvalue,
    $$
        \Vert (T_N(f) - \lambda)^{-1}\Vert \geq M^N,
    $$
    with an exponentially localised pseudoeigenvector $v$ such that
    \begin{align*}
        \frac{\vert v^{(j)}\vert}{\max_j \vert v^{(j)}\vert} \leq \begin{cases}
                                                                      M^{-j}  & \text{ if } \operatorname{wind}(f,\lambda) <0, \\
                                                                      \nm
                                                                      M^{j-N} & \text{ if } \operatorname{wind}(f,\lambda) >0.
                                                                  \end{cases}
    \end{align*}
\end{proposition}

\subsection{Spectra of tridiagonal Toeplitz matrices with perturbations} \label{sec: spectra tridiagonal toeplitz}
In this section, we characterise the spectra of tridiagonal Toeplitz matrices with perturbations in the diagonal corners and show that their associated eigenvectors exhibit exponential decay. 

We denote by $T_{m}^{(a,b)}$ the tridiagonal Toeplitz matrix of order $m$ with perturbations $(a, b)$ in the diagonal corners, that is,
\begin{equation}\label{equ:toeplitzmatrix1}
    T_{m}^{(a, b)}=\left(\begin{array}{cccccc}
            \alpha+a & \beta  &            &           &          &               \\
            \eta    & \alpha & \beta  &           &          &               \\
                         & \eta   & \alpha & \beta &          &               \\
                         &            & \ddots     & \ddots    & \ddots   &               \\
                         &            &            & \eta  & \alpha & \beta      \\
                         &            &            &           & \eta & \alpha+ b \\
        \end{array}\right).
\end{equation}
The tridiagonal Toeplitz matrix is thus $T_m^{(0,0)}$. Its spectrum is given as follows.

\begin{lemma}[Eigenvalues and eigenvectors of tridiagonal Toeplitz matrices]\label{lemma: eigs of toeplitz} The eigenvalues of $T_m^{(0, 0)}$ are given by
    $$
        \lambda_k=\alpha+2 \sqrt{\eta \beta} \cos \left(\frac{k}{m+1} \pi\right), \quad 1 \leq k \leq m,
    $$
    and the corresponding eigenvectors are
    $$
        \bm{v}_k^{(i)}=\left(\frac{\eta}{\beta}\right)^{\frac{1}{2}} \sin \left(\frac{i k}{m+1} \pi\right), \quad 1 \leq k \leq m \text { and } 1 \leq i \leq m,
    $$
    where $\bm{v}_k^{(i)}$ is the $i$\textsuperscript{th} coefficient of the $k$\textsuperscript{th} eigenvector.
\end{lemma}

The spectral theory of the Toeplitz matrices with perturbations in the diagonal corners, \emph{i.e.}, $T_{m}^{(a, b)}$, is also well understood, as the following result \cite[Theorem 3.1]{yueh.cheng2008Explicit} shows. 
\begin{lemma}\label{lemma: spectrum perturbed toeplitz}
    Suppose that $\eta \beta \neq 0, a = \eta, b = \beta,  \alpha+\beta+\eta=0$ and let $\lambda$ be an eigenvalue of $T_m^{(a, b)}$. Then, either $\lambda=\lambda_1:=0$ and the corresponding eigenvector is $\bm{v}_1=\mathbf{1} =(1,\ldots, 1)^\top$, where $\top$ denotes the transpose,  or
    $$
        \lambda=\lambda_k:=\alpha+2 \sqrt{\eta \beta} \cos \left(\frac{\pi}{m}(k-1)\right), \quad 2 \leq k \leq m,
    $$
    and the corresponding eigenvector $\bm{v_k}= (\bm{v}_k^{(1)}, \ldots, \bm{v}_k^{(m)})^\top$ is given by
    $$
        \bm{v}_k^{(j)}=\left(\frac{\eta}{\beta}\right)^{\frac{j-1}{2}}\left(a \sin \left(\frac{j(k-1) \pi}{m}\right)-\eta\sqrt{\frac{\eta}{\beta}} \sin \left(\frac{(j-1)(k-1) \pi}{m}\right)\right)
    $$
    for $2 \leq k \leq m$ and $1 \leq j \leq m$.
\end{lemma}

\subsection{Spectra of tridiagonal $2$-Toeplitz matrices with perturbations}\label{sec: spectrum of tridiag 2 toeplitz matrices with perturbations}
In this section, we generalise the results obtained in the previous section and characterise the spectra of tridiagonal $2$-Toeplitz matrices with perturbations in the diagonal corners. 

Let $A_{2m+1}^{(a,b)}$ be the tridiagonal $2$-Toeplitz matrix of order $2m+1$ with perturbations $(a, b)$ in the diagonal corners, that is,
\begin{equation}\label{equ:oddmatrixtwoperturb1}
    A_{2m+1}^{(a, b)}=\left(\begin{array}{ccccccc}
            \alpha_{1}+a & \beta_{1}  &            &            &          &          &               \\
            \eta_{1}     & \alpha_{2} & \beta_{2}  &            &          &          &               \\
                         & \eta_{2}   & \alpha_{1} & \beta_{1}  &          &          &               \\
                         &            & \eta_{1}   & \alpha_{2} & \ddots   &          &               \\
                         &            &            & \ddots     & \ddots   & \ddots   &               \\
                         &            &            &            & \eta_{1} & \alpha_2 & \beta_2       \\
                         &            &            &            &          & \eta_{2} & \alpha_{1}+ b \\
        \end{array}\right).
\end{equation}
Here, $\beta_i, \eta_i, \alpha_i$, $i=1,2$ and $a,b$ are in $\mathbb{R}$. Denote furthermore by $A_{2m}^{(a,b)}$ the tridiagonal $2$-Toeplitz matrix of order $2m$ with perturbations $(a, b)$ in the diagonal corners, that is,
\begin{equation}\label{equ:evenmatrixtwoperturb1}
    A_{2m}^{(a,b)}=\left(\begin{array}{ccccccc}
            \alpha_{1}+a & \beta_{1}  &            &            &          &          &               \\
            \eta_{1}     & \alpha_{2} & \beta_{2}  &            &          &          &               \\
                         & \eta_{2}   & \alpha_{1} & \beta_{1}  &          &          &               \\
                         &            & \eta_{1}   & \alpha_{2} & \ddots   &          &               \\
                         &            &            & \ddots     & \ddots   & \ddots   &               \\
                         &            &            &            & \eta_{2} & \alpha_1 & \beta_1       \\
                         &            &            &            &          & \eta_{1} & \alpha_{2}+ b \\
        \end{array}\right).
\end{equation}
We assume throughout that the off-diagonal elements in (\ref{equ:oddmatrixtwoperturb1}) and (\ref{equ:evenmatrixtwoperturb1}) are nonzero and satisfy the following condition: $$\eta_i\beta_i >0, \quad i=1,2.$$

We first consider the eigenvalues of $A_{2m+1}^{(0,0)}$ and $A_{2m}^{(0,0)}$ and define the polynomials
\begin{equation}\label{equ:defiofpkstar1}
P_m^*(x)=\left(\sqrt{\eta_{1}\beta_1 \eta_2\beta_{2}}\right)^m U_m\left(\frac{\left(x-\alpha_1\right)\left(x-\alpha_{2}\right)-\eta_{1}\beta_1-\eta_2\beta_2}{2 \sqrt{\eta_{1}\beta_{1} \eta_{2}\beta_2}}\right),
\end{equation}
where $U_m$ is the Chebyshev polynomial of the second kind defined in (\ref{equ:chebyshevpolynomial2}). It is well-known (see \cite{gover1994eigenproblem, dafonseca.petronilho2001Explicit, 2006Orthogonal}) that the characteristic polynomials of the $2$-Toeplitz matrices $A_{2m+1}^{(0,0)}, A_{2m}^{(0,0)}$ are respectively
\begin{equation*}
    \chi_{A_{2m+1}^{(0, 0)}}(x)=\left(x-\alpha_{1}\right) P_m^*\left(x\right)
\end{equation*}
and
\begin{equation*}
    \chi_{A_{2m}^{(0, 0)}}(x)=P_m^*\left(x\right)+\eta_{2}\beta_{2} P_{m-1}^*\left(x\right).
\end{equation*}
More generally, it can be shown that the characteristic polynomials of $A_{2m+1}^{(a, b)}, A_{2m}^{(a, b)}$ are respectively
\begin{equation*}
    \begin{aligned}
        \chi_{A_{2m+1}^{(a, b)}}(x) = & \left(x-\alpha_{1}-a-b\right) P_m^*\left(x\right) +\left( ab\left(x-\alpha_{2}\right)  -a \eta_{1}\beta_{1}-b \eta_{2}\beta_{2}\right) P_{m-1}^*\left(x\right)
    \end{aligned}
\end{equation*}
and
\begin{equation*}
    \begin{aligned}
        \chi_{A_{2m}^{(a, b)}}(x)
        = & P_m^*\left(x\right)+\left(a\left(\alpha_2-x\right)+b\left(\alpha_1-x\right)+ab+\eta_{2}\beta_{2}\right) P_{m-1}^*\left(x\right) +a b \eta_1\beta_{1} P_{m-2}^*\left(x\right).
    \end{aligned}
\end{equation*}
In particular, for the case when $\alpha_1=\alpha_2=\alpha, \eta_j= \beta_j, j=1,2$, the characteristic polynomials of $A_{2m+1}^{(a, b)}$ and $A_{2m}^{(a, b)}$ are respectively
\begin{align}\label{equ:eigenpolynomial5}
    \chi_{A_{2m+1}^{(a, b)}}(x) = \left(x-\alpha-a-b\right) P_m^*\left(x\right)+\left( ab\left(x-\alpha\right)  -a \beta_{1}^2-b \beta_{2}^2\right) P_{m-1}^*\left(x\right)
\end{align}
and
\begin{align}\label{equ:eigenpolynomial6}
    \chi_{A_{2m}^{(a, b)}}(x)=P_m^*\left(x\right)+\left((a+b)\left(\alpha-x\right)+ab+\beta_{2}^2\right) P_{m-1}^*\left(x\right)+a b \beta_{1}^2 P_{m-2}^*\left(x\right)
\end{align}
for the corresponding $P_m^*(x)$'s.

We now present a detailed characterisation of the eigenvalues of $A_{2m+1}^{(a, b)}$ and $A_{2m}^{(a, b)}$. All the following results in this section and their justification can be found in the authors' paper \cite{ammari.barandun.ea2023Perturbed}.

\begin{theorem}\label{thm:eigenvaluethm1}
    Let $m$ be large enough. The eigenvalues $\lambda_r$'s of $A_{2m+1}^{(a, b)}$ are all real numbers. Except for at most $11$ eigenvalues, we can reindex the $\lambda_r$'s to have
    \[
        \lambda_{3}^l < \lambda_{4}^l<\cdots< \lambda_{m-3}^l\leq \min\{\alpha_1, \alpha_2\}\leq  \max\{\alpha_1, \alpha_2\}\leq \lambda_{m-3}^r< \cdots< \lambda_{4}^{r}< \lambda_{3}^{r}.
    \]
    In particular, for $k =3,\cdots, m-3$,
    \begin{equation}\label{equ:eigenvalue1}
        \cos\left(\frac{k\pi}{m}\right)\leq \min\left\{y\left(\lambda_k^l\right), y\left(\lambda_k^r\right)\right\}\leq \max\left\{y\left(\lambda_k^l\right), y\left(\lambda_k^r\right)\right\} \leq \cos\left(\frac{(k-2)\pi}{m}\right)
    \end{equation}
    with
    \begin{equation}\label{equ:normalizedfunction1}
        y(x)=\frac{(x-\alpha_1)(x-\alpha_2)-\eta_1\beta_1-\eta_2\beta_2}{2\sqrt{\eta_1\beta_1\eta_2\beta_2}}.
    \end{equation}
\end{theorem}

\begin{theorem}\label{thm:eigenvaluethm2}
    Let $m$ be large enough. The eigenvalues $\{\lambda_r\}$ of $A_{2m}^{(a, b)}$ are all real numbers. Except for at most $12$ eigenvalues, we can reindex the $\lambda_r$'s to have
    \[
        \lambda_{3}^l < \lambda_{4}^l<\cdots< \lambda_{m-4}^l\leq \min\{\alpha_1, \alpha_2\}\leq  \max\{\alpha_1, \alpha_2\}\leq \lambda_{m-4}^r< \cdots< \lambda_{4}^{r}< \lambda_{3}^{r}.
    \]
    In particular, for $k =3,\cdots, m-4$,
    \begin{equation}\label{equ:eigenvalue2}
        \cos\left(\frac{(k+1)\pi}{m}\right)\leq \min\left\{y\left(\lambda_k^l\right), y\left(\lambda_k^r\right)\right\}\leq \max\left\{y\left(\lambda_k^l\right), y\left(\lambda_k^r\right)\right\} \leq \cos\left(\frac{(k-2)\pi}{m}\right)
    \end{equation}
    with $y(x)$ being defined by (\ref{equ:normalizedfunction1}).
\end{theorem}

We next present explicit formulas for the eigenvectors of tridiagonal $2$-Toeplitz matrices with perturbations on the diagonal corners, that is the matrices $A_{2m+1}^{(a,b)}$ and $A_{2m}^{(a,b)}$, through a direct construction. We start by introducing the following two families of polynomials.
\begin{definition}\label{defi:generalqmpm1}
    We define the two families of polynomials $q_{k}^{(\xi_p, \xi_{q})}, p_{k}^{(\xi_p, \xi_{q})}$ by
    \[
        q_{0}^{(\xi_p, \xi_{q})}(\nu)=\xi_{q}, \qquad p_{0}^{\xi_p, \xi_{p}}(\nu)=\xi_{p},
    \]
    and the recurrence formulas
    \begin{align}\label{equ:recurrencerelation5}
        q_{k}^{(\xi_p, \xi_{q})}(\nu)&=\nu p_{k-1}^{(\xi_p, \xi_{q})}(\nu)-q_{k-1}^{(\xi_p, \xi_{q})}(\nu),\\
        \nm
\label{equ:recurrencerelation6}
        p_{k}^{(\xi_p, \xi_{q})}(\nu) &= q_{k}^{(\xi_p, \xi_{q})}(\nu) - \zeta p_{k-1}^{(\xi_p, \xi_{q})}(\nu),
    \end{align}
    where
    \begin{equation}\label{equ:zeta2}
        \zeta = \frac{\eta_{2} \beta_{2}}{\eta_{1} \beta_{1}}.
    \end{equation}
\end{definition}

\medskip
The recurrence formulas (\ref{equ:recurrencerelation5}) and (\ref{equ:recurrencerelation6}) can be simplified as follows.
\begin{proposition}\label{thm:recurrencerelation1}
    If $p_k^{(\xi_p, \xi_{q})}(\nu)$ and $q_k^{(\xi_p, \xi_{q})}(\nu)$ satisfy (\ref{equ:recurrencerelation5}) and (\ref{equ:recurrencerelation6}) respectively, then
    \begin{equation}\label{equ:recurrencerelation3}
        p_{k+1}^{(\xi_p, \xi_{q})}(\nu)=[\nu-(1+\zeta)] p_k^{(\xi_p, \xi_{q})}(\nu)-\zeta p_{k-1}^{(\xi_p, \xi_{q})}(\nu)
    \end{equation}
    and
    \begin{equation}\label{equ:recurrencerelation4}
        q_{k+1}^{(\xi_p, \xi_{q})}(\nu)=[\nu-(1+\zeta)] q_k^{(\xi_p, \xi_{q})}(\nu)-\zeta q_{k-1}^{(\xi_p, \xi_{q})}(\nu)
    \end{equation}
    with
    \begin{equation}\label{equ:recurrenceinitial1}
        \begin{aligned}
             & p_0^{(\xi_p, \xi_{q})}(\nu)=\xi_p, \quad p_1^{(\xi_p, \xi_{q})}(\nu)=\left(\nu-\zeta\right)\xi_{p}-\xi_{q}, \\
             & q_0^{(\xi_p, \xi_{q})}(\nu)=\xi_{q}, \quad q_1^{(\xi_p, \xi_{q})}(\nu)=\nu\xi_{p}-\xi_{q},
        \end{aligned}
    \end{equation}
    where $\zeta$ is defined in (\ref{equ:zeta2}).
\end{proposition}

For later use of the Chebyshev polynomials, we further normalise the polynomials $p_{k}^{(\xi_p, \xi_{q})}(\nu), q_{k}^{(\xi_p, \xi_{q})}(\nu)$. This can be achieved by setting
\begin{equation}\label{equ:defineofmu1}
    \mu=\frac{\nu-\left(1+\beta^2\right)}{2\beta}
\end{equation}
with
$$
    \beta^2=\frac{\beta_2 \eta_2}{\beta_1 \eta_1},
$$
and
\begin{equation}\label{equ:normalizepolynomial1}
    \widehat p_k^{(\xi_{p}, \xi_{q})}(\mu)=\frac{1}{\beta^k} p_k^{(\xi_{p}, \xi_{q})}\left(1+2\beta \mu+\beta^2\right), \ \widehat q_k^{(\xi_{p}, \xi_{q})}(\mu)=\frac{1}{\beta^k} q_k^{(\xi_{p}, \xi_{q})}\left(1+2\beta \mu+\beta^2\right).
\end{equation}
Thus, from (\ref{equ:recurrencerelation3}) and (\ref{equ:recurrencerelation4}), we respectively get
\begin{equation}\label{equ:Chebyshevrecurrence1}
    \widehat p_{k+1}^{(\xi_{p}, \xi_{q})}(\mu)=2\mu \widehat p_k^{(\xi_{p}, \xi_{q})}(\mu)-\widehat p_{k-1}^{(\xi_{p}, \xi_{q})}(\mu),
\end{equation}
and
\begin{equation}\label{equ:Chebyshevrecurrence2}
    \widehat q_{k+1}^{(\xi_{p}, \xi_{q})}(\mu)=2\mu \widehat q_k^{(\xi_{p}, \xi_{q})}(\mu)-\widehat q_{k-1}^{(\xi_{p}, \xi_{q})}(\mu).
\end{equation}
Equations (\ref{equ:Chebyshevrecurrence1}) and (\ref{equ:Chebyshevrecurrence2}) are the Chebyshev three point recurrence formula. Also, from (\ref{equ:recurrenceinitial1}),  the initial polynomials are given by
\begin{equation}\label{equ:recurrenceinitial2}
    \begin{aligned}
         & \widehat p_{0}^{(\xi_{p}, \xi_{q})}(\mu) =\xi_{p},\quad  \widehat p_{1}^{(\xi_{p}, \xi_{q})}(\mu) =2\mu\xi_{p}+ \frac{\xi_{p}-\xi_{q}}{\beta},        \\
         & \widehat q_{0}^{(\xi_{p}, \xi_{q})}(\mu) =\xi_{q},\quad  \widehat q_{1}^{(\xi_{p}, \xi_{q})}(\mu) =(2\mu+\beta)\xi_{p}+\frac{\xi_{p}-\xi_{q}}{\beta}.
    \end{aligned}
\end{equation}

\bigskip
Now, we characterise the eigenvectors of $A_{2m+1}^{(a,b)}$ and $A_{2m}^{(a,b)}$ by the ``Chebyshev like'' polynomials introduced above.
\begin{theorem}\label{thm:eigenvectoroddmatrixtwoperturb2}
    The eigenvector of $A_{2m+1}^{(a, b)}$ in (\ref{equ:oddmatrixtwoperturb1}) associated with the eigenvalue $\lambda_r$ is given by
    \begin{equation}
        \begin{aligned}\label{equ:eigenvectoroddmatrixtwoperturb3}
            \bm x= & \left(\widehat q_{0}^{(\xi_p, \xi_{q})}\left(\mu_r\right),-\frac{1}{\beta_1}\left(\alpha_1-\lambda_r\right) \widehat p_0^{(\xi_p, \xi_{q})}\left(\mu_r\right), s \widehat q_1^{(\xi_p, \xi_{q})}\left(\mu_r\right),  -\frac{1}{\beta_1} s\left(\alpha_1-\lambda_r\right) \widehat p_1^{(\xi_p, \xi_{q})}\left(\mu_r\right),\right. \\
                        & \qquad  \left. \ldots,  -\frac{1}{\beta_1} s^{m-1}\left(\alpha_1-\lambda_r\right) \widehat p_{m-1}^{(\xi_p, \xi_{q})}\left(\mu_r\right), s^m \widehat q_m^{(\xi_p, \xi_{q})}\left(\mu_r\right)\right)^{\top},
        \end{aligned}
    \end{equation}
    where $\widehat q_{k}^{(\xi_p, \xi_{q})}, \widehat p_{k}^{(\xi_p, \xi_{q})}$ are defined as in (\ref{equ:normalizepolynomial1}) with $$\xi_{q}=(\alpha_1-\lambda_r), \xi_{p}= (\alpha_1+a-\lambda_{r}),$$ and $s, \mu_r$ are respectively given by
    \begin{equation}\label{equ:defiofstwoperturb2}
        s= \sqrt{\frac{\eta_{1}\eta_{2}}{\beta_{1}\beta_2}}, \quad \mu_r = \frac{\left(\alpha_1-\lambda_r\right)\left(\alpha_2-\lambda_r\right)-(\eta_1\beta_1+\eta_2\beta_2)}{2\sqrt{\eta_1\beta_1\eta_2\beta_2}}.
    \end{equation}
    The eigenvector of $A_{2m}^{(a,b)}$ in (\ref{equ:evenmatrixtwoperturb1}) associated with the eigenvalue $\lambda_r$ is given by the first $2m$ elements of $\bm x$ in (\ref{equ:eigenvectoroddmatrixtwoperturb3}).
\end{theorem}

\begin{remark}
    By (\ref{equ:recurrenceinitial2}), in Theorem \ref{thm:eigenvectoroddmatrixtwoperturb2}, the initial values of $\widehat p_{k}^{(\xi_{p}, \xi_{q})}(\mu_r)$ and $\widehat q_{k}^{(\xi_{p}, \xi_{q})}(\mu_r)$ are
    \begin{equation}\label{equ:recurrenceinitial2b}
        \begin{array}{ll}
     \ds       \widehat p_{0}^{(\xi_{p}, \xi_{q})}(\mu_r) =(\alpha_1-\lambda_r),     &  \ds \widehat p_{1}^{(\xi_{p}, \xi_{q})}(\mu_r) =2\mu_r(\alpha_1-\lambda_r) -\frac{a}{\beta},         \\
     \nm
          \ds   \widehat q_{0}^{(\xi_{p}, \xi_{q})}(\mu_r) =(\alpha_1+a-\lambda_{r}), & \ds \widehat q_{1}^{(\xi_{p}, \xi_{q})}(\mu_r) =(2\mu_r+\beta)(\alpha_1-\lambda_r) -\frac{a}{\beta}.
        \end{array}
    \end{equation}
\end{remark}

\subsection{Tridiagonal $k$-Toeplitz operators and 
 matrices and their spectra}\label{section:ktoeplitztheory}
In \cref{sec: spectra tridiagonal toeplitz}, we have analysed the spectrum of tridiagonal Toeplitz matrices. Then we have generalised our results in \cref{sec: spectrum of tridiag 2 toeplitz matrices with perturbations} to tridiagonal $2$-Toeplitz matrices with an extensive complexity overhead. It is natural to ask how the spectrum of a general tridiagonal $k$-Toeplitz operator looks like and this will be done in this section where we will also show that their associated eigenvectors exhibit exponential decay. The tools are, once more, different from those in the previous two cases.
We refer the reader to \cite{gover1994eigenproblem} for the spectral theory of tridiagonal $2$-Toeplitz matrices (without any perturbation) and to \cite{dafonseca2007characteristic} and  the references therein
for further results on the eigenvalues of  tridiagonal $k$-Toeplitz.

Note that tridiagonal $k$-Toeplitz operators and their associated truncated operators (\emph{i.e.}, matrices) are crucial in the study of the one-dimensional polymer systems of subwavelength resonators \cite{ammari.barandun.ea2024Mathematical, ammari.davies.ea2021Functional, feppon.ammari2022Subwavelength, ammari.barandun.ea2023Perturbed}. 
A $k$-Toeplitz operator is a semi-infinite matrix of the following kind:
\begin{equation}\label{equ:tridiagonalktoeplitzop1}
    T(f) = A = \begin{pmatrix}
        \alpha_1 & \beta_1  & 0          &              &             &          &                  \\
        \eta_1   & \alpha_2 & \beta_2    & \ddots       &                                           \\
        0        & \ddots   & \ddots     & \ddots       & \ddots                                    \\
                 & \ddots   & \eta_{k-2} & \alpha_{k-1} & \beta_{k-1} & \ddots                      \\
                 &          & \ddots     & \eta_{k-1}   & \alpha_k    & \beta_k  & \ddots           \\
                 &          &            & \ddots       & \eta_{k}    & \alpha_1 & \beta_1 & \ddots \\
                 &          &            &              & \ddots      & \ddots   & \ddots  & \ddots \\
    \end{pmatrix},
\end{equation}
where $\alpha_i,\beta_i,\eta_i\in\C$ for $i\in\N$ and $A_{ij}=0$ if $\vert i-j \vert > 1$ for $i,j\in\N$. A bi-infinite matrix of the same form is called a tridiagonal $k$-Laurent operator. We consider that these two operators act on sequences of square integrable numbers.

Recall, furthermore, that the tridiagonal $k$-Toeplitz operator $A$ presented in (\ref{equ:tridiagonalktoeplitzop1}) can be reformulated as a tridiagonal block Toeplitz operator, where the blocks repeat in a $1$-periodic way:
\begin{equation}\label{equ:tridiagonalktoeplitzop2}
    A = \begin{pmatrix}
        A_0   & A_{-1} &        \\
        A_{1} & A_0    & \ddots \\
              & \ddots & \ddots \\
    \end{pmatrix},
\end{equation}
with
\begin{equation*}
    A_0 = \begin{pmatrix}
        \alpha_1 & \beta_1  & 0          & \cdots       & 0           \\
        \eta_1   & \alpha_2 & \beta_2    & \ddots       & \vdots      \\
        0        & \ddots   & \ddots     & \ddots       & 0           \\
        \vdots   & \ddots   & \eta_{k-2} & \alpha_{k-1} & \beta_{k-1} \\
        0        & \cdots   & 0          & \eta_{k-1}   & \alpha_k    \\
    \end{pmatrix},  A_{-1} = \begin{pmatrix}
        0       & \cdots & \cdots & 0      \\
        \vdots  & \ddots &        & \vdots \\
        0       &        & \ddots & \vdots \\
        \beta_k & 0      & \cdots & 0
    \end{pmatrix}, A_{1}=\begin{pmatrix}
        0      & \cdots & 0      & \eta_k \\
        \vdots & \ddots &        & 0      \\
        \vdots &        & \ddots & \vdots \\
        0      & \cdots & \cdots & 0
    \end{pmatrix}.
\end{equation*}
In this case, the symbol of $A$ in \eqref{equ:tridiagonalktoeplitzop1} is
\begin{align}
    f:\mathbb{T}^1 & \to \C^{k\times k}\nonumber        \\
    z     & \mapsto A_{-1}z^{-1} + A_0 + A_1z,
\end{align}
or explicitly
\begin{equation}\label{eq: symbol tridiagonal operator}
    f(z) =
    \begin{pmatrix}
        \alpha_1      & \beta_1  & 0        & \dots      & 0            & \eta_k z    \\
        \eta_1        & \alpha_2 & \beta_2  & 0          & \dots        & 0           \\
        0             & \eta_2   & \alpha_3 & \beta_3    & \ddots       & \vdots      \\
        \vdots        & \ddots   & \ddots   & \ddots     & \ddots       & 0           \\
        0             & \dots    & 0        & \eta_{k-2} & \alpha_{k-1} & \beta_{k-1} \\
        \beta_kz^{-1} & 0        & \dots    & 0          & \eta_{k-1}   & \alpha_k    \\
    \end{pmatrix}.
\end{equation}

Further, let now $n, k\geq 1$ and define the projections
\begin{align}
    P_n:\ell^2(\N,\C) & \to\ell^2(\N,\C)\nonumber         \\
    (x_1,x_2,\dots)   & \mapsto(x_1,\dots x_n,0,0,\dots).
\end{align}
Then, the \emph{$k$-Toeplitz matrix} of order $mk$ for $m\in\N$ associated to the symbol $f\in L^\infty(\mathbb{T}^1,\C^{k\times k})$ is given by
$$
    T_{m\times k}(f)\coloneqq P_{mk}T(f)P_{mk},
$$
and can be identified as an $mk\times mk$ matrix. The spectra of tridiagonal $k$-Toeplitz matrices is not well understood yet. For the results of the pseudospectra of tridiagonal $k$-Toeplitz matrices, we refer the reader to  \cite[Section 3]{ammari.barandun.ea2024Spectra}.

Now, we start to introduce the spectra of the operator $T(f)$. We denote by $\sigma, \sigma_{\mathrm{ess}}$  the spectrum and the essential spectrum of the operator, respectively. We define
\begin{equation}\label{equ:detspectraformula1}
    \sigma_{\mathrm{det}}(f) = \{\lambda\in \mathbb C: \det(f(z)-\lambda)=0,\ \exists  z\in \mathbb{T}^1 \}.
\end{equation}
We first have the following result characterising the essential spectrum of $T(f)$.

\begin{proposition}\label{thm: essential spectrum}
    The operator $T(f)$ in \eqref{equ:tridiagonalktoeplitzop1} is Fredholm if and only if
    \(\mathrm{det}(f)\) has no zeros on \(\mathbb{T}^1\). Furthermore, the essential spectrum of $T(f)$ is given by
    \begin{equation*}
        \sigma_{\mathrm{ess}}(T(f)) = \sigma_{\mathrm{det}}(f).
    \end{equation*}
\end{proposition}

Next, we provide a full characterisation of $\sigma(T(f))$. Before stating the result, we first define
\begin{equation}\label{equ:windspectraformula1}
    \sigma_{\mathrm{wind}}(f)\coloneqq \{\lambda \in \mathbb{C}\setminus \sigma_{\mathrm{det}}(f) :  \operatorname{wind}(\det(f - \lambda),0) \neq 0 \}.
\end{equation}
By \cite[ Lemma A.1]{ammari.barandun.ea2024Spectra},
\[
    \det(f(z) - \lambda) = (-1)^{k+1}\left(\prod_{i = 1}^k \eta_i\right) z + (-1)^{k+1}\left(\prod_{i = 1}^k \beta_i\right) z^{-1} + g(\lambda)
\]
with $g(\lambda)$ being given by \cite[(A.2)]{ammari.barandun.ea2024Spectra}. This yields
\begin{equation}\label{equ:windspectraformula2}
    \sigma_{\mathrm{wind}}(f)= \left\{\lambda \in \mathbb{C}\setminus \sigma_{\mathrm{det}}(f):  \operatorname{wind}\left((-1)^{k+1} \left((\prod_{i = 1}^k \eta_i) z + (\prod_{i = 1}^k \beta_i) z^{-1}\right), -g(\lambda)\right) \neq 0 \right\}.
\end{equation}
Furthermore, since
\[
    \mathrm{det}(f(z)-\lambda) = \prod_{j = 1}^k (\lambda_j(z)-\lambda),
\]
where $\lambda_j(z), 1\leq j\leq k,$ are the eigenvalues of the matrix $f(z)$, 
we obtain that
\begin{align*}
    \sigma_{\mathrm{wind}}(f)= & \left\{\lambda \in \mathbb{C}\setminus \sigma_{\mathrm{det}}(f) :  \sum_{j=1}^k\operatorname{wind}(\lambda_j(\mathbb T) - \lambda,0) \neq 0 \right\}, \end{align*}
or equivalently, 
    \begin{align}
       \sigma_{\mathrm{wind}}(f)=                      & \left\{\lambda \in \mathbb{C}\setminus \sigma_{\mathrm{det}}(f) :  \sum_{j=1}^k\operatorname{wind}(\lambda_j(\mathbb T), \lambda) \neq 0 \right\}. \label{equ:windspectraformula3}
\end{align}
In \eqref{equ:windspectraformula3} and in the rest of the text, eigenvalues traces are considered to be concatenated if their endings match.

We now state the following result on the spectrum of the operator $T(f)$. We refer to \cite{ammari.barandun.ea2024Spectra} for its proof. 
\begin{theorem}\label{thm: windingspectrum}
    Let $f \in \C^{k\times k}(\mathbb{T}^1)$ be the symbol of a tridiagonal $k$-Toeplitz operator $T(f)$. Denote by $B_0$ the leading $(k-1)\times(k-1)$ principal minor of $A_0$. It holds that
    \[
        \sigma_{\mathrm{det}}(f)\cup \sigma_{\mathrm{wind}}(f)  \subset \sigma(T(f)) \subset \sigma_{\mathrm{det}}(f)\cup \sigma_{\mathrm{wind}}(f) \cup \sigma(B_0),
    \]
    where $\sigma_{\mathrm{det}}(f), \sigma_{\mathrm{wind}}(f)$ are given by (\ref{equ:detspectraformula1}) and (\ref{equ:windspectraformula3}), respectively.
\end{theorem}

Note that this characterisation is optimal. It cannot be made more precise, as it cannot be guaranteed that an eigenvalue $\lambda$ of $B_0$ with $\operatorname{wind}(\det(f-\lambda),0)=0$ is also an eigenvalue of $T(f)$; see the example given in \cite{ammari.barandun.ea2024Spectra}.

Now that we have characterised the spectra of tridiagonal $k$-Toeplitz operators, we venture onwards to explore their eigenvectors. In particular, we have the following theorem from \cite{ammari.barandun.ea2024Spectra} showing that the eigenvectors exhibit exponential decay.

\begin{theorem}\label{thm: exponential_decay_k_operators}
    Suppose $\Pi_{j=1}^k \eta_j\neq 0$ and $\Pi_{j=1}^k \beta_j\neq 0$. Let $f(z) \in \mathbb{C}^{k\times k}$ be the symbol (\ref{eq: symbol tridiagonal operator}) and let $\lambda \in \C\setminus\sigma_{\mathrm{ess}}(T(f))$. If
    $\sum_{j=1}^{k}\operatorname{wind}(\lambda_j, \lambda)<0$ for $\lambda_j$ being the eigenvalues of $f(z)$, then there exists an eigenvector $\bm x\in\ell^2$ of $T(f)$ associated to $\lambda$ and some $\rho<1$ such that
    \begin{equation}\label{eq: bound on eigenvector2}
        \frac{\lvert \bm x_j\rvert}{\max_{i}\lvert \bm x_i\rvert} \leq C_1 \lceil j/k\rceil \rho^{\lceil j/k\rceil-1},\quad j\geq 1,
    \end{equation}
    where $C_1>0$ is a constant depending only on $\lambda, \alpha_p, \beta_p, \eta_p$ for $1\leq p\leq k$.
\end{theorem}

Last, we introduce a condition for the symbol $f$ to be collapsed. We call a symbol $f$ \emph{collapsed} if for all $\lambda\in \C$ it holds that $\operatorname{wind}(\det(f-\lambda),0)=0$. This corresponds to the fact that the curves traced out by the eigenvalues of the symbol $\mathbb{T}^1\ni z \mapsto\sigma(f(z))$ do not generate winding regions. Relating this to \cref{thm:spectratoeplitz} and its $k$-Toeplitz analogue, there is no exponential behaviour of the eigenvectors of a collapsed symbol.

In the tridiagonal setting this is easy to identify; see \cite{ammari.barandun.ea2024Generalised}.
\begin{lemma}\label{lemma: collapsed symbol}
    Let $T(f)$ be a symmetric or Hermitian tridiagonal $k$-Toeplitz operator. Then the symbol $f(z)$ is collapsed.
\end{lemma} 

\section{One-dimensional subwavelength systems} \label{sec:1d}
We consider a one-dimensional chain of $N$ disjoint subwavelength resonators $D_i\coloneqq (x_i^{\iL},x_i^{\iR})$, where $(x_i^{\iLR})_{1\<i\<N} \subset \R$ are the $2N$ extremities satisfying $x_i^{\iL} < x_i^{\iR} <  x_{i+1}^{\iL}$ for any $1\leq i \leq N$. We fix the coordinates such that $x_1^{\iL}=0$. We also denote by  $\ell_i = x_i^{\iR} - x_i^{\iL}$ the length of the $i$\textsuperscript{th} resonators,  and by $s_i= x_{i+1}^{\iL} -x_i^{\iR}$ the spacing between the $i$\textsuperscript{th} and $(i+1)$\textsuperscript{th} resonator. The system is illustrated in \cref{fig:setting}. We will use
\begin{align*}
    D\coloneqq \bigcup_{i=1}^N(x_i^{\iL},x_i^{\iR})
\end{align*}
to symbolise the set of subwavelength resonators.

\begin{figure}[htb]
    \centering
    \begin{adjustbox}{width=\textwidth}
        \begin{tikzpicture}
            \coordinate (x1l) at (1,0);
            \path (x1l) +(1,0) coordinate (x1r);
            \path (x1r) +(0.75,0.7) coordinate (s1);
            \path (x1r) +(1.5,0) coordinate (x2l);
            \path (x2l) +(1,0) coordinate (x2r);
            \path (x2r) +(0.5,0.7) coordinate (s2);
            \path (x2r) +(1,0) coordinate (x3l);
            \path (x3l) +(1,0) coordinate (x3r);
            \path (x3r) +(1,0.7) coordinate (s3);
            \path (x3r) +(2,0) coordinate (x4l);
            \path (x4l) +(1,0) coordinate (x4r);
            \path (x4r) +(0.4,0.7) coordinate (s4);
            \path (x4r) +(1,0) coordinate (dots);
            \path (dots) +(1,0) coordinate (x5l);
            \path (x5l) +(1,0) coordinate (x5r);
            \path (x5r) +(1.75,0) coordinate (x6l);
            \path (x5r) +(0.875,0.7) coordinate (s5);
            \path (x6l) +(1,0) coordinate (x6r);
            \path (x6r) +(1.25,0) coordinate (x7l);
            \path (x6r) +(0.525,0.7) coordinate (s6);
            \path (x7l) +(1,0) coordinate (x7r);
            \path (x7r) +(1.5,0) coordinate (x8l);
            \path (x7r) +(0.75,0.7) coordinate (s7);
            \path (x8l) +(1,0) coordinate (x8r);
            \draw[ultra thick] (x1l) -- (x1r);
            \node[anchor=north] (label1) at (x1l) {$x_1^{\iL}$};
            \node[anchor=north] (label1) at (x1r) {$x_1^{\iR}$};
            \node[anchor=south] (label1) at ($(x1l)!0.5!(x1r)$) {$\ell_1$};
            \draw[dotted,|-|] ($(x1r)+(0,0.25)$) -- ($(x2l)+(0,0.25)$);
            \draw[ultra thick] (x2l) -- (x2r);
            \node[anchor=north] (label1) at (x2l) {$x_2^{\iL}$};
            \node[anchor=north] (label1) at (x2r) {$x_2^{\iR}$};
            \node[anchor=south] (label1) at ($(x2l)!0.5!(x2r)$) {$\ell_2$};
            \draw[dotted,|-|] ($(x2r)+(0,0.25)$) -- ($(x3l)+(0,0.25)$);
            \draw[ultra thick] (x3l) -- (x3r);
            \node[anchor=north] (label1) at (x3l) {$x_3^{\iL}$};
            \node[anchor=north] (label1) at (x3r) {$x_3^{\iR}$};
            \node[anchor=south] (label1) at ($(x3l)!0.5!(x3r)$) {$\ell_3$};
            \draw[dotted,|-|] ($(x3r)+(0,0.25)$) -- ($(x4l)+(0,0.25)$);
            \node (dots) at (dots) {\dots};
            \draw[ultra thick] (x4l) -- (x4r);
            \node[anchor=north] (label1) at (x4l) {$x_4^{\iL}$};
            \node[anchor=north] (label1) at (x4r) {$x_4^{\iR}$};
            \node[anchor=south] (label1) at ($(x4l)!0.5!(x4r)$) {$\ell_4$};
            \draw[dotted,|-|] ($(x4r)+(0,0.25)$) -- ($(dots)+(-.25,0.25)$);
            \draw[ultra thick] (x5l) -- (x5r);
            \node[anchor=north] (label1) at (x5l) {$x_{N-3}^{\iL}$};
            \node[anchor=north] (label1) at (x5r) {$x_{N-3}^{\iR}$};
            \node[anchor=south] (label1) at ($(x5l)!0.5!(x5r)$) {$\ell_{N-3}$};
            \draw[dotted,|-|] ($(x5r)+(0,0.25)$) -- ($(x6l)+(0,0.25)$);
            \draw[ultra thick] (x6l) -- (x6r);
            \node[anchor=north] (label1) at (x6l) {$x_{N-2}^{\iL}$};
            \node[anchor=north] (label1) at (x6r) {$x_{N-2}^{\iR}$};
            \node[anchor=south] (label1) at ($(x6l)!0.5!(x6r)$) {$\ell_{N-2}$};
            \draw[dotted,|-|] ($(x6r)+(0,0.25)$) -- ($(x7l)+(0,0.25)$);
            \draw[ultra thick] (x7l) -- (x7r);
            \node[anchor=north] (label1) at (x7l) {$x_{N-1}^{\iL}$};
            \node[anchor=north] (label1) at (x7r) {$x_{N-1}^{\iR}$};
            \node[anchor=south] (label1) at ($(x7l)!0.5!(x7r)$) {$\ell_{N-1}$};
            \draw[dotted,|-|] ($(x7r)+(0,0.25)$) -- ($(x8l)+(0,0.25)$);
            \draw[ultra thick] (x8l) -- (x8r);
            \node[anchor=north] (label1) at (x8l) {$x_{N}^{\iL}$};
            \node[anchor=north] (label1) at (x8r) {$x_{N}^{\iR}$};
            \node[anchor=south] (label1) at ($(x8l)!0.5!(x8r)$) {$\ell_N$};
            \node[anchor=north] (label1) at (s1) {$s_1$};
            \node[anchor=north] (label1) at (s2) {$s_2$};
            \node[anchor=north] (label1) at (s3) {$s_3$};
            \node[anchor=north] (label1) at (s4) {$s_4$};
            \node[anchor=north] (label1) at (s5) {$s_{N-2}$};
            \node[anchor=north] (label1) at (s6) {$s_{N-1}$};
            \node[anchor=north] (label1) at (s7) {$s_N$};
        \end{tikzpicture}
    \end{adjustbox}
    \caption{A chain of $N$ subwavelength resonators, with lengths
        $(\ell_i)_{1\leq i\leq N}$ and spacings $(s_{i})_{1\leq i\leq N-1}$.}
    \label{fig:setting}
\end{figure}
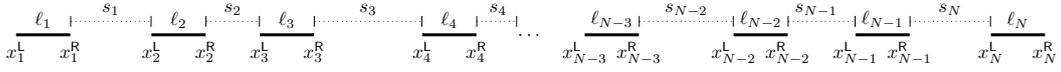

As a scalar wave field $u(t,x)$ propagates in a heterogeneous medium, it satisfies the following one-dimensional wave equation:
\begin{align*}
    \frac{1}{\kappa(x)}\frac{\partial{^2}}{\partial t^{2}}u(t,x) -\frac{\partial}{\partial x}\left(
    \frac{1}{\rho(x)}\frac{\partial}{\partial x}  u(t,x)\right) = 0, \qquad (t,x)\in\R\times\R.
\end{align*}
The parameters $\kappa(x)$ and $\rho(x)$ are the material parameters of the medium. We consider piecewise constant material parameters
\begin{align*}
    \kappa(x)=
    \begin{dcases}
        \kappa_b & x\in D,            \\
        \kappa   & x\in\R\setminus D,
    \end{dcases}\quad\text{and}\quad
    \rho(x)=
    \begin{dcases}
        \rho_b & x\in D,            \\
        \rho   & x\in\R\setminus D,
    \end{dcases}
\end{align*}
where the constants $\rho_b, \rho, \kappa, \kappa_b \in \R_{>0}$. The wave speeds inside the resonators $D$ and inside the background medium $\R\setminus D$, are denoted respectively by $v_b$ and $v$, the wave numbers respectively by $k_b$ and $k$, and the contrast between the densities of the resonators and the background medium by $\delta$:
\begin{align}
    v_b:=\sqrt{\frac{\kappa_b}{\rho_b}}, \qquad v:=\sqrt{\frac{\kappa}{\rho}},\qquad
    k_b:=\frac{\omega}{v_b},\qquad k:=\frac{\omega}{v},\qquad
    \delta:=\frac{\rho_b}{\rho}.
\end{align}

Up to using a Fourier decomposition in time, we can assume that
the total wave field $u(t,x)$ is time-harmonic:
\begin{align*}
    u(t,x)=\Re ( e^{-\i\omega t}u(x) ),
\end{align*}
for a function $u(x)$ which solves the one-dimensional Helmholtz equation:
\begin{align}
    -\frac{\omega^{2}}{\kappa(x)}u(x)-\frac{\dd}{\dd x}\left( \frac{1}{\rho(x)}\frac{\dd}{\dd
        x}  u(x)\right) =0,\qquad x \in\R.
    \label{eq:time indip Helmholtz}
\end{align}

In these circumstances of piecewise constant material parameters, the wave problem determined by \eqref{eq:time indip Helmholtz} can be rewritten as the following system of coupled one-dimensional Helmholtz equations:

\begin{align}
    \begin{dcases}
        u\prii(x) +\frac{\omega^2}{v_b^2}u(x)=0,                                                                                   & x\in D,                                         \\
        u\prii(x) + \frac{\omega^2}{v^2}u(x)=0,                                                                                    & x\in\R\setminus D,                              \\
        u\vert_{\iR}(x^{\iLR}_i) - u\vert_{\iL}(x^{\iLR}_i) = 0,                                                                   & \text{for all } 1\leq i\leq N,                  \\
        \left.\frac{\dd u}{\dd x}\right\vert_{\iR}(x^{\iL}_{{i}})=\delta\left.\frac{\dd u}{\dd x}\right\vert_{\iL}(x^{\iL}_{{i}}), & \text{for all } 1\leq i\leq N,                  \\
        \delta\left.\frac{\dd u}{\dd x}\right\vert_{\iR}(x^{\iR}_{{i}})=\left.\frac{\dd u}{\dd x}\right\vert_{\iL}(x^{\iR}_{{i}}), & \text{for all } 1\leq i\leq N,                  \\
        \frac{\dd u}{\dd\ \abs{x}}(x) -\i k u(x) = 0,                                                                              & x\in(-\infty,x_1^{\iL})\cup (x_N^{\iR},\infty),
    \end{dcases}
    \label{eq:coupled ods}
\end{align}
where, for a one-dimensional function $w$, we denote by
\begin{align*}
    w\vert_{\iL}(x) \coloneqq \lim_{\substack{s\to 0 \\ s>0}}w(x-s) \quad \mbox{and} \quad  w\vert_{\iR}(x) \coloneqq \lim_{\substack{s\to 0\\ s>0}}w(x+s)
\end{align*}
if the limits exist.

We are interested in the subwavelength eigenfrequencies $\omega\in\C$ such that \eqref{eq:coupled ods} has a nontrivial solution. We will perform an asymptotic analysis in a high-contrast limit, given by $\delta\to 0$. We will look for the subwavelength modes within this regime, which we characterise by requiring that $\omega \to 0$ as $\delta\to 0$. This limit will recover subwavelength resonances, while keeping the size of the resonators fixed.

In all what follows we will denote by $H^1(D)$ the usual Sobolev spaces of complex-valued functions on $D$.

\begin{definition}[Subwavelength eigenfrequency]
    A (complex) frequency $\omega(\delta) \in \C$, with nonnegative real part, satisfying
    \begin{align}
        \omega(\delta) \to 0\quad \text{as}\quad \delta\to 0
        \label{eq: asymptotic behaviour subwavelength resonances}
    \end{align}
    and such that \eqref{eq:coupled ods} admits a nonzero solution $u(\omega,\delta)\in H^1(D)$ for $\omega=\omega(\delta)$ is called a \emph{subwavelength eigenfrequency}. The solution $u(\omega,\delta)$ is called a \emph{subwavelength eigenmode}.
\end{definition}

We can immediately see that $\omega = 0$ is a trivial solution to \eqref{eq:coupled ods} corresponding to an eigenmode which is constant across all the resonators. In what follows, we will restrict attention to other solutions of \eqref{eq:coupled ods}.
\subsection{Hermitian case}
In the high-contrast, low-frequency scattering problems, recent work has shown that capacitance matrices can be used to provide asymptotic characterisations of the resonant modes \cite{ammari.davies.ea2021Functional}. Capacitance matrices were first introduced to describe many-body electrostatic systems by Maxwell \cite{maxwell1873treatise}, and have recently enjoyed a resurgence in subwavelength physics. In particular, for one-dimensional models like the one considered here, a capacitance matrix formulation has proven to be very effective in both efficiently solving the problem as well as providing valuable insights in the solution \cite{feppon.cheng.ea2023Subwavelength, ammari.barandun.ea2023Edge}.

For a system as in \cref{fig:setting}, we define the \emph{capacitance matrix} $C$ as
\begin{align}
    \label{eq: def Her cap mat}
    C_{i,j}\coloneqq -\int_{\partial D_i}\frac{\partial V_j}{\partial \nu}\dd x \qquad \text{with}\qquad
    \begin{dcases}
        -\frac{\text{d}^2}{\dd x^2} V_i = 0\quad \text{in } \R\setminus D, \\
        V_i(x_j^{\iLR})=\delta_{i,j} ,                          \\
        V_i(x)=\BO(1)\quad \text{for } \vert x\vert \to + \infty ,
    \end{dcases}
\end{align}
where $\delta_{i,j}$ is the Kronecker delta and $\nu$ the outward pointing normal. Then, it is easy to see that
\begin{align}
    C_{i,j} = - \frac{1}{s_{j-1}}\delta_{i,j-1}+\left(\frac{1}{s_{j-1}}+\frac{1}{s_j}\right)\delta_{i,j}-\frac{1}{s_j}\delta_{i,j+1}, \quad 1\leq i,j\leq N
\end{align}
with the convention that $\frac{1}{s_j}=0$ for $j<1$ and $j>N$. The capacitance matrix has thus the following tridiagonal form:
\begin{equation}\label{equ:capacitance matrix hermitian}
    \begin{aligned}
        \begingroup
        \renewcommand*{\arraystretch}{1.5}
        \begin{pmatrix}
            \cellcolor{red!20}s_1^{-1}     & \cellcolor{green!20}-s_1^{-1}                                                                                                                                                                                           \\
            \cellcolor{yellow!20}-s_1^{-1} & \cellcolor{red!20}s_1^{-1}+s_2^{-1} & \cellcolor{green!20}-s_2^{-1}                                                                                                                                                     \\
                                           & \cellcolor{yellow!20}-s_2^{-1}      & \cellcolor{red!20}\ddots      & \cellcolor{green!20}\ddots                                                                                                                        \\
                                           &                                     & \cellcolor{yellow!20}\ddots   & \cellcolor{red!20}\ddots    & \cellcolor{green!20}\ddots                                                                                          \\
                                           &                                     &                               & \cellcolor{yellow!20}\ddots & \cellcolor{red!20}\ddots           & \cellcolor{green!20}-s_{N-2}^{-1}                                              \\
                                           &                                     &                               &                             & \cellcolor{yellow!20}-s_{N-2}^{-1} & \cellcolor{red!20}s_{N-2}^{-1}+s_{N-1}^{-1} & \cellcolor{green!20}s_{N-1}^{-1} \\
                                           &                                     &                               &                             &                                    & \cellcolor{yellow!20}-s_{N-1}^{-1}          & \cellcolor{red!20}s_{N-1}^{-1}
        \end{pmatrix}.
        \endgroup
    \end{aligned}
\end{equation}
The main findings of \cite{feppon.cheng.ea2023Subwavelength} are the following.
\begin{proposition} \label{prop:capa approx hermitian}
    The scattering problem \eqref{eq:coupled ods} admits exactly $N$ subwavelength resonant frequencies as $\delta\to0$, namely
    \begin{align*}
        \omega_i =  v_b \sqrt{\delta\lambda_i} + \BO(\delta),
    \end{align*}
    where $(\lambda_i)_{1\leq i\leq N}$ are the eigenvalues of the eigenvalue problem
    \begin{equation}
        \label{eq:eigevalue problem capacitance matrix}
        V^2C\bm a_i = \lambda_i L \bm a_i,\qquad 1\leq i\leq N,
    \end{equation}
    where $L\coloneqq\operatorname{diag}(\ell_1,\dots,\ell_N)$ and $V \coloneqq \operatorname{diag}(v_1,\dots,v_N)$.
    Furthermore, let $u_i$ be a subwavelength eigenmode corresponding to $\omega_i$ and let $\bm a_i$ be the corresponding eigenvector of $C$. Then
    \begin{align*}
        u_i(x) = \sum_j \bm a_i^{(j)}V_j(x) + \BO(\delta),
    \end{align*}
    where $V_j$ is defined by \eqref{eq: def Her cap mat} and $\bm a_i^{(j)}$ denotes the $j$\textsuperscript{th} entry of the eigenvector $\bm a_i$.
\end{proposition}

Note that the choice of the $N$ ``physical'' resonant frequencies between $\pm \omega_i$ for $1\leqs i\leqs N$ occurs through the radiation condition. Moreover, the study of Hermitian subwavelength systems is thus entirely governed at the leading order in $\delta$ by the linear eigenvalue problem \eqref{eq:eigevalue problem capacitance matrix}.
\subsection{Non-reciprocal case}\label{sec: non reciprocal 1D systems}
Hermitian systems are a first step, but a much richer physics is observed for non-Hermitian systems. In this subsection, we study one subset of them: non-reciprocal systems. To this end, we consider the following variation of \eqref{eq:time indip Helmholtz}:
\begin{align}
    -\frac{\omega^{2}}{\kappa(x)}u(x)- \gamma(x) \frac{\dd}{\dd x}u(x)-\frac{\dd}{\dd x}\left( \frac{1}{\rho(x)}\frac{\dd}{\dd
        x}  u(x)\right) =0,\qquad x \in\R,
    \label{eq: gen Strum-Liouville}
\end{align}
for a piecewise constant coefficient
\begin{align*}
    \gamma(x) = \begin{dcases}
                    \gamma,\quad x\in D, \\
                    \nm
                    0,\quad x \in \R\setminus D.
                \end{dcases}
\end{align*}
This new parameter $\gamma$ extends the usual scalar wave equation to a generalised Strum--Liouville equation via the introduction of an imaginary gauge potential \cite{yokomizo.yoda.ea2022NonHermitian}. Alternatively, one may think about \eqref{eq: gen Strum-Liouville} as a damped wave equation where the damping acts in the space variable instead of the time variable.

In these circumstances of piecewise constant material parameters, the wave problem determined by \eqref{eq: gen Strum-Liouville} can be rewritten as the following system of coupled one-dimensional equations:

\begin{align}
    \begin{dcases}
        u\prii(x) + \gamma u\pri(x)+\frac{\omega^2}{v_b^2}u(x)=0,                                                                  & x\in D,                                         \\
        u\prii(x) + \frac{\omega^2}{v^2}u(x)=0,                                                                                    & x\in\R\setminus D,                              \\
        u\vert_{\iR}(x^{\iLR}_i) - u\vert_{\iL}(x^{\iLR}_i) = 0,                                                                   & \text{for all } 1\leq i\leq N,                  \\
        \left.\frac{\dd u}{\dd x}\right\vert_{\iR}(x^{\iL}_{{i}})=\delta\left.\frac{\dd u}{\dd x}\right\vert_{\iL}(x^{\iL}_{{i}}), & \text{for all } 1\leq i\leq N,                  \\
        \delta\left.\frac{\dd u}{\dd x}\right\vert_{\iR}(x^{\iR}_{{i}})=\left.\frac{\dd u}{\dd x}\right\vert_{\iL}(x^{\iR}_{{i}}), & \text{for all } 1\leq i\leq N,                  \\
        \frac{\dd u}{\dd\ \abs{x}}(x) -\i k u(x) = 0,                                                                              & x\in(-\infty,x_1^{\iL})\cup (x_N^{\iR},\infty).
    \end{dcases}
    \label{eq:coupled ods2}
\end{align}

We seek a similar capacitance matrix formulation for the case of subwavelength resonators with an imaginary gauge potential.

In the case of systems of Hermitian subwavelength resonators (\emph{i.e.}, when $\gamma=0$), the entries of the capacitance matrix are defined by \eqref{eq: def Her cap mat}. For non-reciprocal systems the formulation is different. 

\begin{definition}[Gauge capacitance matrix]
    For $\gamma \in \R\setminus\{0\}$,  we define the \emph{gauge capacitance matrix} $\capmatg\in\R^{N\times N}$ by
    \begin{align}
        \capmat_{i,j}^\gamma \coloneqq -\frac{|D_i|}{\int_{D_i}e^{\gamma x}\dd x}\int_{\partial D_i} e^{\gamma x} \frac{\partial V_j(x)}{\partial \nu}\dd \sigma,
        \label{eq: def cap mat}
    \end{align}
    where $V_j$ is given by \eqref{eq: def Her cap mat}.
\end{definition}

It is easy to see that the gauge capacitance matrix is tridiagonal, non-symmetric,  and is given by
\begin{align}
    \capmat_{i,j}^\gamma \coloneqq \begin{dcases}
                                       \frac{\gamma}{s_1} \frac{\ell_1}{1-e^{-\gamma \ell_1}},                                                            & i=j=1,               \\
                                       \frac{\gamma}{s_i} \frac{\ell_i}{1-e^{-\gamma \ell_i}} -\frac{\gamma}{s_{i-1}} \frac{\ell_i}{1-e^{\gamma \ell_i}}, & 1< i=j< N,           \\
                                       - \frac{\gamma}{s_i} \frac{\ell_i}{1-e^{-\gamma \ell_j}},                                                          & 1\leq i=j-1\leq N-1, \\
                                       \frac{\gamma}{s_j} \frac{\ell_i}{1-e^{\gamma \ell_j}},                                                             & 2\leq i=j+1\leq N,   \\
                                       - \frac{\gamma}{s_{N-1}} \frac{\ell_N}{1-e^{\gamma \ell_N}},                                                       & i=j=N,               \\
                                   \end{dcases}\label{eq: explicit coef cap mat}
\end{align}
while all the other entries are zero. More visually, we can display the capacitance matrix for $\ell_i=\ell$ and $s=s_i$ for all $i$ as
\begin{equation}\label{equ:capacitance matrix hermitian2}
    \begin{aligned}
        \capmat^\gamma\coloneqq \frac{\gamma}{s}
        \begingroup
        \renewcommand*{\arraystretch}{1.5}
        \begin{pmatrix}
            \cellcolor{red!20}  \frac{\ell}{1-e^{-\gamma \ell}}   & \cellcolor{green!20}-   \frac{\ell}{1-e^{-\gamma \ell}}                                                                                                                                                                                                                                                                                                                                      \\
            \cellcolor{yellow!20}  \frac{\ell}{1-e^{\gamma \ell}} & \cellcolor{red!20}  \frac{\ell}{1-e^{-\gamma \ell}} -  \frac{\ell}{1-e^{\gamma \ell}} & \cellcolor{green!20}-   \frac{\ell}{1-e^{-\gamma \ell}}                                                                                                                                                                                                                                              \\
                                                                  & \cellcolor{yellow!20}  \frac{\ell}{1-e^{\gamma \ell}}                                 & \cellcolor{red!20}\ddots                                & \cellcolor{green!20}\ddots                                                                                                                                                                                                                 \\
                                                                  &                                                                                       & \cellcolor{yellow!20}\ddots                             & \cellcolor{red!20}\ddots    & \cellcolor{green!20}\ddots                                                                                                                                                                                   \\
                                                                  &                                                                                       &                                                         & \cellcolor{yellow!20}\ddots & \cellcolor{red!20}\ddots                              & \cellcolor{green!20}-   \frac{\ell}{1-e^{-\gamma \ell}}                                                                                              \\
                                                                  &                                                                                       &                                                         &                             & \cellcolor{yellow!20}  \frac{\ell}{1-e^{\gamma \ell}} & \cellcolor{red!20}  \frac{\ell}{1-e^{-\gamma \ell}} -  \frac{\ell}{1-e^{\gamma \ell}} & \cellcolor{green!20}-   \frac{\ell}{1-e^{-\gamma \ell}}^{-1} \\
                                                                  &                                                                                       &                                                         &                             &                                                       & \cellcolor{yellow!20}  \frac{\ell}{1-e^{\gamma \ell}}                                 & \cellcolor{red!20}-   \frac{\ell}{1-e^{\gamma \ell}}
        \end{pmatrix}.
        \endgroup
    \end{aligned}
\end{equation}

The following result from \cite{ammari.barandun.ea2024Mathematical}, that is the analogue of \cref{prop:capa approx hermitian}, describes how the gauge capacitance matrix characterises the nontrivial solutions to \eqref{eq:coupled ods2}.
\begin{proposition} \label{cor: approx via eva eve}
    The $N$ subwavelength eigenfrequencies $\omega_i$ satisfy, as $\delta\to0$,
    \begin{align*}
        \omega_i =  v_b \sqrt{\delta\lambda_i} + \BO(\delta),
    \end{align*}
    where $(\lambda_i)_{1\leq i\leq N}$ are the eigenvalues of the eigenvalue problem
    \begin{equation}
        \label{eq:eigevalue problem capacitance matrix2}
        V^2\capmatg\bm a_i = \lambda_i L \bm a_i,\qquad 1\leq i\leq N
    \end{equation}
    with $V$ and $L$ being as in \cref{prop:capa approx hermitian}. Furthermore, let $u_i$ be a subwavelength eigenmode corresponding to $\omega_i$ and let $\bm a_i$ be the corresponding eigenvector of $\capmatg$. Then,
    \begin{align*}
        u_i(x) = \sum_j \bm a_i^{(j)}V_j(x) + \BO(\delta),
    \end{align*}
    where $V_j$ is defined by \eqref{eq: def Her cap mat} and $\bm a_i^{(j)}$ denotes the $j$\textsuperscript{th} entry of the eigenvector $\bm a_i$.
\end{proposition}
\part{Applications}
\section{Topological interface modes}\label{sec: interface modes}
In this section, we consider wave localisation at the interface between two systems of \emph{finitely many} subwavelength resonators. These resonators are arranged in pairs or dimers, such that the model we consider shares many of the features of the Su-Schrieffer-Heeger (SSH) model in quantum mechanics \cite{su.schrieffer.ea1979Solitons,margetis.watson.ea2023Su}. We show the existence of exponentially localised interface eigenmodes in this finite structure. These interface modes have been subject to numerous studies in the setting of infinite structures. A particular focus has been put on studying the topological properties of infinite periodic structures and then introducing carefully designed interfaces so as to create the so-called \emph{topologically protected} eigenmodes \cite{ammari.davies.ea2020Topologically}. These modes have significant implications for applications since they are expected to be robust with respect to imperfections in the design. These concepts have been widely studied in a variety of settings, most notably in quantum mechanics for the Schrödinger operator \cite{fefferman.lee-thorp.ea2014Topologically,fefferman.lee-thorp.ea2017Topologically} and more recently, for related continuous models of classical wave systems in \cite{lin.zhang2022Mathematical,thiang.zhang2023Bulkinterface,qiu.lin.ea2023Mathematical,craster.davies2023Asymptotic,thiang2023Topological,davies.lou2023Landscape}. In finite-sized systems, these eigenmodes have been observed both experimentally and numerically; see, for instance, \cite{ammari.barandun.ea2023Edge,ammari.davies.ea2020Topologically,coutant.achilleos.ea2022Subwavelength,zheng.achilleos.ea2019Observation} and the references therein.

As mentioned before, here we consider the far less-explored but more realistic physical setting of interface modes in finite dimer structures. We present a one-to-one correspondence between the position of the eigenfrequency in the spectrum of the corresponding infinite periodic structure (\emph{i.e.}, in the asymptotic spectral bulk, its boundary or the asymptotic spectral gap; see \Cref{def:spectralgap}) and the behaviour (localised versus delocalised) of the corresponding eigenmode. Our results hold for any finite and large enough chain of dimers with a geometric defect satisfying a mild condition. Furthermore, we show that the eigenfrequencies lying in the gap converge exponentially as the size of the structure goes to infinity and provide an explicit and simple formula for the limit. The results of this section are from \cite{ammari.barandun.ea2023Exponentially}.

\subsection{Dimer chains with a geometric defect}

Systems of repeated dimers (that is, $s_i=s_{i-2}$ for $3\leq i\leq N$) are of particular interest as the corresponding infinite structure can be studied with Floquet--Bloch band theory \cite{ammari.barandun.ea2023Edge}; when the two repeating separation distances are distinct, this  provides a nontrivial example of a band gap between the subwavelength spectral bands. Here, we consider systems of dimers with a defect in the geometric structure, so that at some point the repeating pattern of alternating separation distances is broken. This is inspired by the famous Su-Schrieffer-Heeger (SSH) model from quantum settings and is the canonical example of a topologically protected interface mode \cite{su.schrieffer.ea1979Solitons}. It is a system of $N=4m+1$ for $m\in\N$ resonators such that
\begin{align*}
    s_i & = s_{i-2} \quad \text{for } 3\leq i \leq 2m,      \\
    s_i & = s_{i+2} \quad \text{for } 2m+1\leq i \leq 4m-2,
\end{align*}
where we typically assume $s_1=s_{2m+2}$ and $s_2=s_{2m+1}$ so that the system is symmetric with respect to the center of the $(2m+1)$\textsuperscript{th} resonator with spacings $s_1$ and $s_2$. A sketch of this system with a geometric defect is shown in \cref{fig: geometrical defect}.

\begin{figure}[h]
    \centering
    \begin{adjustbox}{width=\textwidth}
        \begin{tikzpicture}
            \draw[-,thick,dotted] (-.5,-1) -- (-.5,2);
            \draw[|-|,dashed] (0,1) -- (1,1);
            \node[above] at (0.5,1) {$s_2$};
            \draw[ultra thick] (1,0) -- (2,0);
            \node[below] at (1.5,0) {$D_{2m+2}$};
            \draw[-,dotted] (1,0) -- (1,1);

            \draw[|-|,dashed] (2,1) -- (3,1);
            \node[above] at (2.5,1) {$s_1$};
            \draw[ultra thick] (3,0) -- (4,0);
            \node[below] at (3.5,0) {$D_{2m+3}$};
            \draw[-,dotted] (2,0) -- (2,1);
            \draw[-,dotted] (3,0) -- (3,1);
            \draw[-,dotted] (4,0) -- (4,1);

            \begin{scope}[shift={(+4,0)}]
                \draw[|-|,dashed] (0,1) -- (1,1);
                \node[above] at (0.5,1) {$s_2$};
                \draw[ultra thick] (1,0) -- (2,0);
                \node[below] at (1.5,0) {$D_{2m+4}$};
                \draw[-,dotted] (1,0) -- (1,1);
                \node at (2.5,.5) {\dots};
            \end{scope}

            \begin{scope}[shift={(+6,0)}]
                \draw[ultra thick] (1,0) -- (2,0);
                \node[below] at (1.5,0) {$D_{4m-2}$};

                \draw[|-|,dashed] (2,1) -- (3,1);
                \node[above] at (2.5,1) {$s_2$};
                \draw[ultra thick] (3,0) -- (4,0);
                \node[below] at (3.5,0) {$D_{4m}$};
                \draw[-,dotted] (2,0) -- (2,1);
                \draw[-,dotted] (3,0) -- (3,1);
                \draw[-,dotted] (4,0) -- (4,1);
                \begin{scope}[shift={(+2,0)}]
                    \draw[|-|,dashed] (2,1) -- (3,1);
                    \node[above] at (2.5,1) {$s_1$};
                    \draw[ultra thick] (3,0) -- (4,0);
                    \node[below] at (3.5,0) {$D_{4m+1}$};
                    \draw[-,dotted] (2,0) -- (2,1);
                    \draw[-,dotted] (3,0) -- (3,1);
                \end{scope}
            \end{scope}

            \begin{scope}[shift={(-4,0)}]

                \draw[|-|,dashed] (0,1) -- (1,1);
                \node[above] at (0.5,1) {$s_1$};
                \draw[ultra thick] (1,0) -- (2,0);
                \node[below] at (1.5,0) {$D_{2m}$};
                \draw[-,dotted] (1,0) -- (1,1);

                \draw[|-|,dashed] (2,1) -- (3,1);
                \node[above] at (2.5,1) {$s_2$};
                \draw[ultra thick] (3,0) -- (4,0);
                \node[below,fill=white] at (3.5,0) {$D_{2m+1}$};
                \draw[-,dotted] (2,0) -- (2,1);
                \draw[-,dotted] (3,0) -- (3,1);
                \draw[-,dotted] (4,0) -- (4,1);
            \end{scope}

            \begin{scope}[shift={(-8,0)}]
                \node at (.5,.5) {\dots};
                \draw[ultra thick] (1,0) -- (2,0);
                \node[below] at (1.5,0) {$D_{2m-2}$};

                \draw[|-|,dashed] (2,1) -- (3,1);
                \node[above] at (2.5,1) {$s_2$};
                \draw[ultra thick] (3,0) -- (4,0);
                \node[below] at (3.5,0) {$D_{2m-3}$};
                \draw[-,dotted] (2,0) -- (2,1);
                \draw[-,dotted] (3,0) -- (3,1);
                \draw[-,dotted] (4,0) -- (4,1);
            \end{scope}

            \begin{scope}[shift={(-14,0)}]
                \draw[ultra thick] (1,0) -- (2,0);
                \node[below] at (1.5,0) {$D_{1}$};

                \draw[|-|,dashed] (2,1) -- (3,1);
                \node[above] at (2.5,1) {$s_1$};
                \draw[ultra thick] (3,0) -- (4,0);
                \node[below] at (3.5,0) {$D_{2}$};
                \draw[-,dotted] (2,0) -- (2,1);
                \draw[-,dotted] (3,0) -- (3,1);
                \draw[-,dotted] (4,0) -- (4,1);
                \begin{scope}[shift={(+2,0)}]
                    \draw[|-|,dashed] (2,1) -- (3,1);
                    \node[above] at (2.5,1) {$s_2$};
                    \draw[ultra thick] (3,0) -- (4,0);
                    \node[below] at (3.5,0) {$D_{3}$};
                    \draw[-,dotted] (2,0) -- (2,1);
                    \draw[-,dotted] (3,0) -- (3,1);
                \end{scope}
            \end{scope}

        \end{tikzpicture}
    \end{adjustbox}
    \caption{Dimer structure with a geometric defect.}
    \label{fig: geometrical defect}
\end{figure}

For such systems the geometric symmetries mean that the capacitance matrix from \eqref{equ:capacitance matrix hermitian} has the following tridiagonal block structure:

\setcounter{MaxMatrixCols}{20}
\begin{align}
    \label{eq: strucutre capacitance matrix}
    C = \begin{pNiceMatrix}
                      \Block[draw,fill=blue!40,rounded-corners]{7-7}{} \tilde{\alpha} & \beta_{1} &           &           &           &           &                                                     &           &           &           &                                        \\
                      \beta_{1}                                                       & \alpha    & \beta_{2} &           &           &           &                                                     &           &           &           &                                        \\
                                                                                      & \beta_{2} & \alpha    & \beta_{1} &           &           &                                                     &           &           &           &                                        \\
                                                                                      &           & \ddots    & \ddots    & \ddots    &           &                                                     &           &           &           &                                        \\
                                                                                      &           &           & \beta_{2} & \alpha    & \beta_{1} &                                                     &           &           &           &                                        \\
                                                                                      &           &           &           & \beta_{1} & \alpha    & \beta_{2}                                           &           &           &           &                                        \\
                                                                                      &           &           &           &           & \beta_{2} & \Block[draw,fill=red!40,rounded-corners]{7-7}{}\eta & \beta_{2} &           &           &                                        \\
                                                                                      &           &           &           &           &           & \beta_{2}                                           & \alpha    & \beta_{1} &           &                                        \\
                                                                                      &           &           &           &           &           &                                                     & \beta_{1} & \alpha    & \beta_{2} &                                        \\
                                                                                      &           &           &           &           &           &                                                     &           & \ddots    & \ddots    & \ddots    &                            \\
                                                                                      &           &           &           &           &           &                                                     &           &           & \beta_{1} & \alpha    & \beta_{2}                  \\
                                                                                      &           &           &           &           &           &                                                     &           &           &           & \beta_{2} & \alpha    & \beta_{1}      \\
                                                                                      &           &           &           &           &           &                                                     &           &           &           &           & \beta_{1} & \tilde{\alpha}
                  \end{pNiceMatrix},
\end{align}
where
\begin{align}
    \label{eq: translation alpha beta to s1 s2}
    \beta_1 = -s_1\inv,\quad \beta_2 = -s_2\inv,\quad \alpha = s_1\inv + s_2\inv,\quad \eta=2s_2\inv,\quad \tilde{\alpha} = s_1^{-1}.
\end{align}
It will be useful to have notation for the top left $(2m+1)\times(2m+1)$-submatrix $C_1
$ and the bottom right $(2m+1)\times(2m+1)$-submatrix $C_2
$ (as highlighted by the shading in \eqref{eq: strucutre capacitance matrix}).
\cref{fig: eva eve} shows the spectrum and the eigenvectors of the capacitance matrix \eqref{eq: strucutre capacitance matrix}. It is immediately clear that there exists a spectral gap and there exists a localised eigenmode associated to an eigenvalue in the gap.

\begin{figure}[h]
    \centering
    \begin{subfigure}[t]{0.48\textwidth}
        \centering
        \includegraphics[height=0.76\textwidth]{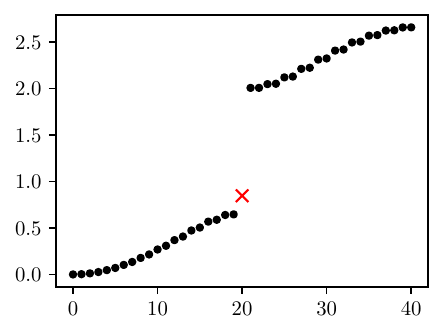}
        \caption{Eigenvalues of \eqref{eq: strucutre capacitance matrix}. As red cross a specific eigenvalue lying isolated from the others.}
        \label{fig: eva}
    \end{subfigure}\hfill
    \begin{subfigure}[t]{0.48\textwidth}
        \centering
        \includegraphics[height=0.76\textwidth]{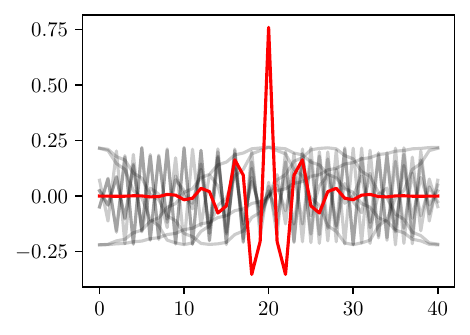}
        \caption{A selection of eigenvectors of $C$ from \eqref{eq: strucutre capacitance matrix}. All eigenvectors are superimposed and with unit norm. The red solid eigenvector corresponds to the red cross eigenvalue in \cref{fig: eva}.}
        \label{fig: eva eve}
    \end{subfigure}

    \caption{Eigenvalues and eigenvectors of the capacitance matrix \eqref{eq: strucutre capacitance matrix} for $N=41, s_1=1$ and $s_2=3$.}
    \label{fig: eigenvector behaviour}
\end{figure}

\subsection{Perturbed tridiagonal $2$-Toeplitz matrices}  \label{sect3}
In this section, we will briefly recall results established in Section \ref{sec: spectrum of tridiag 2 toeplitz matrices with perturbations} about eigenvalues and eigenvectors of tridiagonal $2$-Toeplitz matrices with perturbations on the corners. By slightly abusing the notation, we denote by $A_{2m+1}^{(a, b)}$ in (\ref{equ:oddmatrixtwoperturb1}) and $A_{2m}^{(a, b)}$ in (\ref{equ:evenmatrixtwoperturb1}) the corresponding Hermitian tridiagonal $2$-Toeplitz matrices with $\alpha_1=\alpha_2=\alpha\in \mathbb R, \eta_j=\beta_j\in \mathbb R, j=1,2$.

The eigenvalues of $A_{2m+1}^{(a, b)}$ and $A_{2m}^{(a, b)}$ are given by the roots of the characteristic polynomials \eqref{equ:eigenpolynomial5} and \eqref{equ:eigenpolynomial6}. On the other hand, by Theorem \ref{thm:eigenvectoroddmatrixtwoperturb2}, the eigenvectors $\bm x$ of $A_{2m+1}^{(a, b)}$ and $A_{2m}^{(a, b)}$ are given as follows.
\begin{proposition}\label{thm: eigenvectors of A2k+1 and A2k}
    Let $\lambda$ be an eigenvalue of $A_{2m+1}^{(a, b)}$ in (\ref{equ:oddmatrixtwoperturb1}) with $\alpha_1=\alpha_2=\alpha\in \mathbb R,  \eta_j=\beta_j\in \mathbb R, j=1,2$. Then, using $\mu\coloneqq y(\lambda)$ where
    \begin{equation}\label{eq: def y}
        y(\lambda) := \frac{\left(\alpha-\lambda\right)^2-(\beta_1^2+\beta_2^2)}{2\beta_1\beta_2},
    \end{equation}
    the eigenvector corresponding to $\lambda$ is given by
    \begin{align}\label{eq: eigenvector A2k+1}
        \bm x = & \left( \hat q_{0}^{(\xi_p, \xi_{q})}\left(\mu\right),-\frac{1}{\beta_1}\left(\alpha-\lambda\right) \hat p_0^{(\xi_p, \xi_{q})}\left(\mu\right), \hat  q_1^{(\xi_p, \xi_{q})}\left(\mu\right), -\frac{1}{\beta_1} \left(\alpha-\lambda\right) \hat p_1^{(\xi_p, \xi_{q})}\left(\mu\right), \right. \nonumber \\
                & \qquad  \left.  \ldots, -\frac{1}{\beta_1} \left(\alpha-\lambda\right) \hat p_{m-1}^{(\xi_p, \xi_{q})}\left(\mu\right),
        \hat  q_{m}^{(\xi_p, \xi_{q})}\left(\mu\right)\right ).
    \end{align}
    If $\lambda$ is an eigenvalue of $A_{2m}^{(a, b)}$, then the corresponding eigenvector is given by
    \begin{align}\label{eq: eigenvector A2k}
        \bm x = & \left( \hat q_{0}^{(\xi_p, \xi_{q})}\left(\mu\right),-\frac{1}{\beta_1}\left(\alpha-\lambda\right) \hat p_0^{(\xi_p, \xi_{q})}\left(\mu\right), \hat  q_1^{(\xi_p, \xi_{q})}\left(\mu\right), -\frac{1}{\beta_1} \left(\alpha-\lambda\right) \hat p_1^{(\xi_p, \xi_{q})}\left(\mu\right), \right. \nonumber \\
                & \qquad  \left.  \ldots, -\frac{1}{\beta_1} \left(\alpha-\lambda\right) \hat p_{m-1}^{(\xi_p, \xi_{q})}\left(\mu\right)\right ).
    \end{align}
    In both cases, $\xi_{q}=(\alpha-\lambda), \xi_{p}= (\alpha+a-\lambda)$ and $\widehat p_j^{(\xi_p, \xi_{q})}, \widehat q_j^{(\xi_p, \xi_{q})}$'s are given by \eqref{equ:normalizepolynomial1} for $\beta = \beta_2/\beta_1$.
\end{proposition}

Observing that the two blocks in the capacitance matrix $C$ defined by \eqref{eq: strucutre capacitance matrix} are two Hermitian tridiagonal $2$-Toeplitz matrices with perturbations on diagonal corners and based on the theorem above, we can characterise the eigenvectors of $C$ as follows.

\begin{proposition}\label{thm: eigenvectors of defect C}
    Let $(\lambda,\bm v)$ be an eigenpair of $C$ defined by \eqref{eq: strucutre capacitance matrix}. Then $\bm v$ is given by
    \begin{equation}\label{eq: structure eigenvector defect matrix}
        \bm v = (\bm x^{(1)},\bm x^{(2)},\dots,\bm x^{(2m)},\bm x^{(2m+1)},(-1)^\sigma\bm x^{(2m)},\dots,(-1)^\sigma\bm x^{(2)},(-1)^\sigma\bm x^{(1)})^{\top},
    \end{equation}
    where $\bm x = (\bm x^{(1)},\bm x^{(2)},\dots,\bm x^{(2m)},\bm x^{(2m+1)})^\top \in \R^{2m+1}$ is as in \eqref{eq: eigenvector A2k+1} with $\xi_{q}=(\alpha-\lambda), \xi_{p}= (\alpha+a-\lambda)$ and $\sigma =0$ or $1$ for all $\bm x$ except for $\bm x \in \text{span} \{\bm 1\}$ where $\sigma=0$.
\end{proposition}

\subsection{Asymptotic spectral gap and localised interface modes}  \label{sect4}
In this section, we show the existence of a spectral gap for a defectless structure of dimers. Furthermore, we establish a direct relationship between an eigenvalue being within the spectral gap and the localisation of its corresponding eigenvector.

\bigskip
We first define a spectral gap for the capacitance matrix of an unperturbed structure of dimers.
\begin{definition}[Spectral bulk and gaps] \label{def:spectralgap}
    Consider a finite structure of resonators. We define the \emph{asymptotic spectral bulk} $\Sigma$ and \emph{asymptotic spectral gap} $\Gamma$ of the structure as the spectral bulk and spectral gap (also known as band gap) of the associated infinite periodic system, respectively.
\end{definition}
The spectral gap and spectral bulk of infinite periodic dimer systems have been computed in \cite[Lemma 5.3]{ammari.barandun.ea2023Edge}.
\begin{proposition}
    Consider a system of repeated dimers (without defect) with $N=2m$ resonators. Denote by $C_{N}$ the associated capacitance matrix and let $\Sigma$ be the asymptotic spectral bulk. Then,
    \begin{align*}
        \Sigma = \overline{\lim_{N\to\infty} \sigma(C_{N})} = \left[0,\frac{2}{s_2}\right] \cup \left[\frac{2}{s_1}, \frac{2}{s_1} + \frac{2}{s_2}\right],
    \end{align*}
    where $\lim$ denotes the Hausdorff limit. Consequently, the asymptotic spectral gap is
    \begin{align*}
        \Gamma = \left( \frac{2}{s_2}, \frac{2}{s_1} \right)\subset \R.
    \end{align*}
\end{proposition}

 In view of \cref{prop:capa approx hermitian}, we can refer to spectral bulk and spectral gap of the physical system and of the capacitance matrix interchangeably.

Last, we show that an eigenvalue in the gap is associated with an eigenvector that is exponentially localised in the sense that it decays exponentially away from the interface in both directions, within the limits of the finite structure.

\begin{definition}[Localised interface mode]
    Let $v(x)$ be an eigenmode. Then we say that $v$ is a \emph{localised interface mode} at the point $x_0$, if both $\vert v(x-x_0)\vert $ for $x_0<x\in D$ and $\vert v(x_0-x)\vert $  for $x_0>x\in D$ decay exponentially as a function of $x\in D$. The same terminology applies to the corresponding eigenvector of the capacitance matrix.
\end{definition}

\begin{proposition}[Eigenvectors of ${C}$]\label{prop: exponential decay and sines}
    Let $C\in\mathbb{R}^{4m+1 \times 4m+1}$ be the capacitance matrix of a defected structure as illustrated in \cref{fig: geometrical defect} and let $(\lambda,v)$ be an eigenpair of $C$. Then, there exists $\vert r\vert \geq1$ independent of $m$ and $A,B,\tilde A,\tilde B\in\R$ dependent on $m$ such that, for $y(\lambda)$ defined by (\ref{eq: def y}),  
    \begin{description}
        \item[if $y(\lambda)^2>1$]
            \begin{align*}
                v^{ (\vert 2m-2j\vert) } = Ar^j+Br^{-j}, \\
                v^{ (\vert 2m-2j-1\vert) } = \tilde A r^j+\tilde B r^{-j},
            \end{align*}
            with $A=\frac{r^{1-m} (c_1 r-c_2)}{r^2-1}=\mathcal{O}(\frac{1}{r^{m}})$ and $B = \frac{r^m (c_2 r-c_1)}{r^2-1}= \mathcal{O}(r^{m-1})$ as $m\to\infty$ for $c_1,c_2\in\R$ independent of $m$. The same asymptotics (with a slight different formula) hold for $\tilde A$ and $\tilde B$;
        \item[if $y(\lambda)^2<1$]
            \begin{align*}
                v^{ (\vert 2m-2j\vert) } = A\cos(j\theta)+B\sin(j\theta), \\
                v^{ (\vert 2m-2j-1\vert) } = \tilde A\cos(j\theta)+\tilde B\sin(j\theta),
            \end{align*}
            with $r=e^{\mathbf{i} \theta}$ and $A,B,\tilde A,\tilde B$ bounded as $m\to\infty$;
        \item[if $y(\lambda)^2=1$] $r=\pm 1$ and
            \begin{align*}
                v^{ (\vert 2m-2j\vert) } = Ar^j+Br^j\cdot j, \\
                v^{ (\vert 2m-2j-1\vert) } = \tilde A r^j+\tilde B r^j\cdot j,
            \end{align*}
            with
            $A=\frac{r^{1-m} (c_1 m r-c_1 r-c_2 m)}{m r^2-m-r^2}$ and $B = \frac{r^m (c_2 r-c_1)}{m r^2-m-r^2}$ as $m\to\infty$ for $c_1,c_2\in\R$ independent of $m$. The same asymptotics (with a slight different formula) hold for $\tilde A$ and $\tilde B$.
    \end{description}
\end{proposition}

\subsection{Existence, uniqueness, and convergence of the eigenvalue in the gap}  \label{sect5}
In this section, we will show the existence of a unique eigenvalue in the gap for the defected structure and consequently the existence of a unique localised interface eigenvector. Furthermore, we analyse the behaviour as the size of the system grows, specifically the limiting behaviour as $N\to\infty$. We will show that the eigenvalue of $C$ lying in the asymptotic spectral gap converges exponentially fast to a value in the gap.

\subsubsection{Existence}
We first show the existence of an eigenvalue in the band gap. To illustrate how the existence of the interface eigenvalue is related to the Chebyshev polynomial, we outline the main steps of the proof. For further details on the derivation, we refer the reader to \cite[Section 5]{ammari.barandun.ea2023Exponentially}.

By performing Laplace expansion on the top block of rows $1\leq i\leq 2m+1$ of $C-xI $, we obtain the characteristic polynomial of $C\in\R^{(4m+1)\times (4m+1)}$, that is, 
\begin{align}
    p(x) & = (\left[\left(x-\alpha-\beta_2-(\beta_1 - \beta_2)\right) P_m^*\left(x\right)+\left( \beta_2(\beta_1 - \beta_2)(x-\alpha)  -\beta_2 \beta_{1}^2-(\beta_1 - \beta_2) \beta_{2}^2\right) P_{m-1}^*\left(x\right)\right] \nonumber \\
         & - \beta_2^2\left[\left(x-\alpha-\beta_2\right) P_{m-1}^*\left(x\right)+\left( -\beta_2 \beta_{1}^2\right) P_{m-2}^*\left(x\right)\right])\chi_{A_{2m}^{(a,0)}(\alpha,\beta_1,\beta_2)}(x),
\end{align}
where we have used (\ref{equ:eigenpolynomial5}) and (\ref{equ:eigenpolynomial6}).

For the sake of brevity, we rewrite
\begin{align*}
    p(x) & = \chi_{A_{2m}^{(a,0)}(\alpha,\beta_1,\beta_2)}(x)\left(\left[A(x) P_m^*\left(x\right)+B(x) P_{m-1}^*\left(x\right)\right]
    -\beta_2^2
    \left[E(x) P_{m-1}^*\left(x\right)+F(x) P_{m-2}^*\left(x\right)\right]\right),
\end{align*}
where $A(x),B(x),E(x)$, and $F(x)$ are defined as
\begin{align*}
     & A(x) = x-\alpha-\beta_1,   \quad B(x) = \beta_2(\beta_1 - \beta_2)(x-\alpha)  -\beta_2 \beta_{1}^2-(\beta_1 - \beta_2) \beta_{2}^2, \\
     & E(x) = x-\alpha-\beta_2,  \quad F(x) = -\beta_2 \beta_{1}^2.
\end{align*}
For simplicity, we abbreviate $A(x), B(x), E(x)$, and $F(x)$ as $A, B, D$, and $F$ respectively in the subsequent discussions. Thus, $\lambda\in\Gamma$ is an eigenvalue if and only if
\begin{align}
    \left[A P_m^*\left(\lambda\right)+B P_{m-1}^*\left(\lambda\right)\right]  
     & = \beta_2^2                                                               
    \left[E P_{m-1}^*\left(\lambda\right)+F P_{m-2}^*\left(\lambda\right)\right] \nonumber\\
    \Leftrightarrow P_m^*\left(\lambda\right) \left[A+B\frac{P_{m-1}^*\left(\lambda\right)}{P_{m}^*\left(\lambda\right)}\right]
     & = \beta_2^2                                                              
    P_{m-1}^*\left(\lambda\right) \left[E+F\frac{P_{m-2}^*\left(\lambda\right)}{P_{m-1}^*\left(\lambda\right)}\right].
    \label{eq: lambda in gamma eigenvalues only if}
\end{align}
In the above step, we were able to divide by $\chi_{A_{2m}^{(a,0)}}(\lambda)$ because it is nonzero in $\Gamma$ as will be shown by \cref{prop: exaclty one in gap}. In addition, $P_k^*(\lambda)=\left(\beta_1 \beta_{2}\right)^k U_k\left(y(\lambda)\right)$ with $y(\lambda) := \frac{(\lambda-\alpha)^2-\beta_1^2-\beta_2^2}{2 \beta_{1} \beta_2}$ and $\lambda \in \Gamma$ corresponds to the case when $y(\lambda)<-1$. Since the Chebyshev polynomials $U_k(\cdot)$'s only have roots in $(-1,1)$, $P_{m-1}^*\left(\lambda\right)$ and $P_{m}^*\left(\lambda\right)$ are nonzero in $\Gamma$ and we are able to divide by them in \eqref{eq: lambda in gamma eigenvalues only if}.

Moreover, since
\[
    U_m(y)=\frac{\left(y+\sqrt{y^2-1}\right)^{m+1}-\left(y-\sqrt{y^2-1}\right)^{m+1}}{2 \sqrt{y^2-1}},
\]
the limit  $L(\lambda)\coloneqq \lim_{m\to\infty}\frac{P_{m-1}^*\left(\lambda\right)}{P_{m}^*\left(\lambda\right)}$ exists for all $\lambda\in \mathbb R$. Then,
\begin{align*}
    \left[A+BL\right]= \beta_2^2L
    \left[E+FL\right] 
    \Leftrightarrow L^2(\beta_2^2F) + L(\beta_2^2E-B) - A = 0,
\end{align*}
from which we get the condition
\begin{align}\label{eq: condition L is root}
    L(\lambda) & = \frac{B - E \beta_{2}^{2} \pm \sqrt{4 A F \beta_{2}^{2} + B^{2} - 2 B E \beta_{2}^{2} + E^{2} \beta_{2}^{4}}}{2 F \beta_{2}^{2}}.
\end{align}
On the other hand, by the recurrence formula of $U_m$, for $\lambda \in \Gamma$ with $y(\lambda)<-1$ and using again the shorthand $y$ for $y(\lambda)$, it follows that
\begin{align}
    L(\lambda) =& \lim_{m\to\infty}\frac{P_{m-1}^*\left(\lambda \right)}{P_{m}^*\left(\lambda\right)} 
    = \frac{y-\sqrt{y^2-1}}{\beta_1\beta_2}\nonumber \\
    =&\frac{(\lambda-\alpha)^2-\beta_1^2-\beta_2^2-2 \beta_1 \beta_2 \sqrt{\frac{\left(\beta_1^2+\beta_2^2-(\lambda-\alpha)^2\right)^2}{4 \beta_1^2
                \beta_2^2}-1}}{2 \beta_1^2 \beta_2^2}. \label{eq: condition L is limit}
\end{align}
 Now, an algebraic manipulation shows that conditions \eqref{eq: condition L is root} and \eqref{eq: condition L is limit} have exactly one common solution $\lambda_0 \in \Gamma$  given by
\begin{align}\label{eq: root of L}
    \lambda_0 = \alpha+ \frac{1}{2} \left(-\sqrt{9 \beta_1^2-14 \beta_1 \beta_2+9 \beta_2^2}-\beta_1-\beta_2\right).
\end{align}

We can then state the results of the existence of the interface eigenvector (modes) of ${C}$. 
\begin{proposition}\label{prop:existencedefectfrequency}
    Consider a perturbed structure of dimers as illustrated in \cref{fig: geometrical defect}. For $N$ large enough, there exists at least one localised interface eigenvector of $C$ with eigenvalue $\lambda_{\mathsf{i}}^{(N)}$ in the band gap $\Gamma$.
\end{proposition}
\begin{proof}
    Denote by
    \begin{align*}
        f_m(\lambda) \coloneqq & \left[A+B\frac{P_{m-1}^*\left(\lambda\right)}{P_{m}^*\left(\lambda\right)}\right] - \beta_2^2\frac{P_{m-1}^*\left(\lambda\right)}{P_m^*\left(\lambda\right)}
        \left[E+F\frac{P_{m-2}^*\left(\lambda\right)}{P_{m-1}^*\left(\lambda\right)}\right],                                                                                                  \\
        \nm
        f_{\infty} \coloneqq   & \left[A+BL(\lambda)\right] - \beta_2^2L(\lambda)
        \left[E+FL(\lambda)\right],
    \end{align*}
    and note that $f_\infty(\lambda)=\lim_{m\to\infty}f_m(\lambda)$. By (\ref{eq: root of L}), there exists a $\lambda_0$ in the band gap such that $f_\infty(\lambda_0)=0$. Furthermore, by the above formula of $f_\infty(\lambda)$, we have $f_{\infty}(\lambda_0-\zeta)f_{\infty}(\lambda_0+\zeta)<0$ for some $\zeta>0$ satisfying $[\lambda_0-\zeta, \lambda_0+\zeta]\subset \Gamma$. By the sign-preserving property, we have $f_{m}(\lambda_0-\zeta)f_{m}(\lambda_0+\zeta)<0$ for large enough $m$, which proves the existence of roots of $f_{m}(\lambda)$ in the band gap $\Gamma$.  By \cref{prop: exponential decay and sines}, the corresponding eigenvector is then a localised interface eigenvector.
\end{proof}

\subsubsection{Uniqueness}
We now show the uniqueness of the eigenvalue in the band gap. 
The proof utilises the Cauchy interlacing theorem \cite{cauchy} and 
the  monotonicity of Chebyshev polynomials of the second kind in 
Proposition \ref{lemma: monotonicity_Chebyshev}. We refer the reader to \cite[Section 5]{ammari.barandun.ea2023Exponentially} for further details.

\begin{proposition}\label{prop: exaclty one in gap}
    There exists at most one eigenvalue of $C$ as defined in \eqref{eq: strucutre capacitance matrix} lying in the asymptotic spectral gap $\Gamma = (2/s_2,2/s_1)$. In particular, for $m$ large enough, there exists exactly one eigenvalue in $\Gamma$.
\end{proposition}

\begin{remark}
    We remark that if one can show the existence of an eigenfrequency in the band gap $\Gamma$ for the capacitance matrix $C$ with general size $m$, then by \cref{prop: exaclty one in gap}, it is unique.
\end{remark}

\subsubsection{Convergence}
Last, we show the exponential convergence of the resonant frequency in the gap and conclude the section by the following theorem.

\begin{theorem}\label{thm:existenceofeigenfrquency}
    Consider a perturbed structure of dimers as illustrated in \cref{fig: geometrical defect}. For $N$ large enough, there exists a unique interface mode with eigenfrequency $\omega_{\mathsf{i}}^{(N)}$ in the band gap. The associated eigenfrequency $\omega_{\mathsf{i}}^{(N)}$ converges to
    \begin{align}
        \omega_{\mathsf{i}} = v_b\sqrt{\delta \frac{1}{2} \left(-\sqrt{\frac{9}{s_1^2}- \frac{14}{s_1s_2} + \frac{9}{s_2^2}}+\frac{3}{s_1}+\frac{3}{s_2}\right)}
    \end{align}
    exponentially as $N\to\infty$. In particular, for $N$ large enough,
    \begin{align}\label{eq: error estimate convergence frequency in gap}
        \vert \omega_{\mathsf{i}} - \omega_{\mathsf{i}}^{(N)}\vert
        < Ae^{-BN},
    \end{align}
    for some $A,B>0$ independent of $N$.
\end{theorem}
\begin{proof}
    The limit $\lambda_i$ can be computed from (\ref{eq: root of L}) and $\omega_i$ is $v_b\sqrt{\delta \lambda_i}$ by Corollary \ref{cor: approx via eva eve}. The only part left to prove is the convergence rate, which is derived from the pseudospectrum theory. We refer the reader to \cite[Section 5]{ammari.barandun.ea2023Exponentially} for further details.
\end{proof}

We remark that, combined with Proposition \ref{prop: exponential decay and sines}, Theorem \ref{thm:existenceofeigenfrquency} also gives the decaying rate of the interface mode for a structure with sufficiently many resonators.

\begin{figure}[h]
    \centering
    \includegraphics[width=0.45\textwidth]{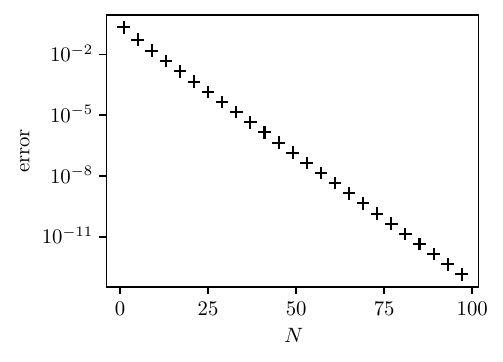}
    \caption{Convergence of the eigenvalue in the gap ($y$-axis in \emph{log} scale). We display the left-hand side of \eqref{eq: error estimate convergence frequency in gap} for a structure with $s_1=1$ and $s_2=2$.
    }
    \label{fig: convergence of error}
\end{figure}

\subsection{Stability analysis} \label{sect6}
Interface modes of SSH-like structures are well-known to be stable, \emph{i.e.}, perturbations of the system affect them only slightly. In this section, we show that perturbations in the geometry have limited effect on both the resonant frequencies and the associated modes. Then we quantify this effect.

To this end, we consider a system of $N=4m+1$ resonators as shown in \cref{fig: geometrical defect} but where the spacings $s_i$ are now perturbed:
\begin{align}
    \label{eq: perturbation si-s}
    s_i = \begin{dcases}
              s_1 + \tilde\epsilon_i, & 1\leq i \leq 2m, \text{ $i$ odd},     \\
              \nm
              s_2 + \tilde\epsilon_i, & 1\leq i \leq 2m, \text{ $i$ even},    \\
              \nm
              s_1 + \tilde\epsilon_i, & 2m+1\leq i \leq 4m, \text{ $i$ even}, \\
              \nm
              s_2 + \tilde\epsilon_i, & 2m+1\leq i \leq 4m, \text{ $i$ odd}.
          \end{dcases}
\end{align}
Furthermore, we denote
\begin{align}
    \label{eq: eps_i}
    \epsilon_i = \begin{dcases}
                     -\frac{\tilde\epsilon_i}{s_1(s_1+\tilde\epsilon_i)}, & 1\leq i \leq 2m, \text{ $i$ odd},     \\
                     -\frac{\tilde\epsilon_i}{s_2(s_2+\tilde\epsilon_i)}, & 1\leq i \leq 2m, \text{ $i$ even},    \\
                     -\frac{\tilde\epsilon_i}{s_1(s_1+\tilde\epsilon_i)}, & 2m+1\leq i \leq 4m, \text{ $i$ even}, \\
                     -\frac{\tilde\epsilon_i}{s_2(s_2+\tilde\epsilon_i)}, & 2m+1\leq i \leq 4m, \text{ $i$ odd}.
                 \end{dcases}
\end{align}

The following proposition handles stability of the eigenvalues and is a direct application of the well-known Weyl theorem.
\begin{proposition}\label{thm:eigenvaluestability1}
    Let $\hat{C}$ be the capacitance matrix associated to the structure described in \eqref{eq: perturbation si-s} and let
    \begin{align}\label{equ:errorbound1}
        \epsilon \coloneqq \max_{1\leq i\leq N-2}\vert \epsilon_i\vert+\vert\epsilon_{i+1}\vert.
    \end{align}
    Then, the eigenvalues $\hat{\lambda}_k$ (sorted increasingly) satisfy
    \begin{align}
        \label{eq: bound eigenvalues error}
        \vert \hat{\lambda}_k - \lambda_k\vert \leq 2\epsilon, \quad 1\leq k\leq N,
    \end{align}
    where $\lambda_k$ are the eigenvalues of $C$.
\end{proposition}

Applying \cref{thm:eigenvaluestability1} to our system of dimers where the perturbations $\vert\tilde\epsilon_i\vert\leq \eta$ are in the interval $(-\eta,\eta)$ for some $\eta>0$, we obtain that the eigenvalue perturbation is bounded by
\begin{align*}
    \frac{2\eta}{s_1(s_1-\eta)}+\frac{2\eta}{s_2(s_2-\eta)}=2\eta\left(\frac{1}{s_1^2}+\frac{1}{s_2^2}\right) + \mathcal{O}(\eta^2)\quad \text{as }\eta\to 0.
\end{align*}

Now, we analyse the stability of the eigenvectors. As a direct consequence of Davis--Kahan theorem for Hermitian matrices \cite{davis}, we have the following result on the stability of the interface eigenmodes.
\begin{theorem}\label{thm: stability interface modes}
    Let $\epsilon< \frac{1}{2}\left(\frac{1}{s_1}-\frac{1}{s_2}\right)$ in (\ref{equ:errorbound1}). Let $\bm v$ and $\hat{\bm v}$ be the eigenvectors corresponding to the eigenvalues $\lambda_i$ and $\hat{\lambda}_i$ in the gap of $C$ and $\hat{C}$, respectively. Then,
    \begin{align}
        \label{eq: stabiltiy v gap}
        \Vert \bm v - \hat{\bm v}\Vert_2 & \leq \frac{2\sqrt{2}\epsilon}{\updelta} \\&\leq \frac{2\sqrt{2}\epsilon}{\updelta_0-2\epsilon},\label{eq: stabiltiy v gap apriori}
    \end{align}
    where $\updelta\coloneqq\min\{\vert\lambda_i-\hat{\lambda}_{i+1}\vert,\vert\lambda_i-\hat{\lambda}_{i-1}\vert\}$ and $\updelta_0 =\min\{\vert\lambda_i-\lambda_{i+1}\vert,\vert\lambda_i-\lambda_{i-1}\vert\}$. The \emph{a priori} estimate \eqref{eq: stabiltiy v gap apriori} holds for $\updelta_0>2\epsilon$.%
\end{theorem}

Remark that, for $s_1=1$ and $s_2=2$ and $m$ large, we have $\updelta_0\approx0.219$ so that the \emph{a priori} estimate \eqref{eq: stabiltiy v gap apriori} holds for $\epsilon<0.1$. This \emph{a priori} estimate is, however, suboptimal as \cref{fig: stability eve} shows.

In \cref{fig: stability}, we show numerically the high stability of the interface modes. We let $\mathcal{U}_{[-\eta,\eta]}$ be a uniform distribution with support in $[-\eta,\eta]$. We consider perturbations of the type $\tilde\epsilon_i\sim \mathcal{U}_{[-\eta,\eta]}$, where we call $\eta$ the perturbation size and display it as percentage relative to the resonator's size. \cref{fig: stability eva} shows that the interface eigenfrequency (lying in the gap) is only minimally perturbed even by perturbation in the size of $20\%$. We remark that numerically the bound of \eqref{eq: bound eigenvalues error} can even be sharpened to $\vert \hat{\lambda_k} - \lambda_k\vert \leq \frac{3}{2}\epsilon$. \cref{fig: stability eve} shows $\left\Vert\bm v - \hat{\bm v}\right\Vert_2$ for various perturbation sizes and normalised $\bm v$ and $\hat{\bm v}$. The black lines shows the average over $10^4$ runs while the gray area encloses the range from the minimum to the maximum value of $\left\Vert\bm v - \hat{\bm v}\right\Vert_2$.
\begin{figure}[h]
    \centering
    \begin{subfigure}[t]{0.48\textwidth}
        \centering
        \includegraphics[height=0.76\textwidth]{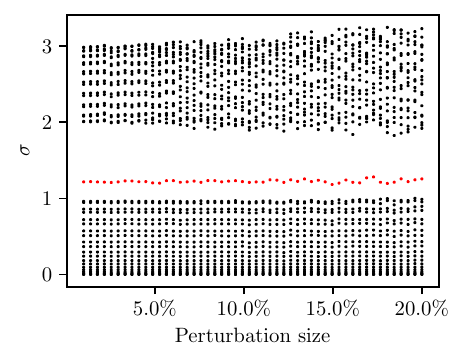}
        \caption{Spectrum of the capacitance matrix with perturbations given by $\tilde\epsilon_i\sim \mathcal{U}_{[-\eta,\eta]}$. For every perturbation size the spectrum of one realisation is shown. The eigenvalue in red corresponds to the localised interface mode.}
        \label{fig: stability eva}
    \end{subfigure}\hfill
    \begin{subfigure}[t]{0.48\textwidth}
        \centering
        \includegraphics[height=0.76\textwidth]{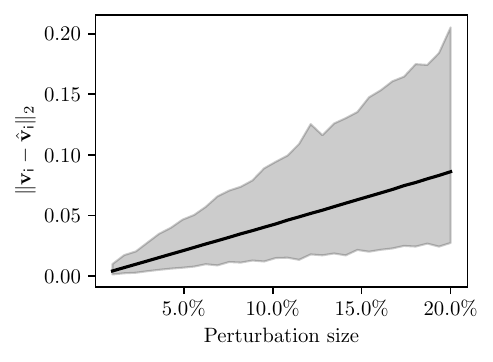}
        \caption{Stability of the interface mode. The solid black line shows the average dislocation over $10^4$ runs, while the gray area encloses the range from the minimum to the maximum dislocation observed over these realisations.}
        \label{fig: stability eve}
    \end{subfigure}

    \caption{The interface eigenvalue and the corresponding interface mode are very stable also in presence of large perturbations. Simulations in a system of $N=41$ resonators with $s_1=1$ and $s_2=2$. Perturbations are uniformly distributed in $(-\eta,\eta)$ where we call $\eta$ the perturbation size, expressed relatively to the resonators' sizes.}
    \label{fig: stability}
\end{figure}

\section{Non-Hermitian skin effect}\label{sec: skin effect}
Understanding the skin effect in non-Hermitian physical systems has been one of the hottest research topics in recent years \cite{zhang.zhang.ea2022review, okuma.kawabata.ea2020Topological,lin.tai.ea2023Topological,yokomizo.yoda.ea2022NonHermitian,wang.chong2023NonHermitian,leykam.bliokh.ea2017Edge,borgnia.kruchkov.ea2020Nonhermitian}. The skin effect is the phenomenon whereby the bulk eigenmodes of a non-Hermitian system are all localised at one edge of an open chain of resonators. This phenomenon is unique to non-Hermitian systems with non-reciprocal coupling. It has been realised experimentally in topological photonics, phononics, and other condensed matter systems \cite{ghatak.brandenbourger.ea2020Observation,longhi.gatti.ea2015Robust,franca.konye.ea2022Nonhermitian, wang.wang.ea2022NonHermitian}.
By introducing an imaginary gauge potential, all the bulk eigenmodes condensate in the same direction \cite{hatano.nelson1996Localization,yokomizo.yoda.ea2022NonHermitian,rivero.feng.ea2022Imaginary}. This has deep implications for the fundamental physics of the system. For example, the non-Hermitian skin effect means that the conventional bulk-edge correspondence principle is violated. This section is dedicated to demonstrating that a non-Hermitian system with an imaginary gauge potential as introduced in \cref{sec: non reciprocal 1D systems} exhibits an eigenmode condensation that is typical of the skin effect. To this end, we will make extensive use of the theory of Toeplitz matrices and operators.

\subsection{The skin effect in monomer systems of subwavlength resonators}\label{sec: nonhermitian skin effect in monomer systems}
We will first consider a system of equally spaced identical resonators (monomers),
that is, a chain of $N$ resonators with $s_i=s$ and $\ell_i=\ell$ for all $1\leq i\leq N$. For this case, the particular structure of the capacitance matrix allows us to apply existing results concerning the spectra of tridiagonal Toeplitz matrices, which we have briefly recalled in \cref{sec: Toeplitz theory}. We will be able to derive a variety of explicit results that are not possible for the more complex situations analysed in the subsequent sections. Nevertheless, \cref{cor: approx via eva eve} applies for more general cases and the numerical results which we present below.

Applying the explicit formula for the eigenpairs of perturbed tridiagonal Toeplitz matrices from \cref{lemma: spectrum perturbed toeplitz} to the gauge capacitance matrix $\capmatg$ given by \eqref{eq: def cap mat}, we obtain the
following theorem from \cite{ammari.barandun.ea2024Mathematical}.
\begin{theorem} \label{thm:decay}
    The eigenvalues of $\capmatg$  are given by
    \begin{align}
        \lambda_1 & = 0,\nonumber                                                                                                                                                                                                       \\
        \lambda_k & = \frac{\gamma}{s} \coth(\gamma\ell/2)+\frac{2\abs{\gamma}}{s}\frac{e^{\frac{\gamma\ell}{2}}}{\vert e^{\gamma\ell}-1\vert}\cos\left(\frac{\pi}{N}(k-1)\right), \quad 2\leq k\leq N . \label{eq: eigenvalues capmat}
    \end{align}
    Furthermore, the associated eigenvector $\bm a_k$ satisfies the following inequality, for $2\leq k\leq N$,
    \begin{align}
        \vert \bm a_k^{(i)}\vert \leq \kappa_k e^{-\gamma\ell\frac{i-1}{2}}\quad \text{for all } 1\leq i\leq N \label{eq: decay for eigemodes},
    \end{align}
    for some $\kappa_k\leq (1+e^{\frac{\gamma\ell}{2}})^2$.
\end{theorem}

Here, $\bm a_k^{(i)}$ denotes the $i$\textsuperscript{th} component of the $k$\textsuperscript{th} eigenvector. It is easy to show that the first eigenvector $\bm a_1$ is a constant vector (\emph{i.e.}, $\bm a_1^{(i)}$ is the same for all $i$). From \cref{thm:decay}, we can see that the eigenmodes display exponential decay both with respect to the site index $i$ and the factor $\gamma$.

Of particular interest is the application of \cref{thm:spectratoeplitz,prop: specra of fam of toeplitz} to the gauge capacitance matrix $\capmatg$. These results show that the localisation of the eigenvectors of both the finite and semi-infinite capacitance matrices depends on the Fredholm index of the symbol of the associated Toeplitz operator at the corresponding eigenvalue or, equivalently, of its winding number.

Indeed, the Toeplitz symbol of the Toeplitz operator associated to the capacitance matrix $\capmatg$ encloses the ellipse $E\subseteq\C$
\begin{align*}
    z-b\in E \Leftrightarrow \frac{\Re(z)^2}{a+c}+\frac{\Im(z)^2}{a-c}\leq1 ,
\end{align*}
where $a, b$ and $c$ are the below diagonal, diagonal and above diagonal elements of the Toeplitz operator. In particular, the boundary of the ellipses is drawn by $\theta\mapsto ae^{\i\theta}+b+ce^{-\i\theta}$ so that the winding in the interior of the ellipse is negative for $\gamma >0$ and positive for $\gamma < 0$. It is thus easy to show that
\begin{align}
    \lambda \in \sigma(\capmatg)\setminus \{0\} \Rightarrow \lambda \in \text{int}(E) \quad \text{and}\quad 0\in \partial E.
\end{align}

Specifically, every eigenvector having the corresponding eigenvalue laying in the region of the complex plane where the Fredholm index (or equivalently, the winding number of the symbol) is negative, has to be exponentially localised.
This shows the intrinsic topological nature of the skin effect.

In \cref{fig: winding and pseudospecturm}, we display the convergence of the pseudospectrum and the topologically protected region of negative winding number. We remark that the trivial eigenvalue $0$ is outside of the region as the winding number is not defined there. As a consequence, the corresponding eigenvector is not localised.
\begin{figure}[htb]
    \centering
    \begin{subfigure}[t]{0.48\textwidth}
        \centering
        \includegraphics[width=\textwidth]{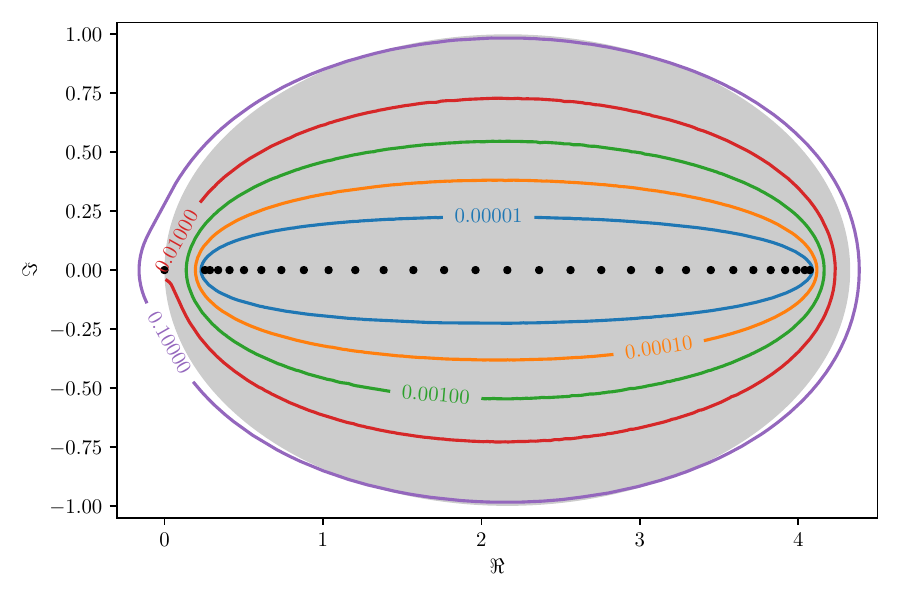}
        \caption{$N=30$.}
        \label{fig:pseudo30}
    \end{subfigure}
    \hfill
    \begin{subfigure}[t]{0.48\textwidth}
        \centering
        \includegraphics[width=\textwidth]{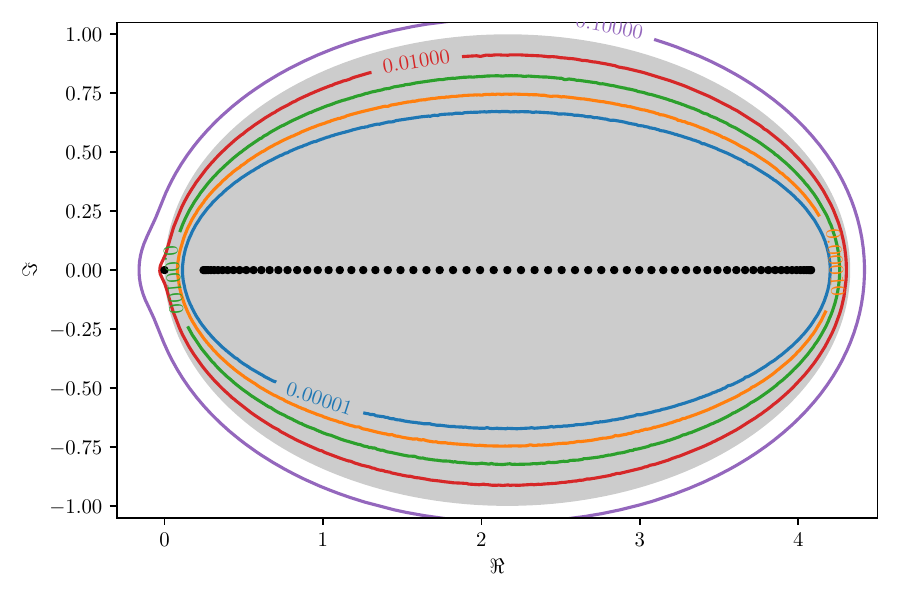}
        \caption{$N=70$.}
        \label{fig:pseudo70}
    \end{subfigure}
    \caption{$\epsilon$-pseudospectra of $\capmatg$ for $\epsilon = 10^{-k}$ for $k=1,\dots,5$. Shaded in grey is
        the region where the symbol $f_T$ has negative winding, for $T$ the Toeplitz operator corresponding to the semi-infinite structure. The pseudospectra are computed using \cite{gaulandrePseudoPy}.}
    \label{fig: winding and pseudospecturm}
\end{figure}

\begin{figure}[t]
    \centering
    \includegraphics[width=0.5\textwidth]{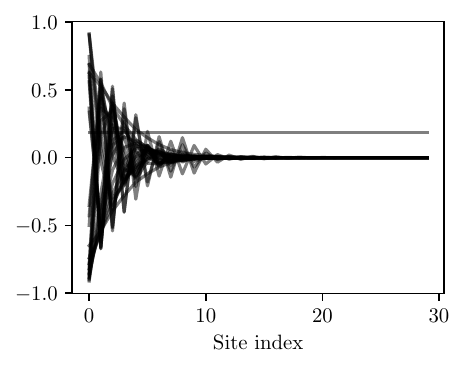}
    \caption{Eigenvectors of the capacitance matrix superimposed to display an exponential decay.}
    \label{fig: superimposed eigenvectors}
    \end{figure}

The explicit formula for eigenvalues from \cref{lemma: spectrum perturbed toeplitz} also gives some insight on the distribution of the \emph{density of states}. The property of the cosine to have minimal slope when it is close to its maxima and minima, causes eigenvalues to cluster at the edges of the range of the spectrum. This is demonstrated by the nonuniform distribution of the black dots in \cref{fig: winding and pseudospecturm}.

We conclude this section with a qualitative analysis of the spectral decomposition of $\capmatg$. In \cref{fig: eignemode localisation}, we show the eigenvector localisation as a function of the site index for finite arrays of various sizes. The localisation of a vector $v$ is measured using the quantity $\|v\|_\infty/\|v\|_2$. After rescaling the site index, we expect there to be some invariance to the array size $N$, based on the formulas of \cref{lemma: spectrum perturbed toeplitz}. \cref{fig: singular values decay} shows, on the other hand, the singular values of the eigenvector matrix, again with the indices normalised. The number of nonzero singular values is a proxy for the dimension of the range of the matrix. This has an exponential dependence on $N$. As a result, \cref{fig: singular val and localisation decay} shows that as $N$ increases, despite the number of eigenvalues growing like $N$, the rank of the matrix of eigenvectors
is fixed. 

\begin{figure}[htb]
    \centering
    \begin{subfigure}[t]{0.48\textwidth}
        \centering
        \includegraphics[width=\textwidth]{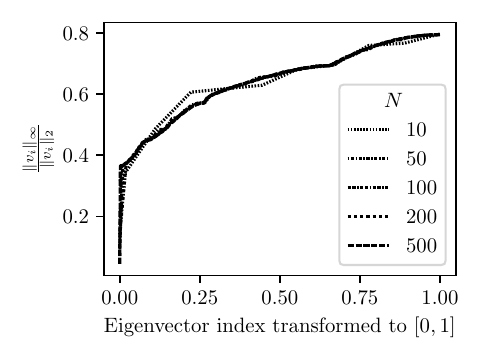}
        \caption{The eigenvector localisation $\frac{\vert v_i\vert_\infty}{\vert v_i\vert_2}$ does not depend on $N$ after rescaling.}
        \label{fig: eignemode localisation}
    \end{subfigure}
    \hfill
    \begin{subfigure}[t]{0.48\textwidth}
        \centering
        \includegraphics[width=\textwidth]{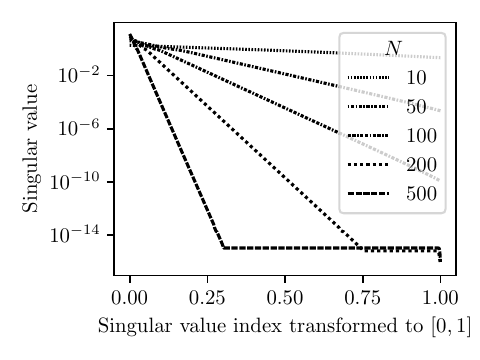}
        \caption{The singular values of the eigenvector matrix decay exponentially in $N$. This shows that the rank of the eigenvector matrix is independent of the size.}
        \label{fig: singular values decay}
    \end{subfigure}
    \caption{As $N\to \infty$, arbitrarily many eigenmodes become close to parallel. Here, this is shown in the case $s=\ell=1$ and $\gamma=0.5$.}
    \label{fig: singular val and localisation decay}
\end{figure}

\subsection{Non-Hermitian skin effect in dimer systems of subwavelength resonators}\label{sec: nonhermitian skin effect in dimer systems}

In \cref{sec: spectrum of tridiag 2 toeplitz matrices with perturbations}, we have discussed the general form of eigenvalues and eigenvectors of tridiagonal $2$-Toeplitz matrices with perturbations in the corners. In this subsection, we will use those results to show that the skin effect --- whose existence has been shown in the previous subsection --- is a feature present also in dimer systems. To do this, we will use quite different tools from those used for the case of identical and equally spaced resonators.

We consider a system as in \cref{fig:setting} composed of dimers --- that is $s_i=s_{i+2}$ for all $1\leq i\leq N-3$ and $\ell_i=\ell$ for all $1\leq i\leq N$ --- governed by \eqref{eq: gen Strum-Liouville}. The gauge capacitance matrix from \eqref{eq: explicit coef cap mat} takes the form
\begin{align}
    C^\gamma =
    \begin{pmatrix}
        \tilde{\alpha}_{1} & \beta_{1}  &            &            &          &          &                    \\
        \eta_{1}           & \alpha_{2} & \beta_{2}  &            &          &          &                    \\
                           & \eta_{2}   & \alpha_{1} & \beta_{1}  &          &          &                    \\
                           &            & \eta_{1}   & \alpha_{2} & \ddots   &          &                    \\
                           &            &            & \ddots     & \ddots   & \ddots   &                    \\
                           &            &            &            & \eta_{2} & \alpha_1 & \beta_1            \\
                           &            &            &            &          & \eta_{1} & \tilde{\alpha}_{2} \\
    \end{pmatrix},
\end{align}
where
\begin{align*}
    \alpha_1         & =\frac{\gamma}{s_1} \frac{\ell}{1-e^{-\gamma \ell}} -\frac{\gamma}{s_2} \frac{\ell}{1-e^{\gamma \ell}}, & \alpha_2         & =\frac{\gamma}{s_2} \frac{\ell}{1-e^{-\gamma \ell}} -\frac{\gamma}{s_1} \frac{\ell}{1-e^{\gamma \ell}}, \\
    \tilde{\alpha}_1 & =\frac{\gamma}{s_1} \frac{\ell}{1-e^{-\gamma \ell}},                                                    & \tilde{\alpha}_2 & =- \tilde{\alpha}_1,                                                                                    \\
    \beta_1          & = - \frac{\gamma}{s_1} \frac{\ell}{1-e^{-\gamma \ell}},                                                 & \beta_2          & =- \frac{\gamma}{s_2} \frac{\ell}{1-e^{-\gamma \ell}},                                                  \\
    \eta_1           & =\frac{\gamma}{s_1} \frac{\ell }{1-e^{\gamma \ell }},                                                   & \eta_2           & =\frac{\gamma}{s_2} \frac{\ell }{1-e^{\gamma \ell }}.                                                   \\
\end{align*}

One may remark that all rows of $C^\gamma$ sum to 0 and thus $\bm 1\in\ker(C^\gamma)$. We will see that this is the only eigenvector of $C^\gamma$ that is not localised.

The first result consists in showing that the eigenvectors of the capacitance matrix exhibit an exponential decay and that the eigenvectors therefore condense. Its proof is nearly a direct application of the eigenvector formula (\ref{equ:eigenvectoroddmatrixtwoperturb3}) in Theorem \ref{thm:eigenvectoroddmatrixtwoperturb2}. In particular, employing (\ref{equ:eigenvectoroddmatrixtwoperturb3}) we only need to control the ``Chebyshev-like'' polynomials $\widehat p_j^{(\xi_p, \xi_q)}(\lambda), \widehat q_j^{(\xi_p, \xi_q)}(\lambda)$'s. This then utilises the boundedness of Chebyshev polynomials in $[-1,1]$ shown in (\ref{equ:boundofchebyshev1}) and the eigenvalue distribution derived in Theorem \ref{thm:eigenvaluethm1}. We thus see that the interface modes and the skin effect actually correspond to the opposite behaviors of the Chebyshev polynomials in and outside the interval $[-1,1]$; see Figure \ref{fig: cheby poly}. We refer the reader to \cite{ammari.barandun.ea2023Perturbed} for a detailed proof. 

\begin{lemma} \label{thm:skineffect1}
    Except for at most $11$ eigenvalues $\{\lambda_r\}$ of $A_{2 m+1}^{(a, b)}$, the corresponding eigenvectors $\bm x$ in Theorem \ref{thm:eigenvectoroddmatrixtwoperturb2} satisfy 
    \begin{equation}\label{equ:skineffectequ1}
        \babs{ \bm x^{(j)}}\leq Mj \left(\sqrt{\frac{\eta_{1}\eta_{2}}{\beta_{1}\beta_{2}}}\right)^{\lfloor\frac{j-1}{2}\rfloor}
    \end{equation}
    for some constant $M>0$ independent of the $\lambda_r$'s, where $ \bm x^{(j)}$ is the $j$\textsuperscript{th} component of ${ \bm x}$. The estimate \eqref{equ:skineffectequ1} holds also for the eigenvectors $ \bm x$ of $A_{2 m}^{(a,b)}$ associated with the eigenvalue $\lambda_r$ in Theorem \ref{thm:eigenvectoroddmatrixtwoperturb2}, except for at most $12$ $r$'s.
\end{lemma}
As shown in \cite{ammari.barandun.ea2023Perturbed}, the natural consequence of this theorem is the condensation of the eigenvectors of the gauge capacitance matrix.

\begin{theorem} \label{thm: condensation of eigenmodes for dimer systems}
    All but a few (independent of $N$) eigenvectors of the gauge capacitance matrix $C^\gamma$ satisfy the following estimate:
    \begin{equation}\label{eq: estimate eigevector cap mat}
        \babs{ \bm x^{(j)}}\leq Mj e^{-\ell\gamma\lfloor\frac{j-1}{2}\rfloor},
    \end{equation}
    where $ \bm x^{(j)}$ is the $j$\textsuperscript{th} component of the eigenvector $ \bm x$.
\end{theorem}

\cref{thm:skineffect1} and \cref{thm: condensation of eigenmodes for dimer systems} are stated to hold up to a finite number of eigenpairs. This is due to the fact that the proof relies on multiple applications of Cauchy's interlacing Theorem. Nevertheless, numerically, we can verify that there is  exactly one eigenvector which is not localised. More precisely, we show in Figure \ref{fig: criteria decay} that only the eigenvalue $\lambda_1=0$ satisfies $\vert y(\lambda_i)\vert \geq 1$ for $y$ as in \eqref{equ:normalizedfunction1}. In Figure \ref{fig: skin effect dimers}, we plot the eigenmodes of a system of $25$ dimers.
\begin{figure}[h]
    \centering
    \begin{subfigure}[t]{0.48\textwidth}
        \includegraphics[width=1\textwidth]{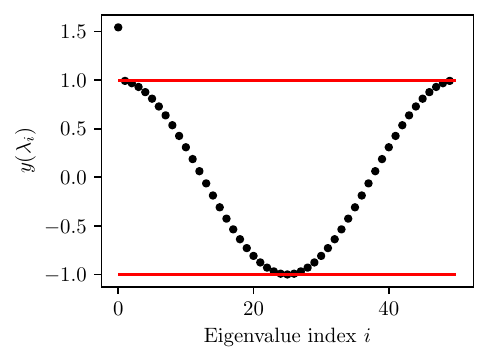}
        \caption{The black dots show the value $y(\lambda_i)$ for the eigenvalues $\lambda_i$ of $C^\gamma$. The red line show the boundaries stability zone $y=\pm 1$. Only for $\lambda=0$, $y(\lambda)$ lays outside of this zone.}
        \label{fig: criteria decay}
    \end{subfigure}\hfill
    \begin{subfigure}[t]{0.48\textwidth}
        \includegraphics[width=1\textwidth]{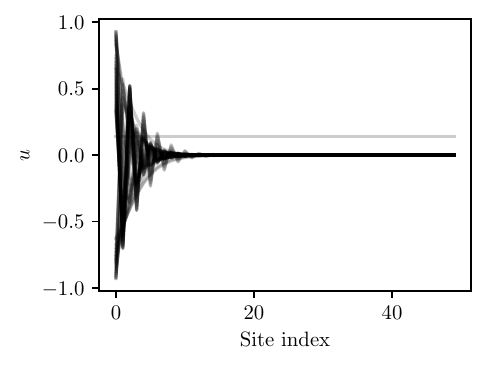}
        \caption{Eigenmodes superimposed on one another to portray the skin effect: all modes except one are exponentially localised on the left edge of the system.}
        \label{fig: skin effect dimers}
    \end{subfigure}
    \caption{Eigenmode localisation for a system of $25$ dimers ($N=50$), $\ell_i=1$, $s_1=1$, $s_2=2$ and $\gamma=1$.}
\end{figure}

It is interesting to remark that the eigenvalue $0$ is both the only outlier of Figure \ref{fig: criteria decay} and the only point laying on the trace of the eigenvalues of the symbol of the $2$-Toeplitz operator; see Figure \ref{Fig: spectrum_dimer_system} for an illustration and Theorem \ref{thm: exponential_decay_capacitance_operators} for a theoretical demonstration.

\subsection{Stability of the skin effect}
In this section, we present a stability estimate of the skin effect. We illustrate the stability of the skin effect for monomer systems of subwavelength resonators. For the case of dimer systems, one can derive similar results based on the constructions in Sections \ref{sec: spectrum of tridiag 2 toeplitz matrices with perturbations} and \ref{sec: nonhermitian skin effect in dimer systems}.

According to the theory in Section \ref{sec: nonhermitian skin effect in monomer systems}, the capacitance matrix in this case is given by $$
    \capmatg = T_N^{(\eta, \beta)},
$$
for $T_{N}^{(\eta, \beta)}$ defined as in (\ref{equ:toeplitzmatrix1}) and $\eta, \alpha, \beta$ are such that $$ \eta \beta > 0, \quad \text{ and } \quad  \eta+\alpha+\beta=0.$$

In order to study the stability of the non-Hermitian skin effect with respect to random imperfections in the system design, we either add random errors to the positions of the resonators (keeping the length of the resonators unchanged) or to the $\gamma$-term and then repeatedly compute the subwavelength eigenfrequencies and their associated eigenmodes. The perturbations in the positions and in the values of the $\gamma$-parameter are drawn at random from uniform distributions with zero-mean values.
Since these random perturbations affect only the tridiagonal entries of the gauge capacitance matrix $\hat{\capmat}^\gamma$ of the randomly perturbed system, we can write in both cases that
\begin{equation}\label{equ:toeplitzallperturbed1}
    \hat{\capmat}^\gamma = \widehat T_N^{(a,b)}:=\left(\begin{array}{cccccc}
            \alpha + a +\epsilon_{\alpha, 1} & \beta+\epsilon_{\beta, 1}     & 0                             & \ldots & 0                               & 0                              \\
            \eta +\epsilon_{\eta, 2}         & \alpha + \epsilon_{\alpha, 2} & \beta +\epsilon_{\beta, 2}  & \ldots & 0                               & 0                              \\
            0                                  & \eta+\epsilon_{\eta, 3}       & \alpha+\epsilon_{\alpha, 3} & \ldots & 0                               & 0                              \\
            \ldots                             & \ldots                          & \ldots                        & \ldots & \ldots                          & \ldots                         \\
            \ldots                             & \ldots                          & \ldots                        & \ldots & \alpha+\epsilon_{\alpha, N-1} & \beta+\epsilon_{\beta, N-1}  \\
            0                                  & 0                               & 0                             & \ldots & \eta+\epsilon_{\eta, N}       & \alpha+b+\epsilon_{\alpha,N}
        \end{array}\right)
\end{equation}
with $a=\eta, b=\beta, \eta = -\frac{\ell\gamma}{ s} \frac{1}{1-e^{-\gamma\ell}}, \alpha= \coth(\gamma\ell/2), \beta = \frac{\ell\gamma}{ s} \frac{1}{1-e^{\gamma\ell}}$.

We first derive stability results for the eigenvalues of $T_N^{(a,b)}$. This relies on a crucial observation that the tridiagonal matrix $T_N^{(a,b)}$ always has the same eigenvalues as a Hermitian matrix. Thus we first recall the following well-known Weyl theorem for the stability of eigenvalues of Hermitian matrices; see \cite[Theorem 1.1.7]{silva2011Matrix}, \cite[Theorem 8.1.6]{golub.vanloan2013Matrix} and \cite[Theorem 10.3.1]{parlett1998symmetric}. 

\begin{proposition}\label{thm:weyltheorem1}
    Let $A$ and $E$ be $N \times N$ Hermitian matrices. For $k \in\{1, \ldots, N\}$, denote by $\lambda_k(A+E), \lambda_k(A)$ the $k$\textsuperscript{th} eigenvalue of $A+E$ and $A$, respectively. Assume these to be arranged in a decreasing sequence. Then,
    \[
        \babs{\lambda_k(A+E)-\lambda_k(A)}\leqs \Vert{E}\Vert_2,
    \]
    where $\Vert{E}\Vert_2$ is the operator norm of $E$.
\end{proposition}

Now we introduce the following result for the stability of the eigenvalues of $T_{N}^{(a,b)}$; see \cite{ammari.barandun.ea2023Stability} for a detailed justification.  From this we can state that the eigenvalues of matrices of the type $T_{N}^{(a,b)}$ with $a, b\in \mathbb R$ are stable under \emph{physical} perturbations. We remark here that physical means perturbing either the geometry or the imaginary gauge potential. It is well-known that $T_{N}^{(a,b)}$ is not stable under arbitrary perturbations, but these do not occur as a result of physical modifications to the system. 

\begin{proposition} \label{thm:eigenvaluestability0}
    The eigenvalues of $T_N^{(a,b)}$ and $\widehat T_N^{(a,b)}$ are all real numbers. Let $\{\lambda_k\}, \{\widehat \lambda_k \}$ be respectively the eigenvalues of $T_N^{(a,b)}$ and $\widehat T_N^{(a,b)}$, arranged in  decreasing sequences. Assuming that
    \begin{equation}
        \max_{j=1,\dots, N}\left(\babs{\epsilon_{\eta, j}}, \babs{\epsilon_{\alpha, j}}, \babs{\epsilon_{\beta, j}}\right)\eqqcolon \epsilon
    \end{equation}
    with $\epsilon\leqs \min(\babs{\eta}, \babs{\beta})$, then we have
    \[
        \babs{\widehat \lambda_k-\lambda_k}\leqs C_1(\eta, \beta, \epsilon)\epsilon,
    \]
    where
    \begin{equation}\label{equ:eigenvaluestability2}
        C_1(\eta, \beta, \epsilon)= \frac{\babs{\beta} +\babs{\eta} +\epsilon}{\sqrt{\beta \eta}}+1.
    \end{equation}
\end{proposition}

Applying \cref{thm:eigenvaluestability0} to the eigenvalues of $T_N^{(\eta, \beta)}$ (stated in \cref{lemma: spectrum perturbed toeplitz}) yields the following stability estimate first derived in \cite{ammari.barandun.ea2023Stability}. 
\begin{theorem} \label{thm:eigenvaluestability1thm}
    For the eigenvalue $\widehat\lambda_{k}$ of the perturbed capacitance matrix $\widehat \capmatg$ with
    \begin{equation}\label{equ:eigenvaluestability1}
        \max_{j=1,\dots,N}\left(\babs{\epsilon_{\eta, j}}, \babs{\epsilon_{\alpha, j}}, \babs{\epsilon_{\beta, j}}\right)\eqqcolon \epsilon
    \end{equation}
    and $\epsilon\leqs \min(\babs{\eta}, \babs{\beta})$, we have $\widehat \lambda_1 = \epsilon_1$ and
    \begin{equation}\label{equ:eigenvaluestability3}
        \widehat\lambda_{k} = \lambda_{k}+\epsilon_k = \alpha+2\sqrt{\eta\beta}\cos\left(\frac{(k-1)\pi}{N}\right)+ \epsilon_k, \quad k=2, \cdots, N,
    \end{equation}
    with $\babs{\epsilon_k}\leqs C_1(\eta, \beta, \epsilon)\epsilon, 1\leqs k\leqs N$ and $C_1(\eta, \beta, \epsilon)$ being defined by (\ref{equ:eigenvaluestability2}). In particular, all the ${\widehat \lambda}_k$'s are real numbers.
\end{theorem}

\begin{remark}
    One can apply \cref{thm:eigenvaluestability0} and \cref{thm:eigenvectorstability1} to derive similar stability results for $T_N^{(0,0)}$. This corresponds to many examples of the non-Hermitian skin effect in condensed matter theory and quantum mechanics. The stability results presented here can therefore be immediately applied to these examples. For the eigenvalues of tridiagonal Toeplitz matrices with various perturbations on the corners, we refer the reader to \cite{yueh2005Eigenvalues, yueh.cheng2008Explicit}.
\end{remark}

Next, we estimate the stability of the eigenvectors of $\capmatg = T_N^{(\eta, \beta)}$. For $\lambda_k$ an eigenvalue of $\capmatg$, we let
\begin{equation}\label{equ:defiofpj1}
    p_{j}(\lambda_{k}) =  \left(\eta \sin \left(\frac{(j+1)(k-1) \pi}{N}\right)-\eta \sqrt{\frac{\eta}{\beta}} \sin \left(\frac{j(k-1) \pi}{N}\right)\right), \quad j=0,\cdots, N-1.
\end{equation}
Note that
\begin{equation}\label{equ:boundonp1}
    \babs{p_j(\lambda_{k})}\leqs \babs{\eta}\left(1+\sqrt{\frac{\eta}{\beta}}\right), \quad j=0, \cdots, N-1.
\end{equation}
The following results hold.
\begin{theorem}\label{thm:eigenvectorstability1}
    For $\widehat T_{N}^{(\eta, \beta)}$ defined by (\ref{equ:toeplitzallperturbed1}) and satisfying (\ref{equ:eigenvaluestability1}) with $\epsilon\leqs \min(\babs{\eta}, \babs{\beta})$ and its eigenvalues $\widehat\lambda_k=\lambda_{k}+\epsilon_k, k=2,\cdots, N$ defined by (\ref{equ:eigenvaluestability3}), the corresponding eigenvectors are given by
    \begin{equation}\label{equ:eigenvectorperturb1}
        \begin{aligned}
            {\widehat{\bm x}}_k = & \left(p_{0}(\lambda_{k})+\delta_0(\lambda_k),s \left(p_{1}(\lambda_{k})+\delta_{1}(\lambda_k)\right), s^2\left(p_2(\lambda_k)+\delta_{2}(\lambda_k)\right),  \cdots,\right. \\
                                    & \qquad \left. s^{N-1} \left(p_{N-1}(\lambda_k)+\delta_{N-1}(\lambda_k)\right) \right)^{\top},\quad k=2,\cdots, N,
        \end{aligned}
    \end{equation}
    where $s=\sqrt{\frac{\eta}{\beta}}$ and $p_{j}(\lambda_k)$ is defined in (\ref{equ:defiofpj1}). Moreover, we have
    \begin{equation}\label{equ:eigenvectorstability1}
        \babs{ (-s)^j p_{j}(\lambda_k)}\leqs \left(\frac{\eta}{\beta}\right)^{\frac{j}{2}}\babs{\eta}\left(1+\sqrt{\frac{\eta}{\beta}}\right), \quad j=0,1,\cdots, N-1,
    \end{equation}
    and
    \begin{equation}\label{equ:eigenvectorstability2}
        \babs{s^{j}\delta_j(\lambda_k)}\leqs \zeta_{k,j}  \epsilon,\quad j=0,1,\cdots, N-1,
    \end{equation}
    where
    \[
        \zeta_{k,j} = \left(\sqrt{\frac{\eta}{\beta}}\left(\frac{\babs{\beta}(\babs{\eta}+\epsilon)}{(\babs{\beta}-\epsilon)|\eta|}\right)\right)^j\left(a_+r_{k,+}^{j}(\eta, \beta, \epsilon) +a_-r_{k,-}^j(\eta, \beta, \epsilon) -\zeta\right)
    \]
    with
    \begin{equation}\label{equ:rpmformula1}
        \begin{aligned}
              & r_{k,\pm}(\eta, \beta, \epsilon)                                                                                                                                                                                                                          \\
            = & \left(\babs{\cos\left(\frac{(k-1)\pi}{N}\right)}+\frac{C_2(\eta, \beta, \epsilon)\epsilon}{\sqrt{\eta \beta }}\right)\pm\sqrt{\left(\babs{\cos\left(\frac{(k-1)\pi}{N}\right)}+\frac{C_2(\eta, \beta, \epsilon)\epsilon}{\sqrt{\eta \beta }}\right)^2+1},
        \end{aligned}
    \end{equation}
    $C_2(\eta, \beta, \epsilon) = \frac{\babs{\beta}+\babs{\eta}+\epsilon}{2\sqrt{\beta \eta}}+1$, and $a_+, a_-, \zeta$ being bounded constants. In particular, for those indices $k$ such that
    \begin{equation}\label{equ:eigenvectorstability3}
        \babs{\sqrt{\frac{\eta}{\beta}}\left(\frac{\babs{\beta}(\babs{\eta}+\epsilon)}{(\babs{\beta}-\epsilon)|\eta|}\right)r_{k,+}}<1,
    \end{equation}
    the corresponding eigenvector still has an exponential decay.

    As a result, there exists a constant $c$ such that if $\sqrt{\eta/\beta}<\sqrt{2}-1$ and $\epsilon< \frac{c}{N^2}$, then we still have an exponential decay for all the corresponding eigenvectors ${\widehat{\bm x}}_k, 2\leqs k \leqs N,$ of $\widehat T_{N}^{(a, b)}$. Furthermore, if we require $\sqrt{\eta/\beta}$ to be even smaller, then this exponential decay will remain for even larger values of $\varepsilon$.
\end{theorem}

Before applying Theorem \ref{thm:eigenvectorstability1}  to show the stability of the skin effect, we illustrate numerically the results stated there. In particular, we consider typical values in physical applications.  We let $\eta=0.15$, $\beta=3.15$ (which corresponds to $\ell=s=1$ and $\gamma=3$) and $\epsilon$ satisfying (\ref{equ:eigenvectorstability3}). The results are presented in \cref{fig: validation theoreical result}, where we show the eigenvectors of a  $50 \times 50$ matrix on a logarithmic axis. If the perturbations are sufficiently small that the condition \eqref{equ:eigenvectorstability3} is satisfied, then the eigenvectors still all have the $(\sqrt{\beta/\eta})$ decay rate. However, when the perturbations are large enough that condition \eqref{equ:eigenvectorstability3} does not hold for some indices,  the corresponding eigenmodes have a much lower decay rate.

\begin{figure}[h]
    \centering
    \begin{subfigure}[t]{0.48\textwidth}
        \centering
        \includegraphics[height=0.76\textwidth]{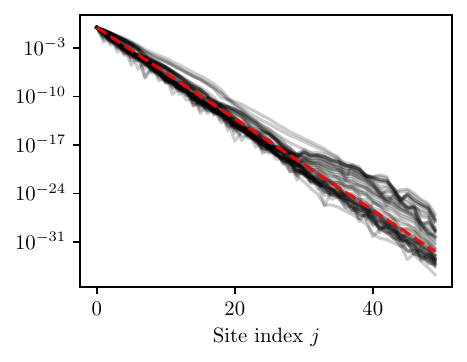}
        \caption{Exponential decay of the eigenvectors for $\epsilon$ satisfying \eqref{equ:eigenvectorstability3}. The eigenvectors superimposed on one another on a semi-log plot. The red dashed line represents $(\sqrt{\beta/\eta})^j$. We observe the same decay rate as the unperturbed case.}
        \label{fig: exponential decay small eps}
    \end{subfigure}\hfill
    \begin{subfigure}[t]{0.48\textwidth}
        \includegraphics[height=0.75\textwidth]{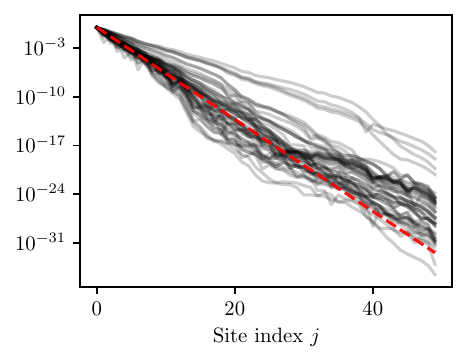}
        \caption{Decay of the eigenvectors for $\epsilon$ \emph{not} satisfying \eqref{equ:eigenvectorstability3}. The eigenvectors superimposed on one another on a semi-log plot. The red dashed line represents $(\sqrt{\beta/\eta})^j$. We observe several eigenvectors with lower decay rate than the one in Figure \ref{fig: exponential decay small eps}.}
        \label{fig: exponential decay big eps}
    \end{subfigure}
    \caption{Numerical illustration of the stability of the eigenvector decay rate predicted by \cref{thm:eigenvectorstability1}. The Toeplitz matrix has coefficients $\eta=0.15$, $\beta=3.15$ and is of size $50\times 50$.}
    \label{fig: validation theoreical result}
\end{figure}

Now, since a perturbation of order $\epsilon$ to the spacings $\{s_i\}$ between the subwavelength resonators or the coefficient $\gamma$ results in an $\BO(\epsilon)$ perturbation in the nonzero entries of the gauge capacitance matrix $\capmatg$, Theorem \ref{thm:eigenvectorstability1} can be immediately applied to $\hat \capmat^{\gamma}$ to obtain the stability of the eigenvectors of $\capmatg$ and the skin effect in the structure. We present below a corollary for a simplified stability estimate on the spectrum of $\capmatg$.

\begin{corollary}\label{cor:stabilityofcapacitancematrix1}
    Let $\widehat \capmat^{\gamma}$ be defined by (\ref{equ:toeplitzallperturbed1}).  Assume that (\ref{equ:eigenvaluestability1}) holds with $\epsilon\leqs \min(\babs{\eta}, \babs{\beta})$. Then the  normalised eigenvector ${\widehat {\bm x}}_k$ corresponding to
    the eigenvalue $\widehat\lambda_k=\lambda_{k}+\epsilon_k, k=2,\cdots, N$, defined by (\ref{equ:eigenvaluestability3}) is given by
    \[
        {\widehat {\bm x}}_k = \bm x_k +  \bm \Delta_k, \quad k=2,\cdots, N,
    \]
    where $\bm x_k$ is the eigenvector of $\capmatg$ and $\vect \Delta_k$ satisfies
    \[
        \bnorm{\vect \Delta_k}_{\infty}\leq C_1 \epsilon,
    \]
   for some constant $C_1$ independent of $\epsilon$, when
    \[
        \babs{\sqrt{\frac{\eta}{\beta}}\left(\frac{\babs{\beta}(\babs{\eta}+\epsilon)}{(\babs{\beta}-\epsilon)|\eta|}\right)r_{k,+}(\eta, \beta, \epsilon)}<1,
    \]
    with $r_{k,+}(\eta, \beta, \epsilon)$ being given by (\ref{equ:rpmformula1}). In particular, ${\widehat {\bm x}}_k$ still has an exponential decay.
\end{corollary}

It is important to notice that the asymptotic approximations of the resonant eigenfrequencies and eigenmodes through the spectrum of the gauge capacitance matrix hold also for the perturbed structure. The stability of the spectrum of the gauge capacitance matrix $C^{\gamma}$ implies the stability of the skin effect in the subwavelength regime. In particular, we have the following corollary on the stability of the eigenfrequencies and eigenmodes of the system (\ref{eq: gen Strum-Liouville}).

\begin{corollary}\label{cor:systemstability1}
Let \begin{equation}\label{equ:hatepsilon1}
\max_{j=1,\cdots, N}(|\epsilon_{\alpha, j}|, |\epsilon_{\beta, j}|, |\epsilon_{\eta, j}|)\leqs \frac{\ell\gamma}{s} \coth(\frac{\gamma \ell}{2})\epsilon : = \hat \epsilon,
\end{equation}
and 
\begin{equation}\label{equ:eigenvectorequ1}
	\bm x_k^{(j)}=\left(\frac{\eta}{\beta}\right)^{\frac{j-1}{2}}\left(\eta \sin \left(\frac{j(k-1) \pi}{N}\right)-\eta \sqrt{\frac{\eta}{\beta}} \sin \left(\frac{(j-1)(k-1) \pi}{N}\right)\right),\ j=1,\cdots, N.
	\end{equation}

    For a perturbed structure of resulting in \eqref{equ:toeplitzallperturbed1} such that
    \begin{equation}\label{equ:perturbedstructue1}
        s_i = s(1 + \epsilon_{i}), \quad i=1,\cdots, N,
    \end{equation}
    with $|\epsilon_{i}|\leqs \epsilon$ for $\epsilon$ sufficiently small, the  $N$ subwavelength eigenfrequencies $\hat \omega_i$ of (\ref{eq: gen Strum-Liouville}) satisfy, as $\delta\to0$,
    \begin{align*}
        \hat \omega_i =  \omega_i + \BO(\delta+\sqrt{\delta \epsilon}),
    \end{align*}
    where $(\omega_i)_{1\leqs i\leqs N}$ are the eigenfrequencies  for the unperturbed structure. Furthermore, let $\hat u_i$ be a subwavelength eigenmode corresponding to $\hat \omega_i$ and let $\bm a_i$ be the corresponding eigenvector of $\ell\inv\capmatg$. Then, for $\ell \vect a_i, i=2,\cdots, n,$ given by (\ref{equ:eigenvectorequ1}) and those index $i$ satisfying 
    \[
        \babs{\sqrt{\frac{\eta}{\beta}}\left(\frac{\babs{\beta}(\babs{\eta}+\hat \epsilon)}{(\babs{\beta}-\hat \epsilon)|\eta|}\right)r_{i,+}(\eta, \beta , \hat \epsilon)}<1
    \]
    with $r_{i,+}(\eta, \beta , \hat \epsilon)$, $\hat\epsilon$ being given by (\ref{equ:rpmformula1}),(\ref{equ:hatepsilon1}), respectively, we have
    \begin{align*}
        \hat u_i(x) = \sum_j \bm a_i^{(j)}V_j(x) + \BO(\delta+\epsilon), \quad i=2,\cdots, N,
    \end{align*}
    where  the $V_j$'s are defined by \eqref{eq: def cap mat}.
\end{corollary}

\subsubsection*{Random perturbations of the geometry}
In this subsection, we will provide numerical evidence of the stability of the non-Hermitian skin effect and show how it competes with Anderson-type localisation of the eigenmodes in the bulk when the disorder is large. We can consider perturbations in both the geometry and the local values of the damping parameter (\emph{i.e.}, the imaginary gauge potential). We will present here the case of a perturbed geometry and refer the reader to \cite{ammari.barandun.ea2023Stability} for the other case. For the sake of brevity, we fix the size $\ell$ of the resonators and perturb independently the spacing $s$ between the resonators. The numerical experiments presented here are for the discrete approximation \eqref{equ:toeplitzallperturbed1} and not for the damped wave equation \eqref{eq: gen Strum-Liouville}.

We consider the following perturbation:
\begin{align}
    {s_i} = 1 + \epsilon_i, \qquad \epsilon_i \sim \mathcal{U}_{[-\varepsilon,\varepsilon]}.
\end{align}
Here, we recall that $\mathcal{U}_{[-\varepsilon,\varepsilon]}$ is a uniform distribution with support in $[-\varepsilon,\varepsilon]$. In \cref{fig: condensation and winding}, we study how the eigenmodes of a system of $30$ subwavelength resonators behave as the disorder increases. These results are averaged based on 500 independent realisations. We show the relative proportion of eigenvalues that fall within the region of negative winding of the associated Toeplitz operator, as well as the proportion of eigenmodes accumulating at the left edge (which for this and the following figures has been defined as the number of eigenvectors that attain their maximal values, in absolute terms, in one of the first four resonators). We consider values of the disorder strength that are small enough to ensure that the resonators do not overlap. Both these quantities are constant for small disorder strengths then decrease once the disorder strength passes a certain threshold (as predicted by Theorem~\ref{thm:eigenvectorstability1}). The intersection of these two sets is also shown.

One notices very similar trends in the three lines in \cref{fig: condensation and winding}, with small differences due to the imperfect formulation of the accumulation measure and the perturbations. On the other hand, \cref{fig: localisation} shows the localisation of the eigenvectors of $\hat{\capmat}^\gamma$ from \cref{equ:toeplitzallperturbed1} for different disorder strengths. 
The localisation of the eigenvectors is measured using the quantity $\Vert v_i\Vert_\infty / \Vert v_i\Vert_2$ and the different lines correspond to different disorder strengths $\varepsilon$. We notice that the lines are indistinguishable, indicating that the localisation of the eigenvectors is independent of any random perturbation of the positions of the resonators.

\cref{fig: phase change and topological protection} shows similar stability properties  as those in \cref{fig: condensation and winding}, but here the relative number of eigenvalues falling within the region with negative winding is plotted for different values of $\varepsilon$ and $\gamma$. On the left side of the figure, we see the topologically protected region: for these values of $\gamma$ any small perturbation size $\varepsilon$ will not cause any eigenvalue to exit the region and thus the corresponding eigenvector remains accumulated at the left edge of the structure.

In \Cref{fig: IPR geometry}, the localisation of the eigenvectors is measured using the inverse participation ratio (IPR) defined as $\Vert v_i\Vert_4 / \Vert v_i\Vert_2$. IPR is commonly used to
distinguish Anderson localised and extended eigenmodes \cite{murphy.wortis.ea2011Generalized,lyra.mayboroda.ea2015Dual}. One notices that the results in \Cref{fig: IPR geometry} are similar to those in \cref{fig: localisation}. 

The results in \cref{fig: geometry perturbation} show how the proportion of eigenvectors localised to the left edge of the system decreases as the disorder increases. Studying in the eigenvectors themselves, as shown in \cref{fig: condensed_single_realisations} for three different values of the disorder strength, we see that increasing disorder means an increasing number of eigenvectors that are localised in the bulk rather than on the left edge. This behaviour is typical of Anderson-type localisation in disordered systems and demonstrates the internal competition between the skin effect and Anderson localisation.

\begin{figure}[h]
    \centering
    \begin{subfigure}[t]{0.45\textwidth}
    \centering
\includegraphics[height=0.75\textwidth]{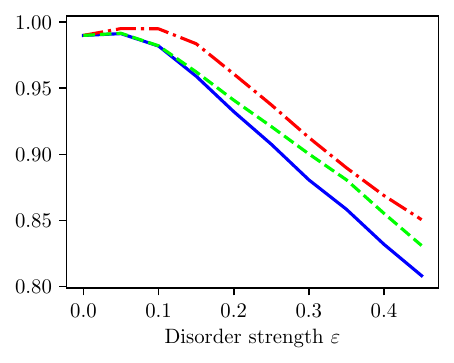}
    \caption{Eigenmode accumulation at one edge and topological winding. The green dashed line shows the average proportion of eigenvectors which are localised at the left edge. The red dash-dot line shows the average proportion of eigenvalues that lay in the topologically protected region. The blue solid line shows the proportion of eigenpairs that have \emph{both} eigenvalues in the topologically protected region and eigenvectors accumulated on the left edge.}
    \label{fig: condensation and winding}
    \end{subfigure}\hfill
    \begin{subfigure}[t]{0.45\textwidth}
    \centering
    \includegraphics[height=0.76\textwidth]{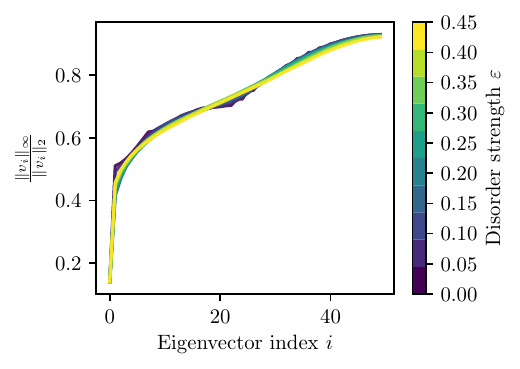}
    \caption{Eigenmode localisation. Each line shows the average eigenmode localisation for a different value of the disorder strength $\varepsilon$. For small $\epsilon$, the localisation is due to the skin effect while, for large $\epsilon$, it is consequence of the Anderson localisation. As the lines are indistinguishable we conclude that the eigenmode localisation is independent of disorder strength; as $\varepsilon$ increases, modes might be localised in the bulk but will not become delocalised.}
    \label{fig: localisation}
    \end{subfigure}
    \\[5mm]
    \begin{subfigure}[t]{0.45\textwidth}
    \centering
    \includegraphics[height=0.76\textwidth]{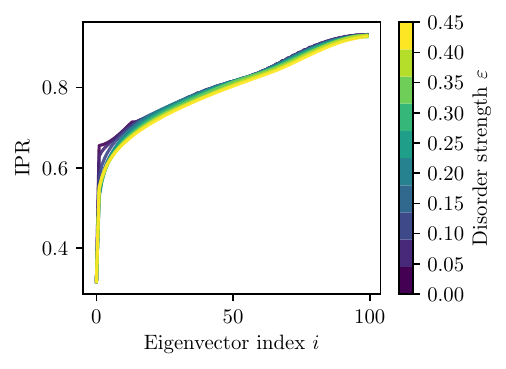}
    \caption{Similar plot as in \cref{fig: localisation} but using IPR as localisation measure. No significant difference is noticed.}
    \label{fig: IPR geometry}
    \end{subfigure}\hfill
    \begin{subfigure}[t]{0.45\textwidth}
    \includegraphics[height=0.75\textwidth]{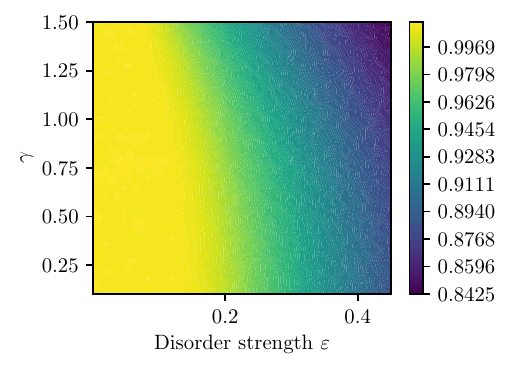}
        \caption{Phase change and topological protection. The color scale shows the average proportion of eigenvalues that lay in the topologically protected region for different values of $\gamma$. The left yellow zone is the stability region.}
    \label{fig: phase change and topological protection}
    \end{subfigure}
    \caption{Competition between the non-Hermitian skin effect and Anderson localisation when perturbing the geometry. The non-Hermitian skin effect shows stability with respect to random perturbations. Outside of the stability region, there is competition with Anderson localisation. Averages are computed over $500$ runs for a system of $50$ resonators with $\ell=s=1$.}
    \label{fig: geometry perturbation}
\end{figure}

\begin{figure}
    \centering
    \begin{subfigure}[t]{0.32\textwidth}
    \centering
    \includegraphics[height=0.76\textwidth]{fig/condensed_eigenvectors_perturbed_s_ep01.pdf}
    \caption{Single realisation with disorder strength $\varepsilon=0.1$. All eigenmodes are accumulated on the left edge.}
    \label{fig: condensed_single_realisation1}
    \end{subfigure}\hfill
    \begin{subfigure}[t]{0.32\textwidth}
    \centering
    \includegraphics[height=0.76\textwidth]{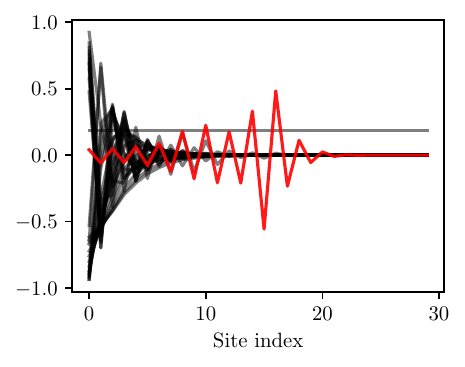}
    \caption{Single realisation with disorder strength $\varepsilon=0.2$. One eigenmode localised in the bulk is highlighted in red.}
    \label{fig: condensed_single_realisation2}
    \end{subfigure}\hfill
    \begin{subfigure}[t]{0.32\textwidth}
    \centering
    \includegraphics[height=0.76\textwidth]{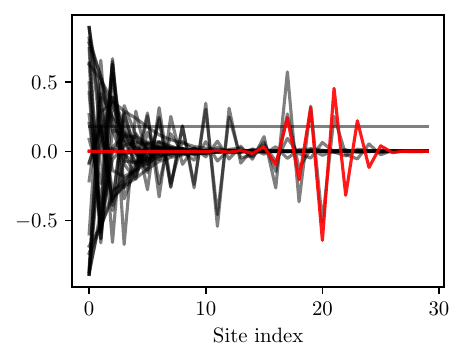}
    \caption{Single realisation with disorder strength $\varepsilon=0.4$. One eigenmode localised in the bulk is highlighted in red.}
    \label{fig: condensed_single_realisation4}
    \end{subfigure}\\[2mm]
    \begin{subfigure}[t]{0.64\textwidth}
    \centering
    \includegraphics[height=0.38\textwidth]{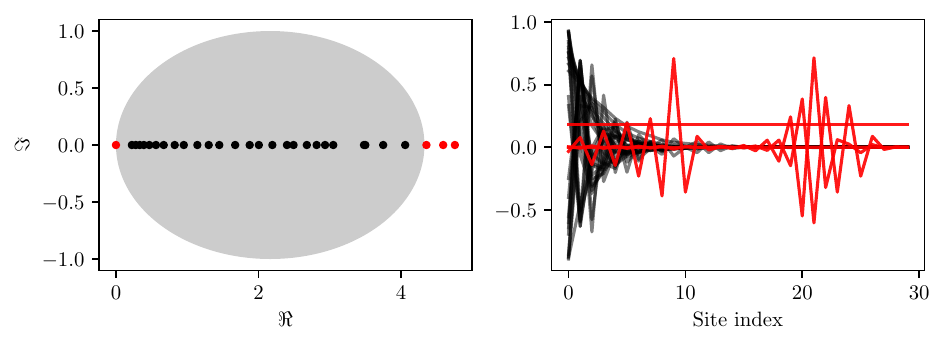}
    \caption{Relation between the eigenvalues in the topologically protected region and the condensation of eigenvectors in a single realisation with $\epsilon=0.4$. Eigenvalues laying outside the protected region and the corresponding eigenvectors are shown in red.}
    \label{fig: condensed_relation}
    \end{subfigure}

    \caption{Stability of the non-Hermitian skin effect under perturbations of the geometry. Eigenmode condensation on the left edge of the structure with some eigenmodes localised in the bulk. Single realisations with  $N=30$, $s=\ell=1$, and $\varepsilon=0.1,0.2,0.4,0.4$ for \cref{fig: condensed_single_realisation1}, \ref{fig: condensed_single_realisation2}, \ref{fig: condensed_single_realisation4} and \ref{fig: condensed_relation},  respectively. This should be compared with \cref{fig: superimposed eigenvectors}, where there is no disorder.}
    \label{fig: condensed_single_realisations}
\end{figure}

\subsection{Topological origin of the skin effect}\label{sec:topologicalorigin}
In this section, employing the Toeplitz theory in Section \ref{section:ktoeplitztheory}, we elucidate the topological origin of the skin effect in the polymer systems of subwavelength resonators which includes the monomer and dimer systems as special cases. For the semi-infinite polymer systems of subwavelength resonators, we obtain a perturbed tridiagonal $k$-Toeplitz operator
\[
    T^{a}(f) = \begin{pmatrix}
        \alpha_1 + a & \beta_1  &                                                           \\
        \eta_1       & \alpha_2 & \beta_2    &                                              \\
                     & \ddots   & \ddots     & \ddots                                       \\
                     &          & \eta_{k-1} & \alpha_{k} & \beta_{k}                       \\
                     &          &            & \eta_{k}   & \alpha_{1} & \beta_{1}          \\
                     &          &            &            & \ddots     & \ddots    & \ddots
    \end{pmatrix}
\]
with $f$ being the symbol defined in (\ref{eq: symbol tridiagonal operator}).

It has been observed in \cite{ammari.barandun.ea2023Perturbed} that the exponential decay of the eigenvectors of $k$-Toeplitz matrices for $k \geq 2$ is due to the winding of the eigenvalues of the symbol being nontrivial. Now, by the spectral theory for the tridiagonal $k$-Toeplitz operators introduced in Section \ref{section:ktoeplitztheory}, we have the following results for the topological origin of the skin effect in subwavelength resonator systems, validating the observation made in \cite{ammari.barandun.ea2023Perturbed}.
\begin{theorem}\label{thm: exponential_decay_capacitance_operators}
    Suppose $\Pi_{j=1}^k \eta_j\neq 0$ and $\Pi_{j=1}^k \beta_j\neq 0$. Let $f(z) \in \mathbb{C}^{k\times k}$ be the symbol (\ref{eq: symbol tridiagonal operator}) and let $\lambda \in \C\setminus\sigma_{\mathrm{ess}}(T^{a}(f))$. If $\sum_{j=1}^k\operatorname{wind}(\lambda_j, \lambda)<0$, then there exists an eigenvector $\bm x$ of $T^{a}(f)$ associated to $\lambda$ and some $\rho<1$ such that
    \begin{equation}\label{eq: bound on eigenvectorN}
        \frac{\lvert \bm x_j\rvert}{\max_{i}\lvert \bm x_i\rvert} \leq C_1 \lceil j/k\rceil \rho^{\lceil j/k\rceil-1},
    \end{equation}
    where $C_1>0$ is a constant depending only on $\lambda, \alpha_j, \beta_j, \eta_j, j=1,\cdots, k$. If $\sum_{j=1}^k\operatorname{wind}(\lambda_j, \lambda)>0$, then the above results hold for the left eigenvectors.
\end{theorem}

Theorem \ref{thm: exponential_decay_capacitance_operators} elucidates the topological origin of the skin effect in the polymer system of subwavelength resonators. In particular, the skin effect holds for all $\lambda$ in the region
\begin{equation}\label{equ:nonzerowindingregion1}
    G:=\bigg\{ \lambda \in \mathbb{C}\setminus \sigma_{\mathrm{det}}(f) :\sum_{j=1}^k\operatorname{wind}(\lambda_j, \lambda)\neq 0 \bigg\}.
\end{equation}
This is a generalisation for the topological origin of the skin effect in the Toeplitz operator case given in \cref{sec: nonhermitian skin effect in monomer systems}.

The last part of this section is devoted to illustrating numerically the skin effect and its topological origin in chains of $2$, $3$, and $4$ periodically repeated resonators.
We start by illustrating in Figure \ref{Fig: spectrum_dimer_system} the results of a
system of $2$ periodically repeated resonators as in \cite{ammari.barandun.ea2023Perturbed}. In Figure \ref{Fig: spectrum_dimer_system}a,
we show the spectrum and pseudospectrum of the gauge capacitance matrix of a system of $25$ dimers together with the winding of the two eigenvalues of the symbol of the corresponding $2$-Toeplitz operator. Figure \ref{Fig: spectrum_dimer_system}b shows that all the eigenvectors (black eigenvectors) associated with eigenvalues inside the region $G$ in (\ref{equ:nonzerowindingregion1}) are localised and the only 
nondecreasing eigenvector (gray eigenvector) corresponds to the eigenvalue $0$ in the boundary of the region $G$. On the other hand, in \cite{ammari.barandun.ea2023Perturbed} it is observed that the nontrivial winding of the eigenvalues $\lambda_j(z)$ predicts the exponential decay of the eigenmodes. This is due to  $\operatorname{wind}(\lambda_j, \lambda)\leq 0, j=1,2$ in the example, which yields
\[
    G = \left\{ \lambda \in \mathbb{C}\setminus \sigma_{\mathrm{det}}(f) :  \bigcup_{j=1}^2 \operatorname{wind}(\lambda_j, \lambda)\neq 0\right\}.
\]

\begin{figure}[h]
    \centering
    \subfloat[][The region of nontrivial winding of the \\eigenvalues and the pseudospectrum.]{{\includegraphics[width=6.3cm]{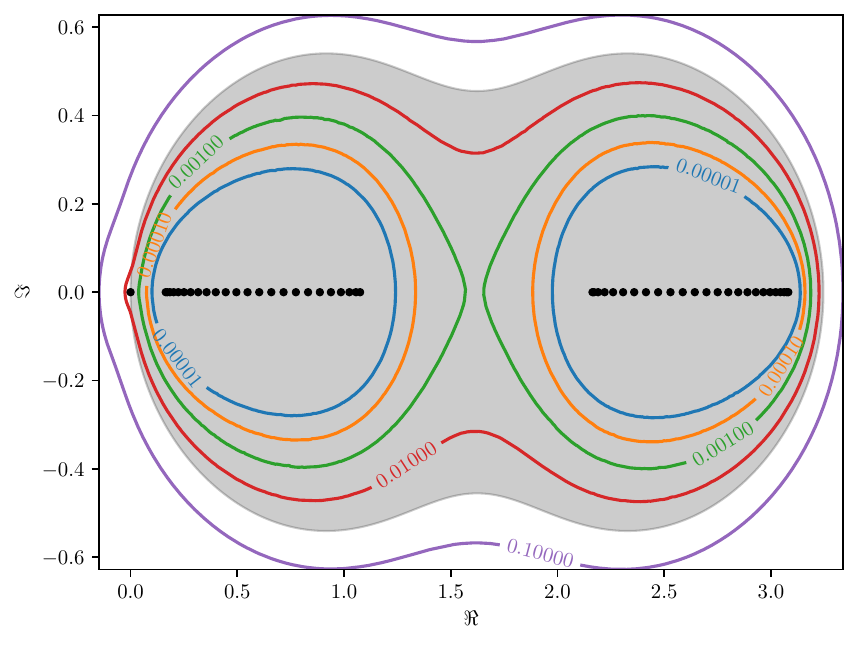} }}
    \qquad
    \subfloat[][Eigenmodes superimposed on one another to portray the skin effect. ]{{\includegraphics[width=6.3cm]{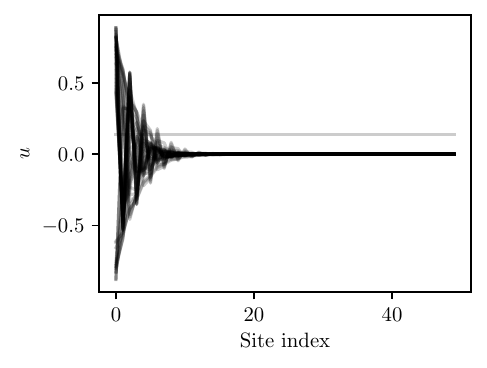} }}
    \caption{The region of $\lambda$ so that $\sum_{j=1}^k\operatorname{wind}(\lambda_j(\mathbb T), \lambda)\neq 0$ and the localisation of the eigenvectors. Computation performed for $s_1 = 1, s_2 = 2,$ and $N = 50$.}
    \label{Fig: spectrum_dimer_system}
\end{figure}

\begin{figure}[h]
    \centering
    \subfloat[][Computation performed for $s_1 = 1, s_2 = 2, s_3 = 3,$ and $N = 48$. ]{{\includegraphics[width=6.3cm]{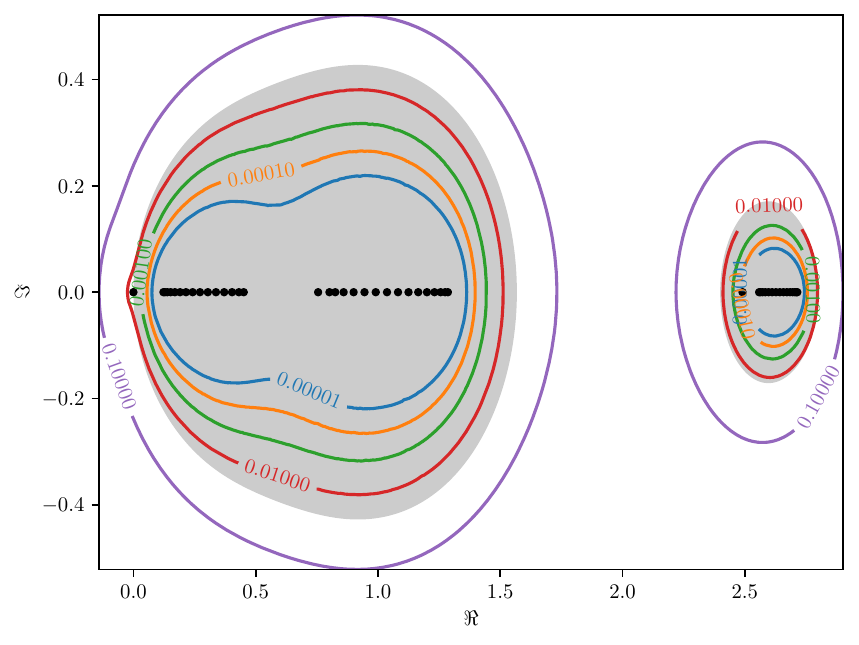} }}
    \qquad
    \subfloat[][ Simulation performed with $s_1 = 1, s_2 =2, s_3 = 3,$ and $N = 48$.]{{\includegraphics[width=6.3cm]{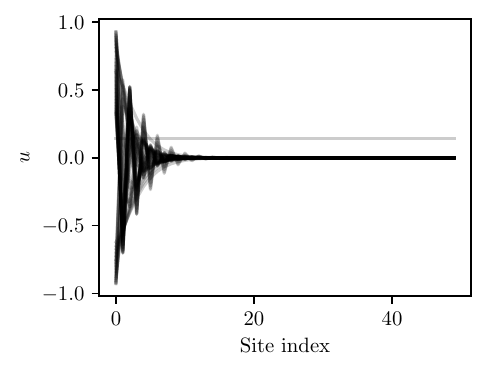} }}
    \qquad
    \subfloat[][Computation performed for $s_1 = 2, s_2 = 3, s_3 = 4,$ and $N = 52$. ]{{\includegraphics[width=6.3cm]{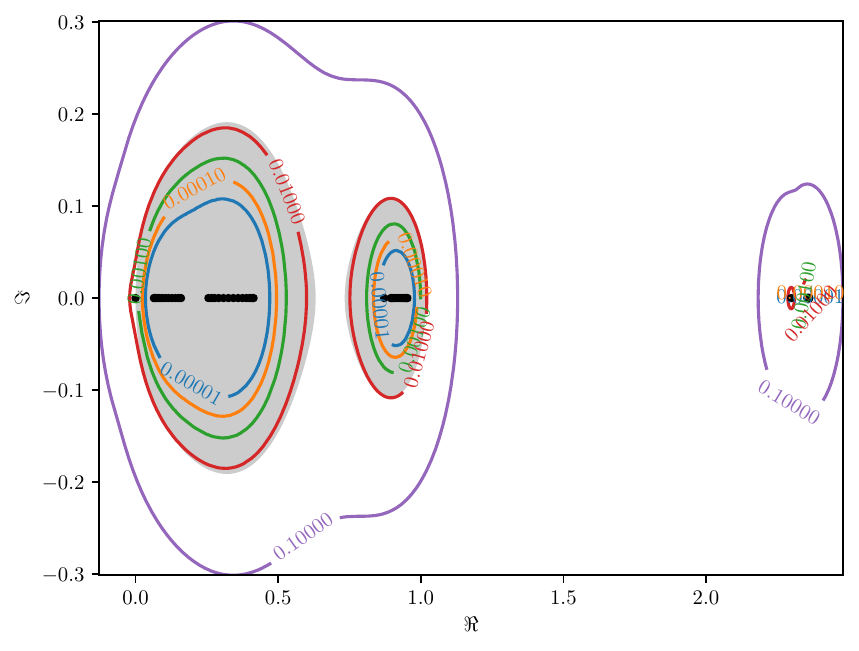} }}
    \qquad
    \subfloat[][Simulation performed with $s_1 = 2, s_2 =3, s_3 = 4,$ and $N = 52$.]{{\includegraphics[width=6.3cm]{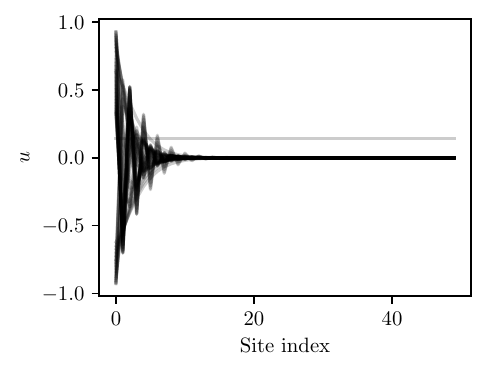} }}
    \caption{Figures A and C show the spectrum of the operator. The green regions consist of all the eigenvalues $\lambda$ that satisfy $\sum_{j=1}^3 \operatorname{wind}(\lambda_j, \lambda)\neq 0$. The black dots along the real line denote the spectrum of the gauge capacitance matrix $C^\gamma$ and the solid blue and orange lines around the spectrum are the $\varepsilon$-pseudospectra for $\varepsilon = 10^{k}$ and $k = -5, -2$. Figures B and D show the eigenvectors of $C^{\gamma}$.}
    \label{Fig: spectrum_trimer_system}
\end{figure}

Figure \ref{Fig: spectrum_trimer_system} illustrates the results for $3$ and $4$ periodically repeated resonators. We numerically verify that indeed all the eigenvectors, except the one associated with eigenvalue $0$ on the boundary of $G$, are localised at the left edge of the structure. We remark that due to the higher freedom in the parameter choice, for $n\geq 3$ periodically repeated resonators, various topologies are achievable for the area $G$ and the one presented in \cref{Fig: spectrum_trimer_system} is only one of the possibilities.

\subsection{Non-Hermitian interface modes between opposing signs of the gauge potential}\label{sec: double gamma}
We now briefly consider interface modes between two structures where the sign of $\gamma$ is switched from negative (on the left part) to positive (on the right part). Most commonly, localised interface modes are formed by creating a defect in the system's geometric periodicity (see, for example, \cite{ammari.davies.ea2020Topologically}). In non-Hermitian systems based on complex material parameters, similar localised interface modes have been shown to exist in the presence of a defect in the periodicity of the material parameters \cite{ammari.barandun.ea2023Edge,ammari.hiltunen2020Edge}. Given the existence of the skin effect, demonstrated in the previous section, it is reasonable to expect that we might be able to produce a similar localisation effect using systems of resonators with imaginary gauge potentials. With this in mind, we consider the following system of $N=2n+1$ resonators:

\begin{align}
    \begin{dcases}
        u\prii(x) \bm+ \gamma u\pri(x)+\frac{\omega^2}{v_b^2}u=0,                                                                  & x\in\bigcup_{i=1}^{n}(x_i^{\iL},x_i^{\iR}),          \\
        u\prii(x) \bm- \gamma u\pri(x)+\frac{\omega^2}{v_b^2}u=0,                                                                  & x\in\bigcup_{n+1}^N(x_i^{\iL},x_i^{\iR}),            \\
        u\prii(x) + \frac{\omega^2}{v^2}u=0,                                                                                       & x\in\R\setminus\bigcup_{i=1}^N(x_i^{\iL},x_i^{\iR}), \\
        u\vert_{\iR}(x^{\iLR}_i) - u\vert_{\iL}(x^{\iLR}_i) = 0,                                                                   & \text{for all } 1\leq i\leq N,                       \\
        \left.\frac{\dd u}{\dd x}\right\vert_{\iR}(x^{\iL}_{{i}})=\delta\left.\frac{\dd u}{\dd x}\right\vert_{\iL}(x^{\iL}_{{i}}), & \text{for all } 1\leq i\leq N,                       \\
        \left.\frac{\dd u}{\dd x}\right\vert_{\iR}(x^{\iR}_{{i}})=\delta\left.\frac{\dd u}{\dd x}\right\vert_{\iR}(x^{\iL}_{{i}}), & \text{for all } 1\leq i\leq N,                       \\
        \frac{\dd u}{\dd\ \abs{x}}u -\i k u = 0,                                                                                   & x\in(-\infty,x_1^{\iL})\cup (x_N^{\iR},\infty).
    \end{dcases}
    \label{eq:coupled ods different gammas}
\end{align}

It is not difficult to see that also with this system we can recover a capacitance matrix for which a similar result as the one in \cref{cor: approx via eva eve} holds. In particular, generalising \eqref{eq: def cap mat}, we get
\begin{align*}
    \capmat_{i,j}^{\gamma,-\gamma} =
    \begin{dcases}
        \capmat_{i,j}^{\gamma},\quad i \leq n, \\
        \capmat_{i,j}^{-\gamma},\quad i \geq n + 1.
    \end{dcases}
\end{align*}

The decay properties of the eigenvectors \eqref{eq: decay for eigemodes} and the symmetry property with respect to $\gamma$ show that this symmetric system \eqref{eq:coupled ods different gammas} has all but two modes localised at the interface. These interface modes are shown in \cref{fig: double_gamma}, superimposed on one another to portray the general trend.

\begin{figure}[htb]

    \centering
    \includegraphics[width=0.5\textwidth]{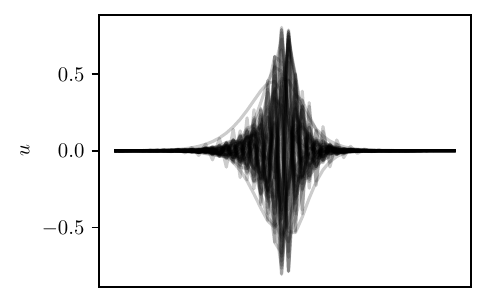}

    \caption{Plot of all the eigenmodes localised at the interface associated to a system described by \eqref{eq:coupled ods different gammas}. The $x$-axis encodes the site index of the resonators. The simulation is performed with a structure of $N=50$ resonators, $\ell=s=1$ and $\gamma=1$. The two trivial eigenmodes are not shown.}
    \label{fig: double_gamma}
\end{figure}

The case presented here is the choice of $\gamma$'s that generates localised mode at the middle of the structure as do defected Hermitian interface structures in \cref{sec: interface modes}. In the next section, we consider the other case, which generates localisation at the edges of the system.

\section{Tunable localisation in parity-time symmetric systems}\label{sec: PT symm EP}

The aim of this section is to consider a mirrored system with two imaginary gauge potentials (opposite to each other as presented briefly in \cref{sec: double gamma}) and study the phase change of the spectrum from purely real to complex when gain and loss are introduced in a balanced way into the system as a function of the gain to loss ratio. This ensures that parity--time  ($\mathcal{PT}$-) symmetry is preserved as this ratio is increased. Crucially, the parity--time symmetry of the system is reflected in the gauge capacitance matrix $C^\gamma$, ensuring that it is \emph{pseudo--Hermitian}, that is there exists some invertible self-adjoint matrix $M$ so that the adjoint $(C^\gamma)^*$ of $C^\gamma$ is given by $(C^\gamma)^*=M C^\gamma M^{-1}$.
We observe that the eigenmodes of the parity--time symmetric system decouple when going through an exceptional point. Tuning the gain-to-loss ratio, we change the system from a phase with unbroken parity-time symmetry to a phase with broken parity-time symmetry where the condensed eigenmodes at one edge are decoupled from the ones at the opposite edge of the structure. To understand this behaviour we extend the standard Toeplitz theory to encompass symmetrical parameter changes across an interface. We show that the intrinsic nature of this switch from unbroken to broken $\mathcal{PT}$-symmetry is due to a change in the topological nature of the mode.
Furthermore, we are able to show that as the number of resonators is increased, the amount of tuning required for exceptional points and the corresponding decoupling to occur goes to zero. This leads to an increasingly dense concentration of exceptional points. The results in this section are from \cite{ammari.barandun.ea2024Tunable}. 

We will only consider systems of equally spaced identical resonators, that is,
\begin{align*}
    \ell_i = \ell \in \R_{>0}\text{ for all } 1\leq i\leq N \quad \text{and} \quad s_i = s \in \R_{>0}  \text{ for all } 1\leq i\leq N-1,
\end{align*}

and apply an imaginary gauge potential as illustrated in \cref{fig:setting_PT}.
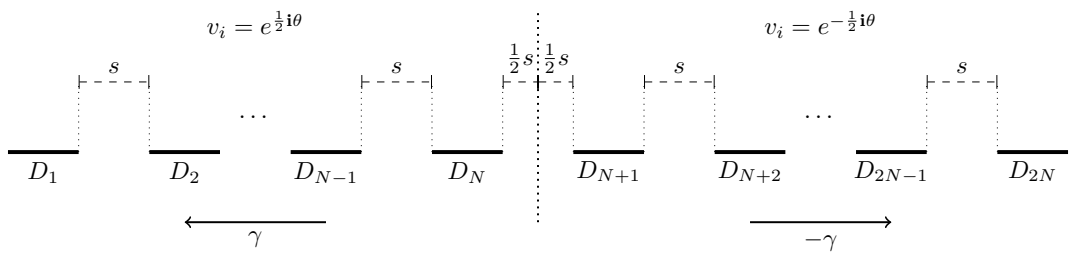
\begin{figure}[htb]
    \centering
    \begin{adjustbox}{width=\textwidth}
        \begin{tikzpicture}

            \begin{scope}[shift={(-7.5,0)}]
                \draw[ultra thick] (0,0) -- node[below]{$D_{1}$} (1,0);

                \draw[-,dotted] (1,0) -- (1,1);
                \draw[|-|,dashed] (1,1) -- node[above]{$s$} (2,1);
                \draw[-,dotted] (2,0) -- (2,1);
            \end{scope}

            \begin{scope}[shift={(-5.5,0)}]
                \draw[ultra thick] (0,0) -- node[below] {$D_{2}$} (1,0);
                \node at (1.5,.5) {\dots};
            \end{scope}

            \begin{scope}[shift={(-3.5,0)}]
                \draw[ultra thick] (0,0) -- node[below]{$D_{N-1}$} (1,0);

                \draw[-,dotted] (1,0) -- (1,1);
                \draw[|-|,dashed] (1,1) -- node[above]{$s$} (2,1);
                \draw[-,dotted] (2,0) -- (2,1);
            \end{scope}

            \begin{scope}[shift={(-1.5,0)}]
                \draw[ultra thick] (0,0) -- node[below] {$D_{N}$} (1,0);
                \draw[-,dotted] (1,0) -- (1,1);
                \draw[|-|,dashed] (1,1) -- node[above]{$\frac{1}{2}s$} (1.5,1);
            \end{scope}

            \draw[->,thick] (-3,-1) -- node[below]{$\gamma$} (-5,-1);
            \node[above] at (-4,1.5) {$v_i = e^{\frac{1}{2}\i\theta}$};

            \draw[-,thick,dotted] (0,-1) -- (0,2);

            \begin{scope}[shift={(-0.5,0)}]
                \draw[|-|,dashed] (0.5,1) -- node[above]{$\frac{1}{2}s$} (1,1);
                \draw[ultra thick] (1,0) -- (2,0);
                \node[below] at (1.5,0) {$D_{N+1}$};
                \draw[-,dotted] (1,0) -- (1,1);
            \end{scope}

            \begin{scope}[shift={(+1.5,0)}]
                \draw[-,dotted] (0,0) -- (0,1);
                \draw[|-|,dashed] (0,1) -- (1,1);
                \node[above] at (0.5,1) {$s$};
                \draw[ultra thick] (1,0) -- (2,0);
                \node[below] at (1.5,0) {$D_{N+2}$};
                \draw[-,dotted] (1,0) -- (1,1);
                \node at (2.5,.5) {\dots};
            \end{scope}

            \begin{scope}[shift={(+3.5,0)}]
                \draw[ultra thick] (1,0) -- (2,0);
                \node[below] at (1.5,0) {$D_{2N-1}$};
            \end{scope}

            \begin{scope}[shift={(+5.5,0)}]
                \draw[-,dotted] (0,0) -- (0,1);
                \draw[|-|,dashed] (0,1) -- (1,1);
                \node[above] at (0.5,1) {$s$};
                \draw[ultra thick] (1,0) -- (2,0);
                \node[below] at (1.5,0) {$D_{2N}$};
                \draw[-,dotted] (1,0) -- (1,1);
            \end{scope}

            \draw[->,thick] (3,-1) -- node[below]{$-\gamma$} (5,-1);
            \node[above] at (4,1.5) {$v_i = e^{-\frac{1}{2}\i\theta}$};
        \end{tikzpicture}
    \end{adjustbox}
    \caption{A chain of $2N$ one-dimensional identical and equally spaced resonators. Material parameters and sign of the imaginary gauge potentials depend on the resonator's position.}
    \label{fig:setting_PT}
\end{figure}
On one side, the gauge capacitance matrix is given by
\begin{gather}
    \label{eq:cdef}
    C^\gamma =
    \left(\begin{array}{ccccc|ccccc}
            \alpha + \beta & \eta   &        &        &        &        &        &        &        &              \\
            \beta          & \alpha & \ddots &        &        &        &        &        &        &              \\
                           & \ddots & \ddots &        &        &        &        &        &        &              \\
                           &        &        & \alpha & \eta   &        &        &        &        &              \\
                           &        &        & \beta  & \alpha & \eta   &        &        &        &              \\
            \hline
                           &        &        &        & \eta   & \alpha & \beta                                   \\
                           &        &        &        &        & \eta   & \alpha & \ddots                         \\
                           &        &        &        &        &        & \ddots & \ddots                         \\
                           &        &        &        &        &        &        &        & \alpha & \beta        \\
                           &        &        &        &        &        &        &        & \eta   & \alpha+\beta
        \end{array}\right) \in \R^{2N\times 2N}
\end{gather}
with
\begin{equation}\label{equ:alphabetagamma}
    \alpha = \frac{\gamma}{1-e^{-\gamma}} - \frac{\gamma}{1-e^{\gamma}} = \gamma\coth(\gamma/2),\quad \eta = \frac{-\gamma}{1-e^{-\gamma}},\quad \beta = \frac{\gamma}{1-e^{\gamma}},
\end{equation}
because of the sign change of the imaginary gauge potential. On the other side, we have to model the complex (and varying) material parameters. Thus, we consider the \emph{generalised gauge capacitance matrix}
\begin{align}
    \label{eq:cgdef}
    C^{\theta,\gamma} = V^\theta C^\gamma \quad \text{ with } \quad V^\theta =
    \left(\begin{array}{c|c}
                  e^{\i\theta}I_{N} & \mathbf{0}         \\
                  \hline
                  \mathbf{0}        & e^{-\i\theta}I_{N}
              \end{array}\right) \in \C^{2N\times 2N}.
\end{align}
The same result as the one stated in \cref{cor: approx via eva eve} holds for the system described by \cref{fig:setting_PT} when considering the generalised gauge capacitance matrix from \eqref{eq:cgdef} (generalising the proof presented in \cite{ammari.barandun.ea2024Mathematical} is easily achieved by the same procedure as the one used in \cite{ammari.barandun.ea2023Edge}). Throughout this section, $C^{\theta,\gamma}$ and $C^{\gamma}$ are $2N\times 2N$ matrices.

This section will extensively study \emph{non-diagonalisability} of $\cg$.
\begin{definition}
    A setup for which $\cg$ is \emph{not} diagonalisable is called an \emph{exceptional point}.
\end{definition}
Specifically, we will study setups where the geometry and the imaginary gauge potentials remain fixed and the material parameters (here modelled by $\theta$) lay in a specific range.

\subsection{Properties of the generalised gauge capacitance matrix}
Let $P\in \R^{2N\times 2N}$ be the anti-diagonal involution, \emph{i.e.},  $P_{ij} = \delta_{i,2N-i+1}$ and $D^\gamma=\diag(1,e^\gamma,\dots,(e^{\gamma})^{N-1},(e^{\gamma})^{N-1},\dots,e^\gamma,1)\in \R^{2N\times 2N}$. We refer, for instance,  to  \cite{ammari.barandun.ea2024Tunable} for the precise definitions of pseudo--Hermitian and quasi-Hermitian matrices. Here and elsewhere in this section, $M^*$ denotes the adjoint of $M$: $(M^*)_{i,j}=\overline{M_{j,i}}$.

\begin{proposition}\label{prop:csymm}
    The generalised gauge capacitance matrix has the following symmetry:
    \begin{equation}\label{eq:ptsymm}
        P\cg P = \overline{\cg}.
    \end{equation}
    For the unmodified gauge capacitance matrix $C^\gamma$, we have
    \begin{equation}\label{eq:csymm}
        PC^\gamma P = C^\gamma.
    \end{equation}
\end{proposition}

\begin{proposition}
    Let $M^{\theta,\gamma} = PV^\theta D^\gamma$.
    Then, $M^{\theta,\gamma}$ is invertible and Hermitian and we have
    \begin{equation}
        M^{\theta,\gamma} \cg = (\cg)^*M^{\theta,\gamma}.
    \end{equation}
\end{proposition}
Since $\cg$ as in \eqref{eq:cgdef} is pseudo--Hermitian, its spectrum must be invariant under complex conjugation \cite{ammari.barandun.ea2024Tunable}, that is,
\[\sigma(\cg) = \overline{\sigma(\cg)}.\]

For the case $\theta=0$, the matrix $\mathnormal{C}^{0,\gamma} = \cm$ satisfies an even stronger notion of Hermiticity.
\begin{proposition}\label{prop:cquasiherm}
    Let $\cm=\mathnormal{C}^{\theta=0,\gamma}$ as in equation (\ref{eq:cdef}). Then, $\cm$ is quasi-Hermitian with \emph{metric operator} $D$, that is,
    \begin{equation}
        D^{-1}\cm  = (\cm)^*D^{-1}.
    \end{equation}
\end{proposition}
As a quasi-Hermitian matrix, $\cm$ is diagonalisable with real spectrum.

Finally, we characterise the kernel of $\cg$.
\begin{lemma}
    For any $\gamma>0$ and $\theta \in [0,2\pi)$, we have $(1,\dots,1)^\top \in \ker \cg\subset \R^{2N}$.
\end{lemma}

As a consequence of \cref{thm:spectrarepresentation1}, the eigenspaces of $\cg$ are always one-dimensional. Consequently, the kernel of $\cg$ is also one-dimensional and is exactly the span of $\bm 1 = (1,\dots,1)^\top$.

The fact that $\cg$ is tridiagonal allows us to determine the eigenvectors recursively. The symmetry of $\cg$ across the interface in the middle will yield very similar forms for the first and second half of its eigenvectors. The eigenvalues will then be characterised by a compatibility condition across this interface.

\begin{proposition}\label{thm:spectrarepresentation1}
    Let the affine transformation $\mu^\theta:\C\to\C$ be defined by
    \begin{align}\label{eq:mudef}
        \mu^\theta(\lambda)\coloneqq \frac{e^{-\i\theta}\lambda-\alpha}{2\sqrt{\beta\eta}}=e^{-\i\theta}\lambda\frac{1}{\gamma}\sinh \frac{\gamma}{2} - \cosh\frac{\gamma}{2}.
    \end{align}
    For $\lambda\in \C$ an eigenvalue of $\cg$, the corresponding eigenvector is given by $\bm v = (\bm x, \bm y)^{\top}$, where
    \begin{equation}\label{eq:evecform}
        \begin{aligned}
             & \bm x=\left(P_0(\mu^\theta(\lambda)), \left(e^{-\frac{\gamma}{2}}\right)P_1(\mu^\theta(\lambda)),  \cdots,  \left(e^{-\frac{\gamma}{2}}\right)^{N-1} P_{N-1}(\mu^\theta(\lambda)) \right)^{\top},            \\
             & \bm y = C_1\left(\left(e^{-\frac{\gamma}{2}}\right)^{N-1} P_{N-1}(\mu^{-\theta}(\lambda)), \cdots,  \left(e^{-\frac{\gamma}{2}}\right)P_1(\mu^{-\theta}(\lambda)), P_0(\mu^{-\theta}(\lambda)) \right)^{\top}.
        \end{aligned}
    \end{equation}
    Here, $P_n (x) = U_n(x) + e^{-\frac{\gamma}{2}}U_{n-1}(x)$ is the sum of two Chebyshev polynomials of the second kind, with $P_0 = 1$. Specifically, $U_{n+1}(x)\coloneqq 2xU_n(x)-U_{n-1}(x)$ for $n\geq 1 $ with $U_0(x)=1$ and $U_1(x)=2x$.
    Furthermore, we have
    \begin{equation}\label{eq:cdef2}
        C_1= e^{-\frac{\gamma}{2}}\frac{P_N(\mu^{\theta}(\lambda))}{P_{N-1}(\mu^{-\theta}(\lambda))}= e^{\frac{\gamma}{2}}\frac{P_{N-1}(\mu^{\theta}(\lambda))}{P_N(\mu^{-\theta}(\lambda))},
    \end{equation}
    which yields the following characterisation of the spectrum of $\cg$:
    \begin{equation}\label{eq:evalconstraint}
        \frac{P_{N}(\mu^{\theta}(\lambda))P_{N}(\mu^{-\theta}(\lambda))}{P_{N-1}(\mu^{\theta}(\lambda))P_{N-1}(\mu^{-\theta}(\lambda))} = e^\gamma.
    \end{equation}
    Namely, $\lambda\in \C$ is an eigenvalue of $\cg$ if and only if it satisfies \eqref{eq:evalconstraint}. Moreover, its corresponding eigenspace is always one-dimensional.
\end{proposition}

These final two facts immediately yield the following characterisation for the exceptional points of $\cg$.
\begin{corollary}\label{cor:epnondestinct}
    An exceptional point occurs when \eqref{eq:evalconstraint} has less than $2N$ distinct solutions.
\end{corollary}

We aim to use this corollary to show that for a given $\theta>0$, any nonreal eigenvalue $\lambda\in \C\setminus\R$ of $\cg$ must have passed through an exceptional point. However, we must first formalise the notion of an eigenvalue ``having passed through'' an exceptional point. To that end, we would like to associate each eigenvalue $\lambda_i$ of $\cg$ with some corresponding continuous path $\lambda_i(\theta)$ such that $\lambda_i(\theta)$ is an eigenvalue of $\cg$ for all values of $\theta$ and all $i=1,\dots, 2N$.
However, precisely because exceptional points occur, we cannot choose these paths in a canonical fashion, as at these exceptional points, two eigenvalue paths $\lambda_i(\theta)$ and $\lambda_j(\theta)$, $i\neq j$, meet and cannot be distinguished.
What we can do is the following: Let $0<\theta'<\frac{\pi}{2}$ be fixed. For any simple eigenvalue $\lambda$ of $\mathnormal{C}^{\theta',\gamma}$, we can then define the maximal unique continuous  \emph{eigenvalue path} $\lambda : \theta\in [\theta_0,\theta']\to \C$ such that $\lambda(\theta)$ is always an eigenvalue of $\cg$ for any $\theta\in [\theta_0,\theta']$ and $\lambda(\theta') = \lambda$. $\theta_0$ is chosen to be either the largest $\theta<\theta'$ such that $\lambda(\theta)$ is an exceptional point, or zero - whichever is greater.

For some $0<\theta'<\frac{\pi}{2}$ and  $\lambda\in\C$ eigenvalue of $\mathnormal{C}^{\theta',\gamma}$, we can then say that $\lambda$ \emph{has passed through an exceptional point} if and only if $\theta_0$ is greater than zero. Note also that because $C^\gamma$ is diagonalisable, $\lambda(0)$ is never an exceptional point.

We can now state the following result.
\begin{corollary}\label{cor:eppassthrough}
    Let $0<\theta'<\frac{\pi}{2}$ and let $\lambda\in\C$ be an eigenvalue of $\cg$. If $\lambda$ is in $\C\setminus \R$, it must have passed through an exceptional point.
\end{corollary}
This corollary is very useful because it allows us to prove the existence of exceptional points merely from the fact that some eigenvalues are nonreal.

The final result of this section characterises the relation between the factors $C_\lambda$ and $C_{\overline{\lambda}}$ for complex conjugate pairs $\lambda, \overline{\lambda}$ of $\cg$.
\begin{corollary}\label{prop:cconj}
    Let $\lambda, \overline{\lambda}, \in \C$ be a pair of eigenvalues of $\cg$ and let $C_\lambda, C_{\overline{\lambda}}$ be the corresponding factors as defined in (\ref{eq:cdef2}). Then, we have
    \begin{equation}
        \overline{C_\lambda}C_{\overline{\lambda}}=1.
    \end{equation}
    In particular, we have $\lvert C_\lambda \rvert =1$ for $\lambda\in \R$.
\end{corollary}

\subsection{Eigenvalue of the generalised gauge capacitance matrix}\label{sec:eva}
In this section, we study the eigenvalues of the generalised gauge capacitance matrix $\cg$ for the system described in the previous section. We prove that for small $\theta$ all the eigenvalues of $\cg$ are real, which corresponds to the coupled regime. We also prove existence and density of exceptional points. Finally,  we  approximate the locations of eigenvalues. Understanding the movement of eigenvalues will prove to be a crucial prerequisite to understand the decoupling behaviour of the eigenvectors in the next section.

\subsubsection{Coupled regime}\label{ssec:evacoupled}
This subsection is dedicated to the case where all eigenvalues of $\cg$ are real. We will show that for small $\theta$ the eigenvalues behave similarly to the ones of the gauge capacitance matrix $\capmat^{\theta=0,\gamma}$. Consequently, as we will show in \cref{sec:eves}, also the eigenvectors of $\cg$ will have a  form similar to that of $\capmat^{\theta=0,\gamma}$.
\begin{proposition}
    For any $N\in \N$ and $\gamma>0$, there exists an $\varepsilon>0$ such that for $0\leq \theta<\varepsilon$ all the eigenvalues of $\cg$ are real. For a real eigenvalue $\lambda$ of $\cg$, the eigenvector $\bm = (\bm x, \bm y)^\top$ decomposed as in \cref{thm:spectrarepresentation1}, has the following symmetry:
    \begin{equation}
        \bm y = e^{\i\phi}P\overline{\bm x}
    \end{equation}
    for some $\phi \in [0,2\pi)$. In particular, we have $\abs{\bm x^{(j)}} = \abs{\bm y^{(2N + 1 - j)}}$ for $j=1,\dots N$.
\end{proposition}

\subsubsection{Existence of exceptional points}\label{ssec:epexistence}
In this subsection, our aim is to show that regardless of $N$ and $\gamma$, all eigenvalues must pass through an exceptional point as $\theta$ is increased from $0$ to $\frac{\pi}{2}$.

\begin{theorem}\label{thm::all_eigs_excpetional_point}
    Let $\gamma>0$ and $N\in\N$. Then, all but two eigenvalues $\lambda\in \C$ of $\cg$ for $\theta=\frac{\pi}{2}$ must have passed through an exceptional point. The remaining two eigenvalues experience an exceptional point at $\theta=\frac{\pi}{2}$.
\end{theorem}

In line with \cref{cor:eppassthrough}, in order to show that an eigenvalue $\lambda$ of $\cg$ for $\theta=\frac{\pi}{2}$ has passed through an exceptional point, it is sufficient to show that $\lambda$ lies in $\C\setminus\R$. Indeed, for $\theta=\frac{\pi}{2}$ we will show that the spectrum of $\cg$ lies entirely on the imaginary axis. All nonzero eigenvalues must thus have gone through an exceptional point. Two eigenvalues will turn out to be zero, yielding another exceptional point exactly at $\theta=\frac{\pi}{2}$.

We will need a good understanding of  the zeros of Chebyshev polynomials.
\begin{proposition}\label{prop:Szeros}
    $\mathcal{S}(\lambda)$ has exactly $(2N-2)$ distinct real zeros and a double zero $\lambda=0$.
\end{proposition}
The main idea of the proof is to exploit the heavily interlaced nature of the Chebyshev polynomials, which will prove to be a robust source of zeros of their composites. This will allow us to guarantee and bound $N$ real zeros for $P_N$ and $(N-1)$ real zeros for $P_N+P_{N-1}$. The evenness of $\mathcal{S}(\lambda)$ will then allow us to use these results to guarantee zeros of $\mathcal{S}$ as well.

We can then combine the arguments of this subsection to prove \cref{thm::all_eigs_excpetional_point}.
\begin{proof}[Proof of \cref{thm::all_eigs_excpetional_point}]
    For $\theta=0$, all the eigenvalues are real and for $\theta=\frac{\pi}{2}$ all but two eigenvalues are purely imaginary by \cref{prop:Szeros}. The two not purely imaginary eigenvalues are both zero, causing an exceptional point by \cref{cor:epnondestinct}. By \cref{cor:eppassthrough}, all the purely imaginary, nonzero eigenvalues must have passed through an exceptional point by $\theta=\frac{\pi}{2}$. Furthermore, these eigenvalues are distinct, which ensures that they passed through that exceptional point \emph{before} $\theta=\frac{\pi}{2}$.
\end{proof}

\subsubsection{Asymptotic density of exceptional points}\label{ssec:epdensity}

\begin{figure}
    \centering
    \includegraphics[width=0.5\textwidth]{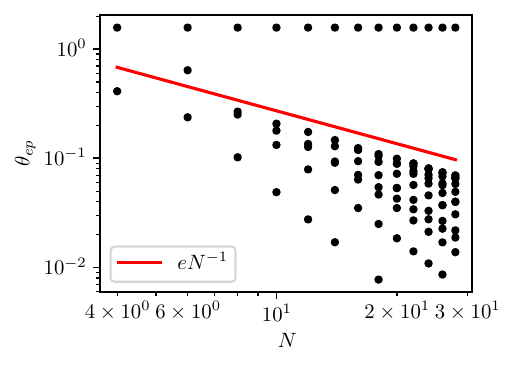}
    \caption{Distribution of the exceptional points for varying $N$. For any $N$, the system exhibits a trivial exceptional point at $\theta=\frac{\pi}{2}$. All other exceptional points concentrate in the interval $[0,e/N]$ and become increasingly dense as $N$ grows.}
    \label{fig:epdependence}
\end{figure}
We are now interested in showing that exceptional points do not only occur (as shown in the previous subsection) but also cluster creating a parameter region with high density of such points. Many of the results developed in this subsection will also be used in \cref{ssec:evalocations} and \cref{sec:eves} as they enable the asymptotic characterisation of the eigenvalue locations and eigenvector growth.
\begin{theorem}\label{thm:asymptotic density}
    Let $0<\theta<\frac{\pi}{2}$ and $\gamma>0$ be fixed. Then, there exists an $N_0\in \N$ such that for every $N\geq N_0$, the corresponding $\cg$ has exactly two real eigenvalues.
\end{theorem}
By \cref{cor:eppassthrough}, this ensures that all other $2N-2$ eigenvalues in $\C\setminus\R$ must have already passed through an exceptional point before $\theta$.

The building blocks for this result start with a helpful reformulation of the characterisation \eqref{eq:evalconstraint} for real eigenvalues.

\begin{proposition}\label{prop: real eig means on level set of ratio of poly}
    $\lambda\in \R$ is a real eigenvalue of $\cg$ if and only if
    \begin{equation}\label{eq:realconstraint}
        \left\vert\frac{P_N(\mu^\theta(\lambda))}{P_{N-1}(\mu^\theta(\lambda))}\right\vert = e^\frac{\gamma}{2}.
    \end{equation}
\end{proposition}

The transformation $\mu^\theta(\lambda)= e^{-\i\theta}\lambda\frac{1}{\gamma}\sinh\frac{\gamma}{2}-\cosh\frac{\gamma}{2}$ maps the real line $\R$ onto a line $\mu^\theta(\R)$ in $\C$ rotated by $-\theta$ around the point $-\cosh \frac{\gamma}{2}$. This provides a very geometric view of the zeros of equation (\ref{eq:realconstraint}). In fact, the real spectrum of $\cg$ corresponds to intersections of the line $\mu^\theta(\R)$ and a level set of $\mu \mapsto \abs{\pnr{\mu}}$:
\[
    \sigma(\cg) \cap \R = \mu^\theta(\R) \cap \left\{\mu\in \C :\, \abs{\pnr{\mu}} = e^\frac{\gamma}{2}\right\}.
\]
\begin{figure}[h]
    \centering
    \begin{subfigure}[t]{0.48\textwidth}
        \centering
        \includegraphics[height=0.7\textwidth]{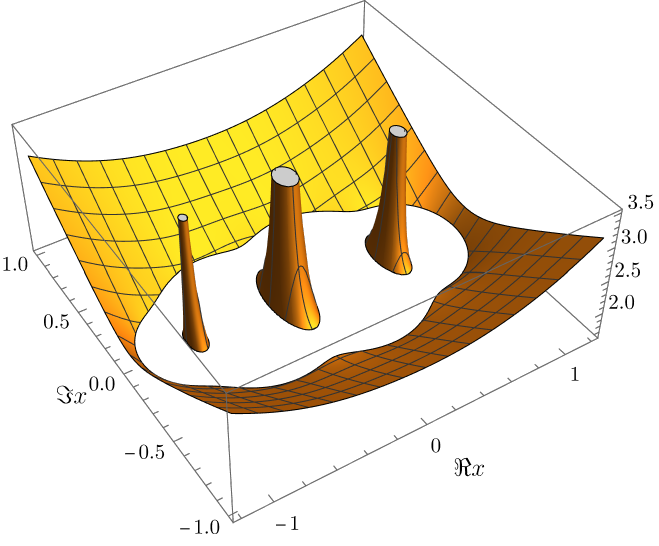}
        \caption{Graph of $\C\ni\mu\mapsto\abs{\pnr{\mu}}$ cut off by the plane $\C\times\{e^\frac{\gamma}{2}\}$.}
        \label{fig: surfaceplot intersect plane}
    \end{subfigure}
    \hfill
    \begin{subfigure}[t]{0.48\textwidth}
        \centering
        \includegraphics[height=0.7\textwidth]{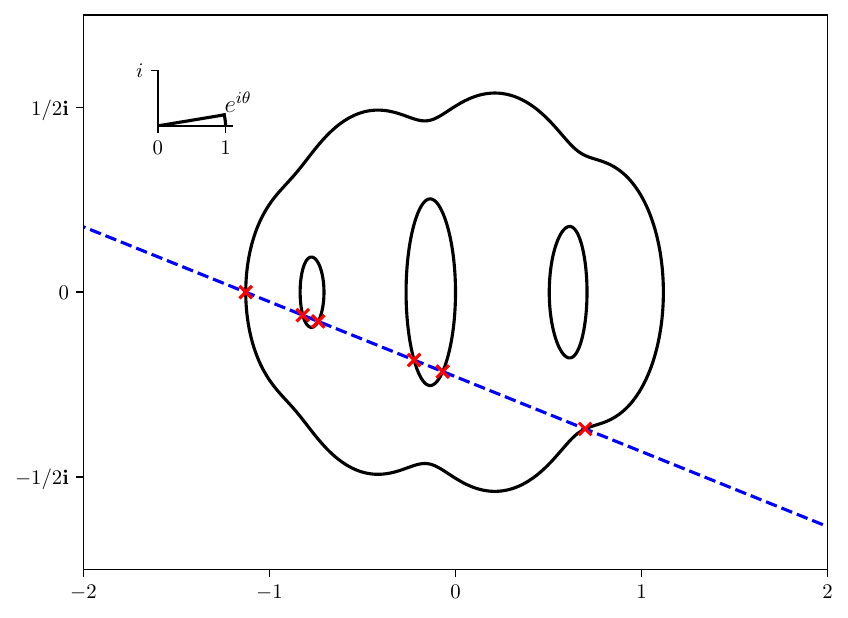}
        \caption{Intersection of $\mu^\theta(\R)$ (in blue dashed) with the level set $\{\mu\in \C :\, \abs{\pnr{\mu}} = e^\frac{\gamma}{2}\}$ (in black). The latter can be seen in \cref{sub@fig: surfaceplot intersect plane}. Intersection points are shown in red. The preimage of these are the real eigenvalues of $\cg$.}
        \label{fig: intersection of levelset with mu}
    \end{subfigure}
    \caption{Geometrical interpretation of the eigenvalues of $\cg$ as given by \cref{prop: real eig means on level set of ratio of poly}. In this view, we can also clearly see the exceptional points, where two real eigenvalues meet and become complex. Namely, this happens exactly when $\mu^\theta(\R)$ goes from passing through one of the inner regions in (B) to moving past them and two red crosses meet.}
    \label{fig: real eigenvalues are points of level set}
\end{figure}
Thus, in order to understand the real eigenvalues of $\cg$, it is crucial to understand the level sets of $\abs{\pnr{\mu}}$. We begin by recalling a well-known equivalent definition of the Chebyshev polynomials:
\[
    U_n(\mu) = \frac{a(\mu)^{n+1}-a(\mu)^{-(n+1)}}{2\sqrt{\mu+1}\sqrt{\mu-1}},
\] where $a(\mu)=\mu + \sqrt{\mu+1}\sqrt{\mu-1}$.

\subsubsection{Eigenvalue locations}\label{ssec:evalocations}
In this subsection, we aim to understand the position of the eigenvalues in the complex plane. This will prove crucial in understanding the behaviour of the eigenvectors in \cref{sec:eves}. As we will observe, for a fixed $\theta$ and increasing $N$, they move arbitrarily close to the two line segments $(\mu^\theta)^{-1}([-1,1])\cup (\mu^{-\theta})^{-1}([-1,1])$.

The following result holds.
\begin{proposition}\label{prop:evalocation}
    Let $0< \theta < \pi/2$ and let $\gamma>0$ be fixed. For any $\varepsilon>0$ small enough, there exists an $N_0\in \N$ such that for any $N\geq N_0$, all but exactly two eigenvalues of $\cg$ lie in an $\varepsilon$-neighbourhood of $K\coloneqq(\mu^\theta)^{-1}([-1,1])\cup (\mu^{-\theta})^{-1}([-1,1])$. Indeed, because of the conjugation symmetry of eigenvalues, we have
    \begin{equation}
        \left|\sigma(\cg)\cap B_\varepsilon((\mu^{\upsigma\cdot\theta})^{-1}([-1,1])) \right| = N-1,
    \end{equation}
    for $\upsigma =\pm 1$.
\end{proposition}

\begin{figure}[ht]
    \centering
    \includegraphics[width=0.5\textwidth]{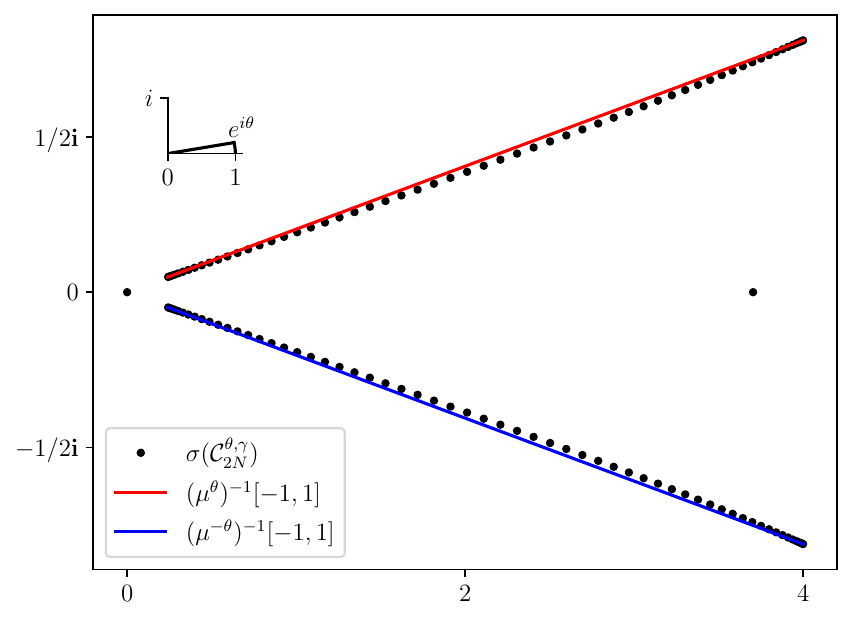}
    \caption{Eigenvalue locations close to the two line segments $(\mu^\theta)^{-1}([-1,1])\cup (\mu^{-\theta})^{-1}([-1,1])$ for $\theta = 0.2$, $\gamma = 1$ and $N=60$.}
    \label{fig:eigenvalues}
\end{figure}

\subsection{Eigenvectors of the generalised gauge capacitance matrix}\label{sec:eves}
As shown in \Cref{sec: skin effect}, systems with an imaginary gauge potential are known for the presence of skin effect, \emph{i.e.}, the condensation of the eigenvectors at one edge of the system. This condensation is exponential \cite{ammari.barandun.ea2023Perturbed}. The system studied here has a much more peculiar property. The symmetric change of sign in the gauge potential implies that the condensation happens on both edges of the system for small values of $\theta$ or $N$. Nevertheless, the non-Hermiticity introduced by $\theta$ can change this symmetry. The exponential nature of the modes has been shown to be caused by the Fredholm index of the Toeplitz operator associated to the system. The system studied here presented in \cref{fig:setting_PT} does not yield a Toeplitz matrix, nevertheless we will show that the same theory can be modified to be used in this situation as well.

\subsubsection{Exponential decay and growth}
As our matrix $\cg$ is split into two parts by an interface we define the \emph{upper and lower symbols} of $\cg$ as
\begin{align}\label{eq: symbol def}
    f_\pm^\theta: \mathbb{T}^1 & \to \C\nonumber                                                               \\
    e^{\i\phi}        & \mapsto e^{\pm\i\theta}(\beta  e^{\pm\i\phi} + \alpha +  \eta e^{\mp\i\phi}).
\end{align}
We further define the \emph{upper and lower regions of topological convergence} as
\begin{align}
    \label{eq:def_E1_E2}
    E^\theta_\pm = \{z\in\C: \pm \operatorname{wind}(f_\pm^\theta,z) < 0 \},
\end{align}
where $\operatorname{wind}(f_\pm^\theta,z)$ denotes the winding number of $f_\pm^\theta$ around $z$.

These concepts are closely linked to our formalism based on Chebyshev polynomials.
The following result holds.
\begin{lemma}\label{lem:defineofE2}
    We have
    \begin{align*}
        E^\theta_\pm = \bigg\{ (\mu^{\pm\theta})^{-1}(a^{-1}(re^{\i\phi})) \text{ for } r\in[1,e^{\gamma/2}), \phi\in[0,2\pi) \bigg\}.
    \end{align*}
\end{lemma}
By \cref{lem:defineofE2} it makes thus sense to lighten the notation and use $E^\theta=E^\theta_+$ and $E^{-\theta}=E^\theta_-$.

\cref{lem:defineofE2} also justifies calling $E^{\pm\theta}$ ``regions of topological convergence''. Namely, for an eigenpair $(\lambda,\bm v)$ of $\cg$,  \cref{thm:spectrarepresentation1} gives the following form for the eigenvector:
\[
    \bm v = (\bm x^{(1)},\dots,\bm x^{(N)},\bm y^{(1)},\dots,\bm y^{(N)})^\top
\]
with $\bm x^{(j)} = (e^{-\frac{\gamma}{2}})^{j-1} P_{j-1}(\mu^\theta(\lambda))$ and $\bm y^{(j)} = (e^{-\frac{\gamma}{2}})^{N-j} P_{N-j}(\mu^{-\theta}(\lambda))$.

The rates of growth for the left and right parts of this eigenvector are then given by
\[
    \frac{\bm x^{(j+1)}}{\bm x^{(j)}} = e^{-\frac{\gamma}{2}}\frac{P_j(\mu^\theta(\lambda))}{P_{j-1}(\mu^\theta(\lambda))}, \quad  \frac{\bm y^{(j+1)}}{\bm y^{(j)}} = e^{\frac{\gamma}{2}}\frac{P_{j-1}(\mu^{-\theta}(\lambda))}{P_{j}(\mu^{-\theta}(\lambda))}.
\]

Focusing on the left part,  we notice that its asymptotic growth behaviour is determined by whether $\abs{\frac{P_j(\mu^\theta(\lambda))}{P_{j-1}(\mu^\theta(\lambda))}}$ is smaller or larger than $e^{\frac{\gamma}{2}}$. Furthermore, we have $\frac{P_j(\mu^\theta(\lambda))}{P_{j-1}(\mu^\theta(\lambda))} \to a(\mu^\theta(\lambda))$. $\bm x$ thus decays or grows asymptotically exactly if $\abs{a(\mu^\theta(\lambda))}$ is smaller or larger than $e^{\frac{\gamma}{2}}$, respectively. But, by \cref{lem:defineofE2}, we can see that $\lambda$ lies in $E^{\theta}$ if and only if $\abs{a(\mu^\theta(\lambda))}<e^{\frac{\gamma}{2}}$, justifying our naming.

For $\bm y$ the rate of growth is exactly the inverse of the rate of growth of $\bm x$ with $\mu^\theta(\lambda)$ replaced by $\mu^{-\theta}(\lambda)$. The above reasoning thus also holds for $\bm y$ with ``growth'' and ``decay'' as well as the sign of $\theta$ is flipped.

By \cref{prop:evalocation}, we know the approximate locations of the eigenvalues of $\cg$. We will now make use of that and the topological convergence to formally prove the above intuition.

\begin{theorem}\label{lemma:exponential decay and decoupling}
    Let $0<\theta < \pi/2$, $\gamma>0$ and let $N$ be large enough\footnote{Specifically so that $\varepsilon$ from \cref{prop:evalocation} is smaller than $\sqrt{\alpha+\beta}-\sqrt{\beta+\eta}$ and thus $B_\varepsilon((\mu^{\pm\theta})^{-1}([-1,1]))\subset E^{\pm\theta}$.}. Fix, furthermore, a $0<\sigma\ll1$. Consider an eigenpair $(\lambda,\bm v)$ of $\cg$ fulfilling \cref{thm:spectrarepresentation1}. Then, one of the following three cases realises:
    \begin{enumerate}
        \item[(i)] If $\lambda \in B_\sigma(\partial E^\theta\cup \partial E^{-\theta})\eqqcolon \Theta$, then either $\lambda=0$ and $\bm v = \bm 1$ or no conclusion is made;
        \item[(ii)] If $\lambda \in E^{\theta} \cap E^{-\theta}\setminus \Theta$, then there exist some $B_1,B_2,C_1,C_2>0$ independent of $N$ so that
              \begin{align*}
                  \vert \bm{v}^{(j)} \vert < C_1 e^{-B_1\frac{j\gamma}{2}} \quad \text{and}\quad\vert \bm{v}^{(2N+1-j)} \vert < C_2 e^{-B_2\frac{j\gamma}{2}},
              \end{align*}
              for $1\leq j\leq N$. In particular, if also $\lambda \in \R$, then $C_1=C_2$ and $B_1=B_2$.
        \item[(iii)] If $\lambda \in E^\theta \triangle E^{-\theta}\setminus \Theta$, then there exit some $B,C>0$ independent of $N$  so that
              \begin{align*}
                  \vert \bm{v}^{(j)} \vert < C e^{-B\frac{j\gamma}{2}} \quad \text{if }\lambda\in E^\theta, \\ \quad \vert \bm{v}^{(2N+1-j)} \vert < C e^{-B\frac{j\gamma}{2}} \quad \text{if }\lambda\in E^{-\theta},
              \end{align*}
              for $1\leq j\leq 2N$.
    \end{enumerate}
    In particular, for $\frac{\pi}{4}\leq \theta \leq \frac{\pi}{2}$,  case (2) never realises for $N$ large enough.
\end{theorem}

The implications of \cref{lemma:exponential decay and decoupling} are important as it shows that the eigenvectors of $\cg$ are not only exponentially decaying or growing but also that the non-Hermiticity introduced by $\theta$ manifests itself at a macroscopic level as a decoupling of the eigenvectors. While for $\theta=0$ the eigenvectors always present a symmetric exponential decay, the non-Hermiticity introduced by $\theta>0$ brings the eigenvalue to eventually migrate to the complex plane and out of one of the two regions $E^\theta$ or $E^{-\theta}$. As a consequence of this, the symmetry is broken. It is also interesting to notice that this process happens pairwise. Since $\cg$ is pseudo--Hermitian, the eigenvalues come in complex conjugated pairs and, as $\theta$ is varying, they meet pairwise at an exceptional point. After the exceptional points, one of the eigenvectors will be decaying while the other will be increasing. The decoupling of the eigenvectors is illustrated in \cref{fig: decoupling of eve}.

From the proof of \cref{lemma:exponential decay and decoupling}, we can read out the decay or growth rate of the eigenvectors.

\begin{table}[h]
\begin{tabular}{c|cc|cc|}
\cline{2-5}
                          & \multicolumn{2}{c|}{Upper branch}                                                   & \multicolumn{2}{c|}{Lower branch}                                                 \\ \cline{2-5} 
                          & \multicolumn{1}{c|}{$\lambda\in E^{-\theta}$} & $\lambda\not\in E^{-\theta}$ & \multicolumn{1}{c|}{$\lambda\in E^{\theta}$} & $\lambda\not\in E^{\theta}$ \\ \hline
\multicolumn{1}{|c|}{$\bm x$} & \multicolumn{1}{c|}{$e^{-\gamma/2}$}          & $e^{-\gamma/2}$              & \multicolumn{1}{c|}{$<1$}                    & $>1$                        \\ \hline
\multicolumn{1}{|c|}{$\bm y$} & \multicolumn{1}{c|}{$>1$}                     & $<1$                         & \multicolumn{1}{c|}{$e^{\gamma/2}$}          & $e^{\gamma/2}$              \\ \hline
\end{tabular}
\caption{Approximated decay and growth rate of the left and right part of an eigenvector of the capacitance matrix depending on the location of the corresponding eigenvalue. Values greater than 1 correspond to growth and lower than 1  correspond to decay. Here, upper and lower branches refer respectively to $\lambda\in B_\varepsilon((\mu^{\pm\theta})^{-1}([-1,1]))$ as in \cref{prop:evalocation}.}
\end{table}

\begin{figure}[!h]
    \centering
    \begin{subfigure}[t]{0.48\textwidth}
        \centering
        \includegraphics[height=0.55\textwidth]{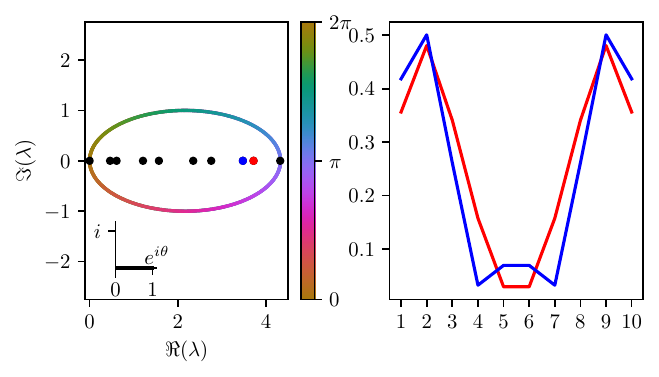}
        \caption{}
    \end{subfigure}
    \hfill
    \begin{subfigure}[t]{0.48\textwidth}
        \centering
        \includegraphics[height=0.55\textwidth]{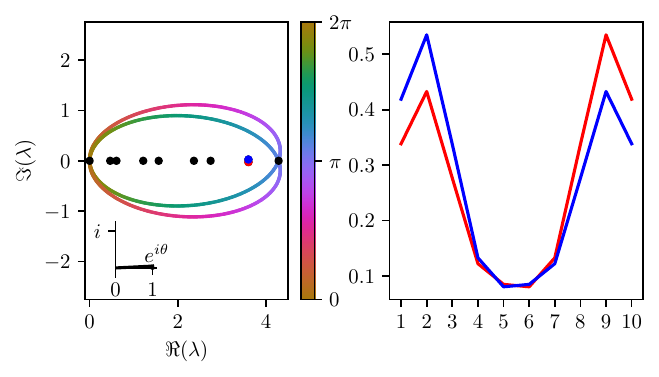}
        \caption{}
    \end{subfigure}\\
    \begin{subfigure}[t]{0.48\textwidth}
        \centering
        \includegraphics[height=0.55\textwidth]{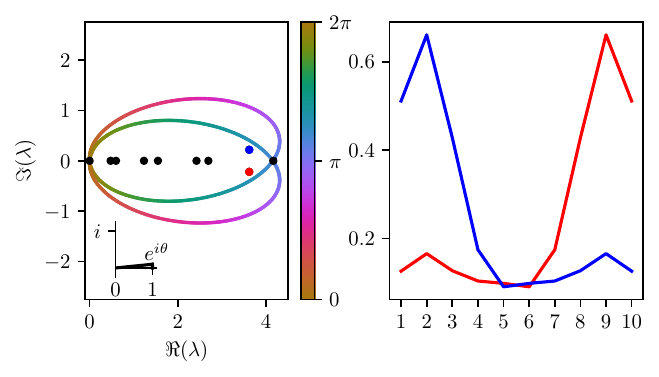}
        \caption{}
    \end{subfigure}
    \hfill
    \begin{subfigure}[t]{0.48\textwidth}
        \centering
        \includegraphics[height=0.55\textwidth]{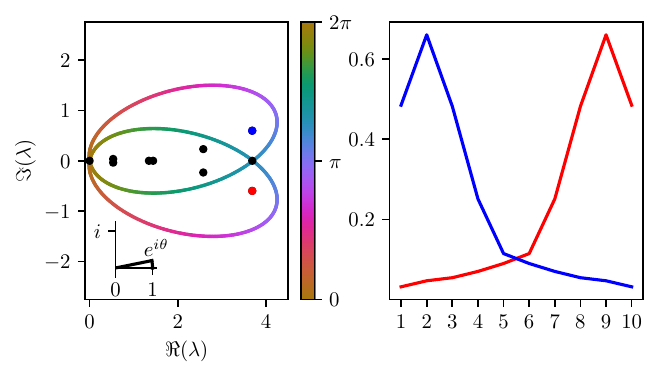}
        \caption{}
    \end{subfigure}
    \caption{Decoupling of the eigenvectors of the gauge capacitance matrix. The macroscopic behaviour of the eigenvectors (exponential decay/growth) is predicted by the location of the eigenvalues in the complex plane with respect to the region of topological convergence defined in \eqref{eq:def_E1_E2} displayed here as trace of \eqref{eq: symbol def}. Looking at the two highlighted eigenvalues (red and blue), Figure (\textsc{A-C}) correspond to item (ii) in \cref{lemma:exponential decay and decoupling} while (\textsc{D}) corresponds to item (i).}
    \label{fig: decoupling of eve}
\end{figure}

\subsubsection{Topological origin}
This section is devoted to illustrating the topological origin of the specific condensation properties of $\mathnormal{C}^{\theta,\gamma}$'s eigenvectors that have been shown in Theorem \ref{lemma:exponential decay and decoupling}. Especially, in Theorem \ref{lemma:exponential decay and decoupling} we demonstrate the condensation properties by the specific behaviours of $a(\mu)$, while here we concentrate on the Fredholm index theory of the symbol of the Toeplitz operator, which directly leads to considering $E^{\pm\theta}$ in (\ref{eq:def_E1_E2}) defined by symbols $f^{\theta}_\pm$. Instead of considering the eigenvector of a Toeplitz operator, we analyse the pseudoeigenvector of the corresponding matrix $\mathnormal{C}^{\theta,\gamma}$ to keep the boundary of the system. We first present the topological origin of the case (2) in Theorem \ref{lemma:exponential decay and decoupling}.
\begin{theorem}
    Suppose that $\lambda \in E^\theta
        \cap E^{-\theta}$. For some $0<\rho<1$ and sufficiently large $N$, there exists a pseudoeigenvector $\bm v$ of $\mathnormal{C}^{\theta,\gamma}$ satisfying
    \[
        \frac{\lVert (\mathnormal{C}^{\theta,\gamma}-\lambda)\bm v\rVert}{\lVert \bm v \rVert} \leq \max(C_1, N C_2)\rho^{N-1}
    \]
    such that
    \[
        \begin{cases} \ds \frac{|\bm v_j|}{\max_{i}|\bm v_i|} \leq C_3 N\rho^{j-1}, \quad j =1,\cdots, N, \\
        \nm
            \ds \frac{|\bm v_j|}{\max_{i}|\bm v_i|} \leq C_3 N\rho^{2N-j}, \quad j =N+1,\cdots, 2N,
        \end{cases}
    \]
    where $C_1, C_2, C_3$ are constants independent of $N$.
\end{theorem}

Now, we present the topological origin of case (iii) in Theorem \ref{lemma:exponential decay and decoupling}.
\begin{theorem}
    Suppose that $\lambda \in E^{\theta} \triangle E^{-\theta}$. For some $0<\rho<1$ and sufficiently large $N$, there exists a pseudoeigenvector $\bm v$ of $\mathnormal{C}^{\theta,\gamma}$ satisfying
    \[
        \frac{\lVert (\mathnormal{C}^{\theta,\gamma}-\lambda)\bm v\rVert}{\lVert \bm v \rVert} \leq \max(C_1, N C_2)\rho^{N-1}
    \]
    such that
    \[
        \begin{cases} \ds \frac{|\bm v_j|}{\max_{i}|\bm v_i|} \leq C_3 N\rho^{j-1}, \quad j =1,\cdots, 2N, \text{ if $\lambda \in E^{\theta}$, } \\
        \nm
           \ds  \frac{|\bm v_j|}{\max_{i}|\bm v_i|} \leq C_3 N\rho^{2N-j}, \quad j =1,\cdots, 2N, \text{ if $\lambda \in E^{-\theta}$, }
        \end{cases}
    \]
    where $C_1, C_2, C_3$ are constants independent of $N$.
\end{theorem}

\section{Generalised Brillouin zone for non-reciprocal problems}\label{sec: GBZ}

Various physical systems are modelled through Toeplitz matrices and operators and variations thereof. In the previous sections, we have seen systems of subwavelength resonators in (classical) one-dimensional wave physics \cite{ammari.barandun.ea2023Exponentially, ammari.barandun.ea2024Mathematical} but also other system as the tight-binding model with nearest neighbour approximation in condensed matter theory \cite{okuma.sato2023Nonhermitian,mironov.oleinik1994Limits,mironov.oleinik1997limits,fefferman.lee-thorp.ea2017Honeycomb,thouless1974Electrons} falls into this category. The tools developed in this section apply to all of these models.

System are typically constituted by resonators or particles all of which we will call sites in this section.
\subsection{Physical systems and their mathematical models}

We assume that the interactions between the sites repeat periodically with period $k$, so that if the interactions are all the same $k=1$ holds. We denote by $\spatialperiod$ the spatial period of recurrence.

In all of the aforementioned examples, the following modelling applies:
\begin{description}
    \item[Finite systems] are constituted by a finite number of sites. These are modelled by tridiagonal $k$-Toeplitz matrices;
    \item[Semi-infinite systems] are constituted by an infinite number of sites but only in one direction from a fixed origin. These are modelled by tridiagonal $k$-Toeplitz operators;
    \item[Infinite systems] are constituted by an infinite number of sites where no point is a privileged choice of origin. These are modelled by tridiagonal $k$-Laurent operators.
\end{description}
In the literature, these three cases are also known as open boundary conditions, semi-infinite boundary conditions, and periodic boundary conditions \cite{okuma.sato2023Nonhermitian}. The subsequent theory recently introduced in \cite{ammari.barandun.ea2024Generalised} applies to all physical systems that can be modelled through these mathematical objects.

\subsection{Floquet--Bloch theory in the Hermitian case}
Floquet--Bloch theory is the appropriate tool for analysing periodic systems in the Hermitian case, in particular because Floquet's theorem relates the spectra of the infinite operator to the spectra of the corresponding quasiperiodic operators. Here, the Hermiticity of the system is reflected in the Hermiticiy of the matrices and operators, \emph{i.e.}, $M=M^*\coloneqq\bar{M}^\top$.

One may quickly notice that when studying a tridiagonal system associated to a tridiagonal $k$-Laurent operator $L(f)$ (defined in Section \ref{section:ktoeplitztheory})  and denoting by $\alpha$ the quasiperiodicity, the operator associated to the Floquet-transformed system is simply given by the symbol
$f(e^{-\i\alpha \spatialperiod})$. Using \cite[Theorem 2.8]{ammari.barandun.ea2024Spectra}, we can find that
\begin{align}\label{eq:ClassicalFBT}
    \sigma(L(f)) = \bigcup_{\alpha \in Y^*} \sigma(f(e^{-\i\alpha \spatialperiod})),
\end{align}
where $Y^*\coloneqq [-\pi/\spatialperiod,\pi/\spatialperiod)$ is the first \emph{Brillouin zone}. This exactly mirrors the Floquet--Bloch decomposition of the spectrum for a periodic self-adjoint elliptic operator.

Combining the Bauer–Fike theorem together with \cite[Corollary 6.16]{bottcher.silbermann1999Introduction} shows that for these Hermitian systems, the spectrum of the finite system converges to the spectrum of the infinite one, meaning that
\begin{align}\label{eq:ClassicalConv}
    \sigma(T_{mk}(f))\xrightarrow{m\to\infty}\sigma (L(f))
\end{align}
in the Hausdorff sense. Recall that $T_{mk}(f)$ denotes a $k$-Toeplitz matrix of order $mk$ and $T(f)$ a $k$-Toeplitz operator as defined in \cref{section:ktoeplitztheory}. On the other side, the Hermiticity of the symbol immediately implies that
\begin{align*}
    \sigma(T(f))=\sigma(L(f)).
\end{align*}
\subsection{The non-reciprocal case}
We now shift our focus to the case of non-reciprocal systems, that is the case where the matrices and operators we are working with are no longer Hermitian. Non-reciprocal systems are peculiar for having eigenmodes which are condensed on one edge of the system as shown in \Cref{sec: skin effect} (see also \cite{yao.wang2018Edge,ammari.barandun.ea2024Mathematical}) and therefore present a privileged choice of origin, making a semi-infinite system the natural corresponding physical structure.

We consider the following prototypical example:
\begin{align}\label{eq: example symbol}
    f(z)=\begin{pmatrix}
             0 & -2 + -\frac{z}{10} \\ -\frac{9}{10} + \frac{1}{z} & 0
         \end{pmatrix},
\end{align}
and look at the spectra plotted in \cref{fig:different spectra}.
\begin{figure}[h]
    \centering
    \includegraphics[width=0.75\linewidth]{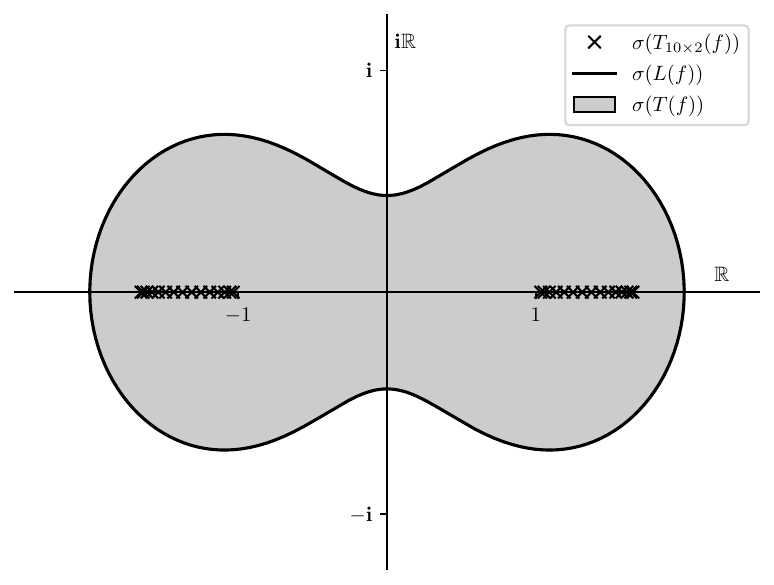}
    \caption{Spectra of the different mathematical objects related to the symbol $f$ from \eqref{eq: example symbol}.}
    \label{fig:different spectra}
\end{figure}

We observe the following:
\begin{enumerate}
    \item[(i)] The symbol $f(z)$ is no longer collapsed (see \cref{lemma: collapsed symbol}) and now has nonempty interior. This will prove to be the crucial difference between the non-reciprocal and the reciprocal cases, as in the non-reciprocal setting the spectra $\sigma (L(f))$ and $\sigma (T(f))$ do not agree anymore.
    \item[(ii)] $\sigma(L(f)) = \bigcup_{\alpha \in Y^*} \sigma(f(e^{-\i\alpha \spatialperiod}))$ still holds also in the non-reciprocal case. Nevertheless, non-reciprocal systems have a privileged choice of origin as they present an exponential decay in their eigenmodes. One would wish that the Floquet--Bloch decomposition could model this decay and would cover $\sigma(T(f))$ and not only $\sigma(L(f))$.
    \item[(iii)] The convergence $\sigma(T_{mk}(f))\xrightarrow{m\to\infty}\sigma (L(f))$ does not hold anymore and neither does $\sigma(T_{mk}(f))\xrightarrow{m\to\infty}\sigma (T(f))$. The spectrum of $T_{mk}(f)$ is purely real while the ones of $L(f)$ and of $T(f)$ have nontrivial imaginary parts.
\end{enumerate}

In the following sections, we will address the issues (ii) and (iii) above and resolve both of them.

\subsection{$k$-Toeplitz operators and the generalised Brillouin zone} \label{sec:topgbz}
As it can be seen for the non-Hermitian skin effect in \Cref{sec: skin effect}, the classical Floquet--Bloch transform with real quasiperiodicities $\alpha \in Y^*$ fails to capture non-reciprocal decay and only covers the $k$-Laurent operator $L(a)\subsetneq T(a)$ (identity (\ref{eq:ClassicalFBT}) still holds).
In order to rectify this, we extend the allowable quasiperiodicities into the complex plane. This is a natural extension. Indeed, considering the quasiperiodicity condition
$$
    u(x + \spatialperiod) = e^{\i \spatialperiod \alpha} u(x),
$$
one immediately notices that decaying functions (as are the eigenmodes of non-reciprocal systems) cannot be described through this relation for $\alpha\in\R$. This would, however, be the case if we allow $\alpha \in \C$.
\begin{definition}
    For a tridiagonal $k$-Toeplitz operator with symbol
    \begin{align} \label{eq: k-Toeplitz base symbol}
        f:z\mapsto \begin{pmatrix}
                       a_1        & b_1 & 0      & \cdots  & 0       & c_k z   \\
                       c_1        & a_2 & b_2    &         &         & 0       \\
                       0          & c_2 & \ddots & \ddots  &         & \vdots  \\
                       \vdots     &     & \ddots & \ddots  & b_{k-2} & 0       \\
                       0          &     &        & c_{k-2} & a_{k-1} & b_{k-1} \\
                       b_k z^{-1} & 0   & \cdots & 0       & c_{k-1} & a_k
                   \end{pmatrix}
    \end{align}
    with nonzero off-diagonal entries we define the \emph{non-reciprocity rate} as
    \begin{equation}\label{eq:deltadef}
        \Delta = \ln \prod_{j=1}^k \left\vert \frac{b_j}{c_j}\right\vert.
    \end{equation}
    Furthermore, we define the \emph{generalised Brillouin zone} to be
    \begin{align}
        \mathcal{B} = \bigg\{ \qp \mid \alpha\in [-\pi/\spatialperiod,\pi/\spatialperiod), \beta \in [0,\Delta/\spatialperiod] \bigg\},
        \label{eq: GBZ}
    \end{align}
    where $\spatialperiod$ denotes the physical length of the unit cell.
\end{definition}
We aim at showing that this expansion of the Brillouin zone allows us to reinstate the Floquet--Bloch theorem in a physical sense, which we encompass in the following theorem.
\begin{theorem}\label{thm: GFBT}
    Consider a tridiagonal $k$-Toeplitz operator with symbol $f$ as in \eqref{eq: k-Toeplitz base symbol} and with nonzero off-diagonal entries and let $\mathcal{B}$ be the generalised Brillouin zone from \eqref{eq: GBZ}. Then,
    \begin{align}\label{equ:floquetidentitynonhermitian}
        \sigma(T(f)) = \bigcup_{\qp \in \mathcal{B}} \sigma(f(\associateqp)),
    \end{align}
    up to at most $(k-1)$ points which may be in $\sigma(T(f))$ but not in $\bigcup_{\qp \in \mathcal{B}} \sigma(f(\associateqp))$. Furthermore, for every $\lambda \in \sigma(T(f))$, the Brillouin zone $\mathcal{ B}$ contains exactly two corresponding quasiperiodicities
    \begin{align*}
        \alpha+\i\beta & \in [-\pi/\spatialperiod,\pi/\spatialperiod)+\i[0,\Delta/(2\spatialperiod)],\\ \alpha'+\i\beta' =  (-\zeta/\spatialperiod-\alpha) + \i(\Delta/\spatialperiod-\beta) &\in [-\pi/\spatialperiod,\pi/\spatialperiod)+\i[\Delta/(2\spatialperiod),\Delta/\spatialperiod]
    \end{align*}
    such that $\lambda\in \sigma(a(e^{-\i\spatialperiod(\qp)}) = \sigma(a(e^{-\i\spatialperiod(\alpha'+\i\beta')}))$. Here, $\zeta$ denotes the shift $\zeta \coloneqq \Arg (\prod_{i=1}^k \frac{b_i}{c_i})$.
\end{theorem}
\begin{remark}\hfill
    \begin{itemize}
        \item From now on, we will take $\spatialperiod=1$ without any loss of generality. To reintroduce $\spatialperiod$, we simply rescale $\mathcal{B}$ by $1/\spatialperiod$ and all occurrences of $\qp$ in the formulas by $\spatialperiod$.
        \item We will also use $\qp$ and $z=e^{-\i(\qp)}$ interchangeably and refer to them as \emph{associated}.
        \item For a given quasiperiodicity $\qp\in \mathcal{B}$ with $\beta\in[0,\Delta/(2\spatialperiod)]$, we call $\alpha'+\i\beta'$ with $\alpha' = \zeta/\spatialperiod-\alpha$ and $\beta' = \Delta/\spatialperiod-\beta$ the \emph{conjugate quasiperiodicities}.
        \item In the above definition, we have assumed that $\Delta>0$. However, this is not necessarily the case. If $\Delta<0$, then the eigenmodes\footnote{While these eigenmodes $\bm u$ are eigenmodes in the symbolic sense $(T(a)-\lambda I) \bm u = 0$, they no longer lie in $\ell^2$ due to their exponential growth.} of $T(a)$ will turn out to be exponentially growing and we observe a \emph{negative} decay parameter $\beta$. We can then let $\beta \in [\Delta/\spatialperiod,0]$ and all of the arguments below work analogously.
        \item As $\alpha\mapsto e^{-\i(\alpha+\i\beta)}$ is periodic in $\alpha$ with period $2\pi$ we consider equality with respect to $\alpha$ modulo $2\pi$ and choose the convention $\alpha\in [-\pi, \pi)$.
        \item For reciprocal systems (\emph{i.e.}, for $\Delta=0$),  the generalised Brillouin zone reduces to the standard Brillouin zone, effectively making $\mathcal{B}$ an extension of $Y^*$.
    \end{itemize}
\end{remark}

To prove \cref{thm: GFBT}, we will need Theorem \ref{thm: windingspectrum} for the spectra of tridiagonal $k$-Toeplitz operators and the following lemma. 

\begin{lemma}\label{lem:detform}
    Let $T(f)$ be a tridiagonal $k$-Toeplitz operator with symbol $f(z)$ as in \eqref{eq: k-Toeplitz base symbol}. Then, we have
    \begin{equation*}
        \det (f(z) - \lambda) = \psi(z) + g(\lambda) \qquad \lambda,z \in \C,
    \end{equation*}
    where
    \begin{equation}
        \psi(z)=(-1)^{k+1}\left((\prod_{i=1}^k c_i)z+(\prod_{i=1}^k b_i)z^{-1}\right),
    \end{equation}
    and $g(\lambda)$ is a polynomial\footnote{See \cite[Appendix A]{ammari.barandun.ea2024Spectra} for further details.} of degree $k$.
\end{lemma}
As a consequence of the above two results, we can see that if we define $E$ to be the ellipse (with interior) traced out by $\psi(\mathbb{T}^1)$, then we have
\begin{align}
    \label{eq: det and wind based on E}
    \sigma_{\det}(f) = (-g)^{-1}(\partial E) \quad \text{and} \quad \sigma_{\mathrm{wind}}(f) = (-g)^{-1}(\interior E),
\end{align}
where $\sigma_{\mathrm{det}}(f), \sigma_{\mathrm{wind}}(f)$ are given by (\ref{equ:detspectraformula1}) and (\ref{equ:windspectraformula3}), respectively.

The following results hold.
\begin{lemma}\label{lem:ellipseprops}
    The map
    \begin{align*}
        \psi:\mathcal{B} & \to E                      \\
        \qp              & \mapsto \psi(e^{-\i(\qp)})
    \end{align*}
    is well-defined, surjective and for every $\xi \in E$, there exist unique $$\alpha \in [-\pi,\pi], \beta \in [0,\Delta/2]$$ such that
    \[
        \psi^{-1}(\xi) = \big\{\qp, \qpconj \big\}.
    \]

    Furthermore, if we denote by $E_\beta$ the (with interior) ellipse traced out by $\psi([-\pi,\pi]+\i\beta)$, then we have $E_\beta = E_{\Delta-\beta}$ and $\operatorname{int} E_{\beta_1} \supset E_{\beta_2}$ for $0\leq \beta_1<\beta_2 \leq \Delta/2$.
\end{lemma}
The proof of \cref{lem:ellipseprops} can be found in \cite{ammari.barandun.ea2024Generalised}. An immediate consequence of this lemma is that for a given quasiperiodicity $\qp$ we have $\sigma(f(\qp)) = \sigma(f(\qpconj))$. We show graphically the statement of \cref{lem:ellipseprops} in \cref{fig:betacollapse}. We can prove \cref{thm: GFBT} by using Lemma  \ref{lem:ellipseprops}. We omit the details which can be found in \cite{ammari.barandun.ea2024Generalised}.

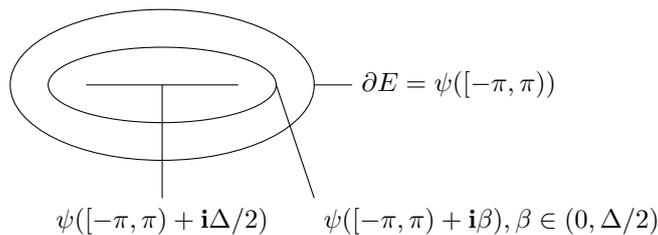
\begin{figure}[!h]\label{fig:betacollapse}
    \centering
    \begin{tikzpicture}
        \draw (0,0) ellipse (2 and 1);
        \draw (0,0) ellipse (1.5 and 0.5);

        \draw[-] (-1,0) -- (1,0);

        \draw[-] (0,-1.5) node[anchor=north]{$\psi([-\pi,\pi)+\i\Delta/2)$} -- (0,0);
        \draw[-] (2.5,0) node[anchor=west]{$\partial E = \psi([-\pi,\pi))$} -- (2,0);
        \draw[-] (2,-1.5) node[anchor=north west]{$\psi([-\pi,\pi)+\i\beta), \beta \in (0,\Delta/2)$} -- (1.5,0);
    \end{tikzpicture}
    \caption{Parametrisation of the ellipse $E$ by $\psi$ for $\alpha\in Y^*\simeq \mathbb{T}^1$ and $\beta\in [0,\Delta/2]$. As $\beta$ increases from $0$ to $\Delta/2$ the corresponding ellipse drawn out by $\psi(\mathbb{T}^1+\i\beta)$ shrinks uniformly.}
\end{figure}

At the beginning of this section, we have motivated the introduction of an imaginary part with the heuristic of taking into account the possible decay of the eigenvectors. This heuristic is made formal with the following proposition.
\begin{theorem}\label{prop:eigenvectors}
    Let $\lambda\in \sigma_{\det}(f) \cup \sigma_{\mathrm{wind}}(f)$, with the uniquely determined corresponding quasiperiodicities $\qp, \qpconj$.
    Let $\bm u_1, \bm u_2$ be eigenvectors of $f(\qp)$, $f(\qpconj)$ associated with that eigenvalue $\lambda$. Then, we can obtain a symbolic eigenvector $T(f) \bm u = \lambda \bm u$ of the Toeplitz operator as a linear combination of the $(\qp)$-quasiperiodic extension $\widetilde{\bm u}_1$ and the $(\qpconj)$-quasiperiodic extension $\widetilde{\bm u}_2$ of $\bm u_1, \bm u_2$, respectively. Furthermore, all eigenvectors $\bm u$ of $T(f)$ are of this form.
    Here, the $(\qp)$-quasiperiodic extension of a vector $\bm v$ is defined as
    \[
        \widetilde{\bm v} \coloneqq (\bm v^\top, z^{-1}\bm v^\top, z^{-2}\bm v^\top, \dots)^\top
    \]
    for $z = e^{-\i(\qp)}$. Consequently, for every $j\in\N$,
    \begin{align}
        \frac{\vert \bm u^{(j+k)} \vert}{\vert \bm u^{(j)} \vert} = e^{-\beta}.
    \end{align}
\end{theorem}

\begin{remark}
    \cref{prop:eigenvectors} hence justifies why $\Delta$ is called the \emph{non-reciprocity rate}, as it directly translates into the decay rate of the eigenvector, a peculiarity of non-reciprocal systems. However, in the generic case an eigenvector $\bm u$ of $T(a)$ will be a linear combination of $\widetilde{\bm u}_1$ and $\widetilde{\bm u}_2$ with decay rates $\beta$ and $\Delta-\beta$ and thus has decay rate $\max\{\beta,\Delta-\beta\}$ which is maximised if $\beta=\Delta/2$.
    Hence, the actual \emph{maximal rate of decay} is $\Delta/2$.

    We can also see that the construction in \cref{prop:eigenvectors} is independent of the actual entries of $T(a)$ in its first row. Hence, it continues to work even if the top left edge of $T(a)$ is perturbed.
\end{remark}

\subsection{Three spectral limits} \label{sec:3limits}
\begin{figure}
    \centering
    \includegraphics[width=\textwidth]{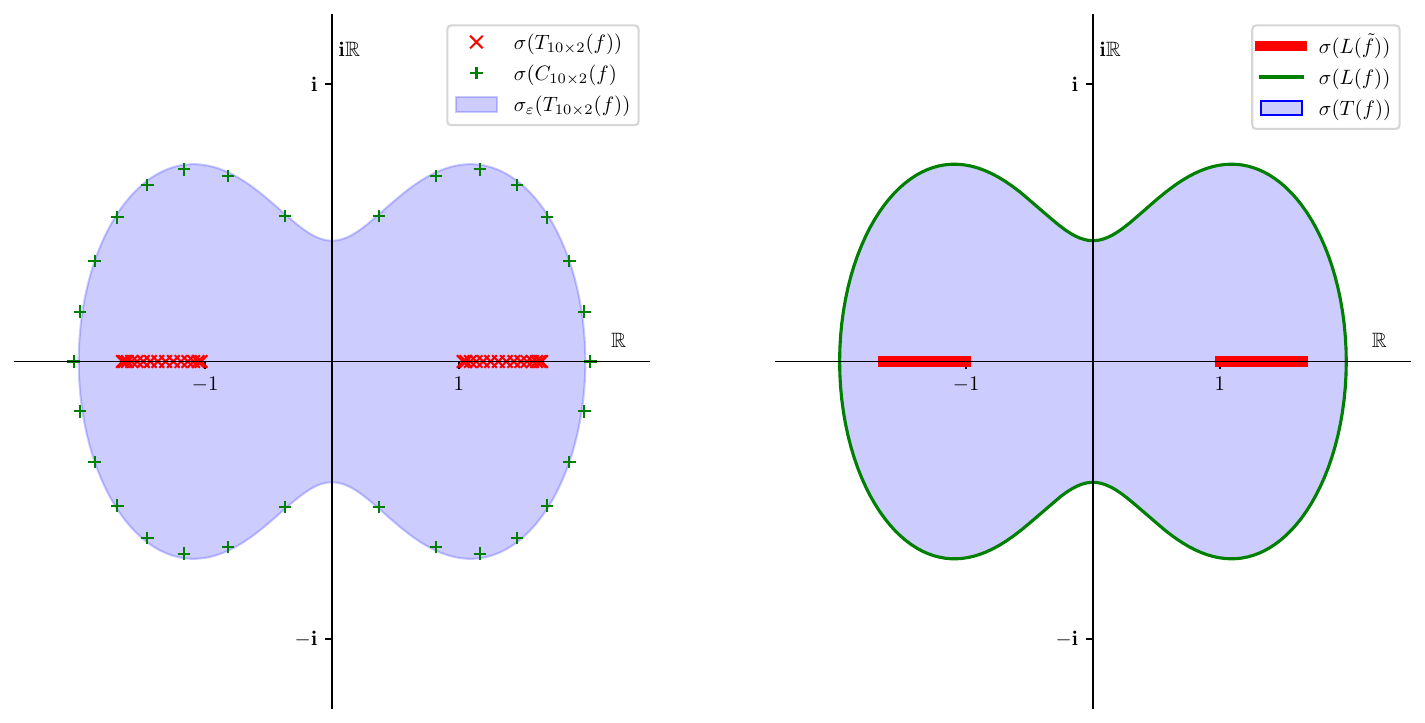}
    \caption{Illustration of the three spectral limits. We see that the eigenvalues of the circulant matrix $C_{2\times 10}(f)$ (green) arrange around the symbol curve and hence converge to the Laurent operator limit $L(f)$. The eigenvalues of the Toeplitz matrix $T_{2\times 10}(f)$ (red) arrange around the collapsed symbol $\Tilde{f}$ (as defined in the proof of \cref{thm:obcconv}) and hence converge to the collapsed Toeplitz (or Laurent) operator $T(\Tilde{f})$. The $\varepsilon$-pseudospectrum of $T_{mk}(f)$ corresponds exactly to the interior of the symbol curve and thus converges to the actual Toeplitz limit $T(f)$. }
    \label{fig:3lim}
\end{figure}

Having restored the Floquet--Bloch decomposition for the Toeplitz operator limit $T(f)$, we now aim to understand how and if finite Toeplitz matrices $T_{mk}(f)$ converge to this limit as $m\to \infty$. Crucially, while in the Hermitian setting the symbol $f(z)$ is collapsed, in general its eigenvalues trace out a path with nonempty interior.
Any point in this interior is  exponentially close to a pseudoeigenvalue of $T_{mk}(f)$ in the limit $m\to \infty$.
This in turn causes the finite system and its limiting spectrum to be highly sensitive to boundary conditions. \cref{fig:3lim} shows the finite spectra under different boundary conditions as well as the pseudospectrum. The collapsed symbol causes their respective limits to coincide in the Hermitian setting but we can see that they diverge in the non-reciprocal setting.

The appropriate generalisation of the Brillouin zone depends on the limit of interest and does not necessarily match the generalised Brillouin zone for the Toeplitz operator limit $T(f)$ as defined in \cref{eq: GBZ}.
For open and periodic boundary conditions, we will now characterise the limiting spectrum and give the appropriate generalised Brillouin zone. Finally, we will investigate the limiting pseudospectrum and see how this connects the two boundary conditions.

\subsubsection{Open limit}
First, we aim to characterise the limiting spectrum $\sigma(T_{mk}(f))$ as $m\to \infty$, which corresponds to a growing finite system.
The main idea is that for non-reciprocal tridiagonal systems, we can perform a change of basis to obtain a similar symmetric system. The symbol of this system is collapsed and we can recover the traditional spectral convergence and decomposition. Then we may transform back into the original basis in order to achieve a Floquet--Bloch relation, but with a Brillouin zone shifted into the complex plane to account for the exponential decay of the eigenvectors.

\begin{theorem}\label{thm:obcconv}
    Let $T(f)$ be a tridiagonal Toeplitz operator with symbol $f(z)$ and $b_i, c_i\neq 0$ for all $1\leq i \leq k$. We then have
    \[
        \lim_{m\to \infty} \sigma(T_{mk}(f)) = \bigcup_{\alpha\in Y^*}\sigma(f(e^{-\i\spatialperiod(\alpha+\i\Delta/(2\spatialperiod))})).
    \]
\end{theorem}

This result is illustrated in \cref{fig:3lim} where the spectrum of the Toeplitz matrix on the left-hand side (corresponding to the open boundary condition) converges to the spectrum of the Toeplitz operator with collapsed symbol $T(\Tilde{f})$ on the right-hand side. As we just proved the appropriate \emph{generalised Brillouin zone} to decompose the spectrum of this operator is the classical Brillouin zone, shifted into the complex plane:
\[
    \mathcal{B}_\text{OBC} = \bigg\{ \alpha + \i\Delta/(2\spatialperiod) \mid \alpha\in Y^* \bigg\}.
\]

Consequently, for tridiagonal Toeplitz systems with open boundary conditions, shifting the Brillouin zone by $\Delta/(2\spatialperiod)$ restores spectral convergence as well as the Floquet--Bloch decomposition. This corresponds to the fact that in the tridiagonal case, all the eigenmodes have the same rate of decay. Namely, this decay is the maximal decay that is given explicitly by $\Delta/(2\spatialperiod)$.

\subsubsection{Periodic boundary conditions and the Laurent operator limit}
We can impose periodic boundary conditions on $T_{mk}(f)$ and  get the tridiagonal $k$-circulant matrix
\begin{equation}
    (C_{mk}(f))_{ij} \coloneqq \begin{cases}
        c_k              & \text{for } i=0, j=mk,        \\
        b_k              & \text{for } i=mk, j=0,        \\
        (T_{mk}(f))_{ij} & \text{otherwise}. \\
    \end{cases}
\end{equation}

The following result holds.
\begin{theorem}\label{thm:pbcconv}
    Let $C_{mk}(f)$ be a tridiagonal $k$-circulant matrix as above. Then, we have the following spectral decomposition:
    \[
        \sigma(C_{mk}(f)) = \bigcup_{j=0}^{m-1}\sigma(f(e^{2\pi\i j/m})).
    \]
    Furthermore, if we let $m \to \infty$,
    \begin{equation}\label{equ:periodiclimiting}
        \sigma(C_{mk}(f)) = \bigcup_{j=0}^{m-1}f(e^{2\pi\i j/m}) \to \bigcup_{\alpha\in Y^*}f(e^{-\i\spatialperiod \alpha}) = \sigma(L(f)).
    \end{equation}
\end{theorem}

We have thus shown the analogue of \cref{thm:obcconv} for periodic boundary conditions. As we can see in \cref{fig:3lim}, imposing periodic boundary conditions on the finite system causes the spectrum of $C_{mk}(f)$ to diverge drastically from the spectrum of $T_{mk}(f)$. This corresponds to the fact that while the eigenmodes of $T_{mk}(f)$ have a decay of $\Delta/2L$, imposing periodic boundary conditions forces the eigenmodes to be decay-free, causing a large perturbation. The non-Hermitian non-reciprocity thus causes the system to be highly sensitive to boundary conditions.
\cref{fig:3lim} further illustrates how the spectrum of the circulant matrices $C_{mk}(f)$ arranges around the symbol curve and thus converges to the Laurent operator limit as $m\to \infty$. The above theorem therefore shows that the appropriate Brillouin zone for this setting is the classical Brillouin zone with no decay: $$\mathcal{B}_\text{PBC}=Y^*.$$

\subsubsection{Pseudospectrum and the Toeplitz operator limit}
Finally, we investigate the pseudospectrum of the finite system $T_{mk}(a)$. Crucially, while the spectrum of Toeplitz matrices is highly sensitive to boundary conditions, the pseudospectrum is not and using \cite[Corollary 6.16]{bottcher.silbermann1999Introduction}, we find that it converges to the Toeplitz operator limit.
\begin{theorem}
    Consider a continuous symbol $f\in L^\infty(\mathbb{T}^1,\C^{k\times k})$ so that $T(f)$ is tridiagonal. Then, for every $\epsilon > 0$,
    \begin{align}
        \lim_{m\to\infty} \sigma_\epsilon(T_{mk}(f)) = \sigma_\epsilon (T(f)).
    \end{align}
\end{theorem}
In particular, the previous theorem implies that
$$
    \lim_{\epsilon \to 0}\lim_{m\to\infty}\sigma_\epsilon (T_{mk}(f)) = \lim_{\epsilon \to 0}\sigma_\epsilon (T(f))   = \sigma (T(f)).
$$

Hence, by \cref{thm: GFBT}, the appropriate generalised Brillouin zone to recover the pseudospectral limit is the same as for the Toeplitz operator:
\[
    \mathcal{B} = \bigg\{\qp \mid \alpha\in [-\pi/\spatialperiod,\pi/\spatialperiod), \beta \in [0,\Delta/\spatialperiod] \bigg\}.
\]
Notably, this Brillouin zone is of higher dimension than the previous two Brillouin zones (\emph{i.e.}, a two dimensional region in $\C$ different from the previous $\mathbb{T}^1\cong Y^*$) as it contains a range of possible decay rates. Furthermore, it contains the shifted Brillouin zone of the open boundary condition and the classical Brillouin zone of the periodic boundary condition as special cases ($\beta = \Delta/(2\spatialperiod)$ and $\beta=0$). This is in line with the fact that the Toeplitz operator spectrum contains both the Laurent operator spectrum, as well as the collapsed Toeplitz spectrum, as seen in \cref{fig:3lim}. The open and periodic boundary condition settings thus correspond to the maximal decay and zero decay extremes of a range of possible pseudospectral decays, captured by the Toeplitz operator spectrum.

\appendix

\section{Capacitance matrix}
The capacitance matrix introduced in \cref{sec:1d} plays a crucial role throughout the presented work. In this appendix, we present the details of its derivation referring to \cite{feppon.cheng.ea2023Subwavelength} in the Hermitian case and to \cite{ammari.barandun.ea2024Mathematical} for the non-Hermitian proofs.

We consider the non-reciprocal case first and think of the reciprocal one as the special case where $\gamma=0$.

The use of the Dirichlet-to-Neumann map has proven itself to be an effective tool to find solutions to \eqref{eq:coupled ods2} that lend themselves to asymptotic analysis in the high-contrast regime. Hereafter, we will adapt the methods of \cite{feppon.cheng.ea2023Subwavelength} to \eqref{eq:coupled ods2}.

The first step is to solve the outer problem
\begin{align}
    \begin{dcases}
        w\prii(x) + \frac{\omega^2}{v^2}w(x)=0,       & x\in\R\setminus\bigcup_{i=1}^N(x_i^{\iL},x_i^{\iR}), \\
        w(x_i^{\iLR})=f_i^{\iLR},                     & \text{for all }1\leq i\leq N,                        \\
        \frac{\dd w}{\dd\ \abs{x}}(x) -\i k w(x) = 0, & x\in(-\infty,x_1^{\iL})\cup (x_N^{\iR},\infty),
    \end{dcases}
    \label{eq:coupled ods outside}
\end{align}
for some boundary data $f_i^{\iLR}\in \C^{2N}$. Its solution is simply given by
\begin{align}
    w(x)=\begin{dcases}
             f_1^{\iL} e^{-\i k (x-x_1^{\iL})}, & x\leq x_1^{\iL},           \\
             a_i e^{\i k x}+b_i e^{-\i k x},    & (x_i^{\iR},x_{i+1}^{\iL}), \\
             f_N^{\iR} e^{-\i k (x-x_N^{\iR})}, & x_N^{\iR}\leq x,           \\
         \end{dcases}
    \label{eq: outer solution}
\end{align}
where $a_i$ and $b_i$ are given by
\begin{align*}
    \begin{pmatrix}
        a_i \\b_i
    \end{pmatrix} = - \frac{1}{2\i \sin(k s_i)}
    \begin{pmatrix}
        e^{-\i k x_{i+1}^{\iL}} & -e^{-\i k x_i^{\iR}} \\
        -e^{\i k x_{i+1}^{\iL}} & e^{\i k x_i^{\iR}}
    \end{pmatrix}
    \begin{pmatrix}
        f_i^{\iR} \\ f_{i+1}^{\iL}
    \end{pmatrix},
\end{align*}
as shown in \cite[Lemma 2.1]{feppon.cheng.ea2023Subwavelength}. The second and harder step is to combine the expression for the outer solution \eqref{eq: outer solution} with the boundary conditions in order to find the full solution. We will handle this information through the Dirichlet-to-Neumann map.

\begin{definition}[Dirichlet-to-Neumann map]
    For any $k\in\C$ which is not of the form
    $n\pi/s_{i}$ for some $n\in\Z\backslash\{0\}$ and $1\<i\<N-1$,
    the \emph{Dirichlet-to-Neumann map} with wave number $k$ is the linear operator
    $\mathcal{T}^{k}\,:\,\C^{2N}\to \C^{2N}$ defined by
    \begin{align*}
        \mathcal{T}^{k}[(f_i^{\iLR})_{1 \leq i \leq N}]=
        \left(\pm\frac{\dd w}{\dd x}(x_i^{\iLR})\right)_{1 \leq i \leq N},
    \end{align*}
    where $w$ is  the unique solution to \eqref{eq:coupled ods outside}.
\end{definition}

We refer to \cite[Section 2]{feppon.cheng.ea2023Subwavelength} for a more extensive discussion of this operator, but recall that $\mathcal{T}^{k}$ has a block-diagonal matrix representation
\begin{align}
    \label{eq: Matrix form DTN}
    T^{k}\begin{pmatrix}
             f_1^{\iL} \\
             f_1^{\iR} \\
             \vdots    \\
             f_N^{\iL} \\
             f_N^{\iR}
         \end{pmatrix} = \begin{pmatrix}
                             \i k &              &              &        &                  &      \\
                                  & A^{k}(s_{1}) &              &        &                  &      \\
                                  &              & A^{k}(s_{2}) &        &                  &      \\
                                  &              &              & \ddots &                  &      \\
                                  &              &              &        & A^{k}(s_{(N-1)}) &      \\
                                  &              &              &        &                  & \i k \\
                         \end{pmatrix}\begin{pmatrix}
                                          f_1^{\iL} \\
                                          f_1^{\iR} \\
                                          \vdots    \\
                                          f_N^{\iL} \\
                                          f_N^{\iR}
                                      \end{pmatrix},
\end{align} where, for any real $\ell\in\R$, $A^{k}(\ell)$ denotes the $2\times 2$ symmetric matrix given by
\begin{equation}
    \label{eqn:1lzi8}
    A^{k}(\ell):=\begin{pmatrix}
        -\dfrac{k \cos(k\ell)}{\sin(k\ell)} & \dfrac{k}{\sin(k\ell)}             \\
        \dfrac{k}{\sin(k\ell)}              & -\dfrac{k\cos(k\ell)}{\sin(k\ell)}
    \end{pmatrix}.
\end{equation}
Consequently, $T^{k}$ is holomorphic in $k$ in a neighbourhood of the origin and admits a power series representation $T^{k} = \sum_{n\geq 0}k^n T_n$, where
\begin{align}
    T_0 = \begin{pmatrix}
              0 &              &              &        &                  &   \\
                & A^{0}(s_{1}) &              &        &                  &   \\
                &              & A^{0}(s_{2}) &        &                  &   \\
                &              &              & \ddots &                  &   \\
                &              &              &        & A^{0}(s_{(N-1)}) &   \\
                &              &              &        &                  & 0 \\
          \end{pmatrix}, \label{eq: DTN T0}
\end{align}
and $A^0(s):=\lim_{k\to0} A^k(s)$.

The above properties of the Dirichlet-to-Neumann map will be crucial to find subwavelength eigenfrequencies. We will use $\mathcal{T}^{k}$ and $T^k$ interchangeably.

\subsubsection{Characterisation of the subwavelength eigenfrequencies}
The Dirichlet-to-Neumann map allows us to reformulate \eqref{eq:coupled ods2} as follows:
\begin{align}
    \begin{dcases}
        u\prii(x) + \gamma u\pri(x)+\frac{\omega^2}{v_b^2}u=0,                                               & x\in\bigcup_{i=1}^N(x_i^{\iL},x_i^{\iR}), \\
        \left.\frac{\dd u}{\dd x}\right\vert_{\iR}(x^{\iL}_{{i}})=-\delta\dtn^{\frac{\omega}{v}}[u]^{\iL}_i, & \forall\ 1\leq i\leq N,                   \\
        \left.\frac{\dd u}{\dd x}\right\vert_{\iL}(x^{\iR}_{{i}})=\delta\dtn^{\frac{\omega}{v}}[u]^{\iR}_i,  & \forall\ 1\leq i\leq N.
    \end{dcases}
    \label{eq:coupled ods with dtn}
\end{align}
We can then further rewrite \eqref{eq:coupled ods with dtn} in weak form by multiplying it by a test function $w\in H^1(D)$ and integrating on the intervals. Explicitly, we obtain that $u$ is a solution to \eqref{eq:coupled ods with dtn} if and only if
\begin{align}
    a(u,w)=0 \label{eq:formulation in weak form}
\end{align}
for any $w\in H^1(D)$, where
\begin{align}
    a(u,w) & = \sum_{i=1}^N\int_{x_i^{\iL}}^{x_i^{\iR}} u\pri\bar{w}\pri -\gamma u\pri \bar{w}-\frac{\omega^2}{v_b^2}u\bar{w}\dd x \nonumber   \\
           & -\delta \sum_{i=1}^N \bar{w}(x_i^{\iR})\dtn^{\frac{\omega}{v}}[u]_i^{\iR} + \bar{w}(x_i^{\iL})\dtn^{\frac{\omega}{v}}[u]_i^{\iL}.
\end{align}
We also introduce a slightly modified bilinear form
\begin{align}
    a_{\omega,\delta}(u,w) & = \sum_{i=1}^N\left(\int_{x_i^{\iL}}^{x_i^{\iR}} u\pri\bar{w}\pri -\gamma u\pri \bar{w}\dd x + \int_{x_i^{\iL}}^{x_i^{\iR}} u\dd x\int_{x_i^{\iL}}^{x_i^{\iR}} \bar{w}\dd x \right)\nonumber                                \\
                           & -\sum_{i=1}^N\left(\int_{x_i^{\iL}}^{x_i^{\iR}}\frac{\omega^2}{v_b^2}u\bar{w}\dd x + \delta\left(\bar{w}(x_i^{\iR})\dtn^{\frac{\omega}{v}}[u]_i^{\iR} + \bar{w}(x_i^{\iL})\dtn^{\frac{\omega}{v}}[u]_i^{\iL}\right)\right).
    \label{eq: def a weak form}
\end{align}
This new form parametrised by $\omega$ is an analytic perturbation of the $a_{0,\delta}$ form, which is continuous and its associated variational problem is well-posed on $H^1(D)$ and so $a_{\omega,\delta}$ inherits this property. Specifically, for every summand of $\Re(a_{0,\delta})$ the following holds for any $\epsilon>0$
\begin{align*}
      & \sum_{i=1}^N\left[\int_{x_i^{\iL}}^{x_i^{\iR}} u\pri\bar{u}\pri\dd x -\gamma \int_{x_i^{\iL}}^{x_i^{\iR}}u\pri \bar{u} +\int_{x_i^{\iL}}^{x_i^{\iR}} u\dd x\int_{x_i^{\iL}}^{x_i^{\iR}} \bar{u}\dd x\right]                                     \\ \geq &\sum_{i=1}^N  \bigg[\Vert u\pri\Vert_{L^2({x_i^{\iL}},{x_i^{\iR}})}^2 - \frac{\epsilon \gamma}{2}\Vert u\pri\Vert_{L^2({x_i^{\iL}},{x_i^{\iR}})}^2 - \frac{\gamma}{2\epsilon} \Vert u\Vert_{L^2({x_i^{\iL}},{x_i^{\iR}})}^2  + \bigg| \int_{x_i^{\iL}}^{x_i^{\iR}}  u \dd x \bigg|^2 \bigg]\\
    = & \sum_{i=1}^N  \bigg[ (1-\frac{\gamma \epsilon}{2}) \Vert u \pri\Vert_{L^2({x_i^{\iL}},{x_i^{\iR}})}^2 -\frac{\gamma}{2 \epsilon} \Vert u \Vert_{L^2({x_i^{\iL}},{x_i^{\iR}})}^2 + \bigg| \int_{x_i^{\iL}}^{x_i^{\iR}}  u \dd x \bigg|^2 \bigg].
\end{align*}
From the compactness of the injection of $H^1$ into $L^2$ and for $\epsilon$ small enough, it follows that  $\Re(a_{0,\delta})$ satisfies a G\aa rding inequality.
On the other hand, it follows from the explicit form \eqref{eq: Matrix form DTN} of the Dirichlet-to-Neumann operator that
$$
    \Im(a_{0,\delta}(u,u)) = 0,
$$
which implies that $u$ is constant on each $[{x_i^{\iL}},{x_i^{\iR}}]$.  But, $\Re (a_{0,\delta}(u,u)) = 0$ for $u$ constant on each $[{x_i^{\iL}},{x_i^{\iR}}]$ implies that the constants are all zero. Thus one obtains the well-posedness of the associated variational problem.

We exploit this, by defining the $h_j(\omega,\delta)$ functions as the Lax--Milgram solutions to the variational problem
\begin{align}
    a_{\omega,\delta}(h_j(\omega,\delta),w) = \int_{x_j^{\iL}}^{x_j^{\iR}} \bar{w}\dd x \label{eq: def h_j}
\end{align}
for every $w\in H^1(D)$ and $1\leq j\leq N$. In the following lemma, we show that the functions $h_j$ allow us to reduce \eqref{eq:formulation in weak form} to a finite dimensional $N\times N$ linear system by acting as basis functions.

\begin{lemma}\label{lemma: I-capmat=0}
    Let $\omega \in \C$ and $\delta \in \R$ belong to a neighbourhood of zero such that $a_{\omega,\delta}$ is coercive. The variational problem \eqref{eq:formulation in weak form} admits a nontrivial solution $u\equiv u(\omega,\delta)$ if and only if the $N\times N$ nonlinear eigenvalue problem
    \begin{align*}
        (I-\exactcapmat(\omega,\delta))\bm x = 0
    \end{align*}
    has a  solution $\omega$ and $\bm x:=(x_i(\omega,\delta))_{1\leq i\leq N}$, where $\exactcapmat(\omega,\delta)$ is the matrix  given by
    \begin{align}
        \label{eq: def exact cap mat}
        \exactcapmat(\omega,\delta)\equiv (\exactcapmat(\omega,\delta)_{ij})_{1\leq i,j\leq N}:=
        \left( \int_{x_i^{\iL}}^{x_i^{\iR}} h_j(\omega,\delta)\dd x
        \right)_{1\leq i,j\leq N}.
    \end{align}
\end{lemma}
\begin{proof}
    This can be shown using the arguments presented in \cite[Lemma 3.4]{feppon.cheng.ea2023Subwavelength}.
\end{proof}
Subwavelength eigenfrequencies are thus the values $\omega(\delta)$ satisfying \eqref{eq: asymptotic behaviour subwavelength resonances} for which $I-\exactcapmat(\omega,\delta)$ is not invertible.

\section*{Acknowledgements}
This work was partially supported by Swiss National Science Foundation grant number 200021-200307.

\section*{Code Availability}
The data that support the findings of this work are openly available at \url{https://doi.org/10.5281/zenodo.10438679}, \url{https://zenodo.org/doi/10.5281/zenodo.10361315},\url{https://doi.org/10.5281/zenodo.8210678}, and \url{https://doi.org/10.5281/zenodo.8081076}.

\section*{Conflict of interest} 
The authors have no competing interests to declare that are relevant to the content of this article.

\printbibliography

@article{ammari.barandun.ea2023Edge,
  title = {Edge Modes in Subwavelength Resonators in One Dimension},
  author = {Ammari, Habib and Barandun, Silvio and Cao, Jinghao and Feppon, Florian},
  date = {2023},
  journaltitle = {Multiscale Modeling \& Simulation},
  volume = {21},
  number = {3},
 pages = {964--992},
}

@article{ammari.barandun.ea2023Exponentially,
  title = {Exponentially Localised Interface Eigenmodes in Finite Chains of Resonators},
  author = {Ammari, Habib and Barandun, Silvio and Davies, Bryn and Hiltunen, Erik Orvehed and Kosche, Thea and Liu, Ping},
 journaltitle = {Stud. Appl. Math.},
year={2024},
}

@online{ammari.barandun.ea2023Perturbed,
  title = {Perturbed {{Block Toeplitz}} Matrices and the Non-{{Hermitian}} Skin Effect in Dimer Systems of Subwavelength Resonators},
  author = {Ammari, Habib and Barandun, Silvio and Liu, Ping},
  year = {2023},
  eprint = {2307.13551},
  eprinttype = {arXiv},
}

@article{ammari.barandun.ea2023Stability,
  title = {Stability of the Non-{{Hermitian}} Skin Effect},
  author = {Ammari, Habib and Barandun, Silvio and Davies, Bryn and Hiltunen, Erik Orvehed and Liu, Ping},
   JOURNAL = {SIAM J. Appl. Math.},
  FJOURNAL = {SIAM Journal on Applied Mathematics},
    VOLUME = {84},
      YEAR = {2024},
    NUMBER = {4},
     PAGES = {1697--1717},
}

@article{ammari.barandun.ea2024Mathematical,
  title = {Mathematical {{Foundations}} of the {{Non-Hermitian Skin Effect}}},
  author = {Ammari, Habib and Barandun, Silvio and Cao, Jinghao and Davies, Bryn and Hiltunen, Erik Orvehed},
   JOURNAL = {Arch. Ration. Mech. Anal.},
  FJOURNAL = {Archive for Rational Mechanics and Analysis},
    VOLUME = {248},
      YEAR = {2024},
    NUMBER = {3},
     PAGES = {Paper No. 33, 34},
}

@online{ammari.barandun.ea2024Spectra,
  title = {Spectra and Pseudo-Spectra of Tridiagonal $k$-{{Toeplitz}} Matrices and the Topological Origin of the Non-{{Hermitian}} Skin Effect},
  author = {Ammari, Habib and Barandun, Silvio and De Bruijn, Yannick and Liu, Ping and Thalhammer, Clemens},
  year = {2024},
  eprint = {2401.12626},
  eprinttype = {arXiv},
}

@article{ammari.davies.ea2020Topologically,
  title = {Topologically Protected Edge Modes in One-Dimensional Chains of Subwavelength Resonators},
  author = {Ammari, Habib and Davies, Bryn and Hiltunen, Erik Orvehed and Yu, Sanghyeon},
  year = {2020},
 JOURNAL = {J. Math. Pures Appl. (9)},
  volume = {144},
  pages = {17--49},
}

@article{ammari.davies.ea2021Functional,
  title = {Functional Analytic Methods for Discrete Approximations of Subwavelength Resonator Systems},
  author = {Ammari, Habib and Davies, Bryn and Hiltunen, Erik Orvehed},
	date = {2024},
	journal = {Pure and Applied Analysis},
}

@book{cbms, 
title={Mathematical theories for metamaterials:  From condensed matter theory to subwavelength physics}, 
series={NSF--CBMS Regional Conf. Ser.}, 
author = {Ammari, Habib and Davies, Bryn and Hiltunen, Erik Orvehed},
publisher={American Mathematical Society, Providence, to appear},
year={2024},
}

@article{francesco,
	title={On the validity of the tight-binding method for describing systems of subwavelength resonators},
	author={Ammari, Habib and Fiorani, Francesco and Hiltunen, Erik Orvehed },
	journal={SIAM J. Appl. Math.},
	volume={82},
	number={4},
	pages={1611--1634},
	year={2022}
}

@article{kadic20193d,
	title={{3D} metamaterials},
	author={Kadic, Muamer and Milton, Graeme W and van Hecke, Martin and Wegener, Martin},
	journal={Nat. Rev. Phys.},
	volume={1},
	number={3},
	pages={198--210},
	year={2019},
	publisher={Nature Publishing Group}
}

@article{ammari.davies.ea2022Exceptional,
  title = {Exceptional points in parity-time-symmetric subwavelength
              metamaterials},
  author = {Ammari, Habib and Davies, Bryn and Hiltunen, Erik Orvehed and Lee, Hyundae and Yu, Sanghyeon},
   JOURNAL = {SIAM J. Math. Anal.},
  FJOURNAL = {SIAM Journal on Mathematical Analysis},
    VOLUME = {54},
      YEAR = {2022},
    NUMBER = {6},
     PAGES = {6223--6253},
}

@online{ammari.hiltunen2020Edge,
  title = {Edge Modes in Active Systems of Subwavelength Resonators},
  author = {Ammari, Habib and Hiltunen, Erik Orvehed},
  year = {2020},
 eprint = {2006.05719},
  eprinttype = {arXiv},
}

@article{ammari.zhang2015Superresolution,
  title = {Super-Resolution in High-Contrast Media},
  author = {Ammari, Habib and Zhang, Hai},
   JOURNAL = {Proc. A.},
  FJOURNAL = {Proceedings A},
    VOLUME = {471},
      YEAR = {2015},
    NUMBER = {2178},
     PAGES = {20140946, 11},
}

@article{borgnia.kruchkov.ea2020Nonhermitian,
  title = {Non-Hermitian Boundary Modes and Topology},
  author = {Borgnia, Dan S. and Kruchkov, Alex Jura and Slager, Robert-Jan},
  date = {2020-02},
  journaltitle = {Phys. Rev. Lett.},
  volume = {124},
  number = {5},
  pages = {056802},
}

@book{bottcher.silbermann1999Introduction,
  title = {Introduction to Large Truncated Toeplitz Matrices},
  author = {B\"ottcher, Albrecht and Silbermann, Bernd},
  date = {1999},
  publisher = {Springer New York},
}

@article{coutant.achilleos.ea2022Subwavelength,
  title = {Subwavelength {{Su-Schrieffer-Heeger}} Topological Modes in Acoustic Waveguides},
  author = {Coutant, Antonin and Achilleos, Vassos and Richoux, Olivier and Theocharis, Georgios and Pagneux, Vincent},
  date = {2022},
  journaltitle = {J. Acoust. Soc. Am.},
  volume = {151},
  number = {6},
  pages = {3626--3632},
}

@article{davies.lou2023Landscape,
  title = {Landscape of Wave Focusing and Localization at Low Frequencies},
  author = {Davies, Bryn and Lou, Yiqi},
  date = {2023},
  journaltitle = {Stud. Appl. Math.},
  volume = {152},
  number = {2},
  pages = {760--777},
}

@article{fefferman.lee-thorp.ea2014Topologically,
  title = {Topologically Protected States in One-Dimensional Continuous Systems and {{Dirac}} Points},
  author = {Fefferman, Charles L. and Lee-Thorp, James P. and Weinstein, Michael I.},
  date = {2014-06},
  journaltitle = {Proc. Natl. Acad. Sci. USA},
  volume = {111},
  number = {24},
  pages = {8759--8763},
}

@article{fefferman.lee-thorp.ea2017Topologically,
  title = {Topologically Protected States in One-Dimensional Systems},
  author = {Fefferman, Charles L. and Lee-Thorp, James P. and Weinstein, Michael I.},
  date = {2017},
  journaltitle = {Mem. Amer. Math. Soc.},
  shortjournal = {Mem. Amer. Math. Soc.},
  volume = {247},
  number = {1173},
  pages = {vii+118},
}

@article{feppon.ammari2022Subwavelength,
    AUTHOR = {Feppon, F. and Ammari, H.},
     TITLE = {Subwavelength resonant acoustic scattering in fast
              time-modulated media},
   JOURNAL = {J. Math. Pures Appl. (9)},
  FJOURNAL = {Journal de Math\'ematiques Pures et Appliqu\'ees. Neuvi\`eme
              S\'erie},
    VOLUME = {187},
      YEAR = {2024},
     PAGES = {233--293},
}

@article{feppon.cheng.ea2023Subwavelength,
  title = {Subwavelength Resonances in One-Dimensional High-Contrast Acoustic Media},
  author = {Feppon, Florian and Cheng, Zijian and Ammari, Habib},
  date = {2023},
  journaltitle = {SIAM J. Appl. Math.},
  volume = {83},
  number = {2},
  pages = {625--665},
}

@article{franca.konye.ea2022Nonhermitian,
  title = {Non-Hermitian Physics without Gain or Loss: {{The}} Skin Effect of Reflected Waves},
  author = {Franca, Selma and K\"onye, Viktor and Hassler, Fabian and van den Brink, Jeroen and Fulga, Cosma},
  options = {useprefix=true},
  date = {2022},
  journaltitle = {Phys. Rev. Lett.},
  volume = {129},
  number = {8},
}

@software{gaulandrePseudoPy,
  title = {{{PseudoPy}}},
  author = {Gaul, Andr\'e},
  url = {https://github.com/andrenarchy/pseudopy},
  version = {1.2.5}
}

@article{ghatak.brandenbourger.ea2020Observation,
  title = {Observation of Non-{{Hermitian}} Topology and Its Bulk-Edge Correspondence in an Active Mechanical Metamaterial},
  author = {Ghatak, Ananya and Brandenbourger, Martin and van Wezel, Jasper and {Corentin Coulais}},
  options = {useprefix=true},
  date = {2020},
  journaltitle = {Proc. Natl. Acad. Sci. USA},
  volume = {117},
  number = {47},
  pages = {29561--29568},
}

@book{golub.vanloan2013Matrix,
  title = {Matrix Computations},
  author = {Golub, Gene H. and Van Loan, Charles F.},
  date = {2013},
  series = {Johns {{Hopkins}} Studies in the Mathematical Sciences},
  edition = {Fourth edition},
  publisher = {The Johns Hopkins University Press},
  location = {Baltimore},
  isbn = {978-1-4214-0794-4},
  langid = {english},
  pagetotal = {756}
}

@article{gover1994eigenproblem,
  title = {The Eigenproblem of a Tridiagonal 2-{{Toeplitz}} Matrix},
  author = {Gover, M.J.C.},
  date = {1994},
  journaltitle = {Linear Algebra Appl.},
  volume = {197--198},
  pages = {63--78},
}

@article{hatano.nelson1996Localization,
  title = {Localization Transitions in Non-Hermitian Quantum Mechanics},
  author = {Hatano, Naomichi and Nelson, David R.},
  date = {1996},
  journaltitle = {Phys. Rev. Lett.},
  volume = {77},
  number = {3},
  pages = {570--573},
}

@article{heiss2012physics,
  title = {The Physics of Exceptional Points},
  author = {Heiss, W D},
  date = {2012},
  journaltitle = {J. Phys. A: Math. Theor.},
  volume = {45},
  number = {44},
  pages = {444016},
}

@article{leykam.bliokh.ea2017Edge,
  title = {Edge Modes, Degeneracies, and Topological Numbers in Non-Hermitian Systems},
  author = {Leykam, Daniel and Bliokh, Konstantin Y. and Huang, Chunli and Chong, Y. D. and Nori, Franco},
  date = {2017},
  journaltitle = {Phys. Rev. Lett.},
  volume = {118},
  number = {4},
  pages = {040401},
}

@article {davis,
    AUTHOR = {Davis, Chandler and Kahan, W. M.},
     TITLE = {The rotation of eigenvectors by a perturbation. {III}},
   JOURNAL = {SIAM J. Numer. Anal.},
  FJOURNAL = {SIAM Journal on Numerical Analysis},
    VOLUME = {7},
      YEAR = {1970},
     PAGES = {1--46},
}

@article{lin.tai.ea2023Topological,
  title = {Topological Non-{{Hermitian}} Skin Effect},
  author = {Lin, Rijia and Tai, Tommy and Li, Linhu and Lee, Ching Hua},
  date = {2023},
  journaltitle = {Frontiers of Physics},
  volume = {18},
  number = {5},
  pages = {53605},
}

@article{lin.zhang2022Mathematical,
  title = {Mathematical Theory for Topological Photonic Materials in One Dimension},
  author = {Lin, Junshan and Zhang, Hai},
  date = {2022-12-08},
  journaltitle = {J. Phys. A: Math. Theor. },
  volume = {55},
  number = {49},
  pages = {495203},
}

@article{longhi.gatti.ea2015Robust,
  title = {Robust Light Transport in Non-{{Hermitian}} Photonic Lattices},
  author = {Longhi, Stefano and Gatti, Davide and Valle, Giuseppe Della},
  date = {2015},
  journaltitle = {Sci. Rep.},
  volume = {5},
  number = {1},
}

@article{margetis.watson.ea2023Su,
  title = {On the {{Su}}--{{Schrieffer}}--{{Heeger}} Model of Electron Transport: {{Low}}-temperature Optical Conductivity by the {{Mellin}} Transform},
  author = {Margetis, Dionisios and Watson, Alexander B. and Luskin, Mitchell},
  date = {2023},
  journaltitle = {Stud. Appl. Math.},
  volume = {151},
  number = {2},
  pages = {555--584},
}

@book{maxwell1873treatise,
  title = {A Treatise on Electricity and Magnetism},
  author = {Maxwell, James Clerk},
  date = {1873},
  series = {Oxford Classic Texts in the Physical Sciences},
  volume = {1},
  pages = {xxxii+521},
  publisher = {The Clarendon Press, Oxford University Press},
  location = {New York},
  isbn = {0-19-850373-3},
  mrclass = {01A75 (01A55 78-03 78A25)},
  mrnumber = {1673643}
}

@article{miri.alu2019Exceptional,
  title = {Exceptional Points in Optics and Photonics},
  author = {Miri, Mohammad-Ali and Al\`u, Andrea},
  date = {2019},
  journaltitle = {Science},
  shortjournal = {Science},
  volume = {363},
  number = {6422},
  pages = {eaar7709},
}

@article{okuma.kawabata.ea2020Topological,
  title = {Topological Origin of Non-{{Hermitian}} Skin Effects},
  author = {Okuma, Nobuyuki and Kawabata, Kohei and Shiozaki, Ken and Sato, Masatoshi},
  date = {2020},
  journaltitle = {Phys. Rev. Lett.},
  volume = {124},
  number = {8},
  pages = {086801, 7},
}

@article{okuma.sato2023Nonhermitian,
  title = {Non-Hermitian Topological Phenomena: A Review},
  author = {Okuma, Nobuyuki and Sato, Masatoshi},
  date = {2023},
  journaltitle = {Annual Rev. Cond. Matt. Phys.},
  volume = {14},
  number = {1},
  pages = {83--107},
}

@book{parlett1998symmetric,
  title = {The Symmetric Eigenvalue Problem},
  author = {Parlett, Beresford N.},
  date = {1998},
  series = {Classics in Applied Mathematics},
  number = {20},
  publisher = {{Society for Industrial and Applied Mathematics}},
  location = {Philadelphia, Pa},
  isbn = {978-0-89871-402-9},
  langid = {english},
  pagetotal = {398}
}

@online{qiu.lin.ea2023Mathematical,
  title = {Mathematical Theory for the Interface Mode in a Waveguide Bifurcated from a {{Dirac}} Point},
  author = {Qiu, Jiayu and Lin, Junshan and Xie, Peng and Zhang, Hai},
  date = {2023},
  eprint = {2304.10843},
  eprinttype = {arXiv},
}

@article{rivero.feng.ea2022Imaginary,
  title = {Imaginary Gauge Transformation in Momentum Space and Dirac Exceptional Point},
  author = {Rivero, Jose H. D. and Feng, Liang and Ge, Li},
  date = {2022},
  journaltitle = {Phys. Rev. Lett.},
  volume = {129},
  number = {24},
  pages = {243901},
}

@book{silva2011Matrix,
  title = {Matrix Perturbations: Bounding and Computing Eigenvalues},
  shorttitle = {Matrix Perturbations},
  author = {da Silva, Ricardo Reis},
  date = {2011},
  publisher = {BOXPress},
  location = {Oisterwijk},
  isbn = {978-90-8891-288-7},
  langid = {english}
}

@article{su.schrieffer.ea1979Solitons,
  title = {Solitons in {{Polyacetylene}}},
  author = {Su, W. P. and Schrieffer, J. R. and Heeger, A. J.},
  date = {1979},
  journaltitle = {Phys. Rev. Lett.},
  volume = {42},
  number = {25},
  pages = {1698--1701},
}

@article {craster.davies2023Asymptotic,
    AUTHOR = {Craster, Richard V. and Davies, Bryn},
     TITLE = {Asymptotic characterization of localized defect modes:
              {S}u-{S}chrieffer-{H}eeger and related models},
   JOURNAL = {Multiscale Model. Simul.},
  FJOURNAL = {Multiscale Modeling \& Simulation. A SIAM Interdisciplinary
              Journal},
    VOLUME = {21},
      YEAR = {2023},
    NUMBER = {3},
     PAGES = {827--848},
}

@article{thiang.zhang2023Bulkinterface,
  title = {Bulk-Interface Correspondences for One-Dimensional Topological Materials with Inversion Symmetry},
  author = {Thiang, Guo Chuan and Zhang, Hai},
  date = {2023},
  journaltitle = {Proc. R. Soc. A.},
  volume = {479},
  number = {2270},
  pages = {20220675},
}

@article{thiang2023Topological,
  title = {Topological Edge States of {{1D}} Chains and Index Theory},
  author = {Thiang, Guo Chuan},
  date = {2023},
  journaltitle = {J. Math. Phys.},
  volume = {64},
  number = {6},
}

@book{trefethen.embree2005Spectra,
  title = {Spectra and Pseudospectra},
  author = {Trefethen, Lloyd N. and Embree, Mark},
  date = {2005},
  pages = {xviii+606},
  publisher = {Princeton University Press, Princeton, NJ},
}

@article{wang.chong2023NonHermitian,
  title = {Non-{{Hermitian}} Photonic Lattices: Tutorial},
  author = {Wang, Qiang and Chong, Y. D.},
  date = {2023},
  journaltitle = {J. Opt. Soc. Am. B},
  volume = {40},
  number = {6},
  pages = {1443--1466},
}

@article{wang.wang.ea2022NonHermitian,
  title = {Non-{{Hermitian}} Morphing of Topological Modes},
  author = {Wang, Wei and Wang, Xulong and Ma, Guancong},
  date = {2022},
  journaltitle = {Nature},
  volume = {608},
  number = {7921},
  pages = {50--55},
}

@article{yao.wang2018Edge,
  title = {Edge States and Topological Invariants of Non-Hermitian Systems},
  author = {Yao, Shunyu and Wang, Zhong},
  date = {2018},
  journaltitle = {Phys. Rev. Lett.},
  volume = {121},
  number = {8},
}

@article{yokomizo.yoda.ea2022NonHermitian,
  title = {Non-{{Hermitian}} Waves in a Continuous Periodic Model and Application to Photonic Crystals},
  author = {Yokomizo, Kazuki and Yoda, Taiki and Murakami, Shuichi},
  date = {2022},
  journaltitle = {Phys. Rev. Res.},
  volume = {4},
  number = {2},
  pages = {023089},
}

@article{yueh.cheng2008Explicit,
  title = {Explicit Eigenvalues and Inverses of Tridiagonal {{Toeplitz}} Matrices with Four Perturbed Corners},
  author = {Yueh, Wen-Chyuan and Cheng, Sui Sun},
  date = {2008},
  journaltitle = {ANZIAM J.},
  volume = {49},
  number = {3},
  pages = {361--387},
}

@article{yueh2005Eigenvalues,
  title = {Eigenvalues of Several Tridiagonal Matrices.},
  author = {Yueh, Wen-Chyuan},
  date = {2005},
  journaltitle = {Appl. Math.},
  volume = {5},
  pages = {66--74},
}

@article{zhang.zhang.ea2022review,
  title = {A Review on Non-{{Hermitian}} Skin Effect},
  author = {Zhang, Xiujuan and Zhang, Tian and Lu, Ming-Hui and Chen, Yan-Feng},
  date = {2022},
  journaltitle = {Adv. Phys.: X},
  volume = {7},
  number = {1},
  pages = {2109431},
}

@article{zheng.achilleos.ea2019Observation,
  title = {Observation of Edge Waves in a Two-Dimensional Su-Schrieffer-Heeger Acoustic Network},
  author = {Zheng, Li-Yang and Achilleos, Vassos and Richoux, Olivier and Theocharis, Georgios and Pagneux, Vincent},
  date = {2019},
  journaltitle = {Phys. Rev. Appl.},
  volume = {12},
  number = {3},
}

@article{dafonseca.petronilho2001Explicit,
  title = {Explicit Inverses of Some Tridiagonal Matrices},
  author = {family=Fonseca, given=C.M., prefix=da, useprefix=true and Petronilho, J.},
  date = {2001},
  journaltitle = {Linear Algebra Appl.},
  volume = {325},
  number = {1--3},
  pages = {7--21},
}

@article{dafonseca2007characteristic,
  title = {The Characteristic Polynomial of Some Perturbed Tridiagonal {{k}}-{{Toeplitz}} Matrices},
  author = {family=Fonseca, given=C.M., prefix=da, useprefix=true},
  date = {2007},
  journaltitle = {Appl. Math. Sci. (Ruse)},
  volume = {1},
  number = {1-4},
  pages = {59--67},
  issn = {1312-885X,1314-7552},
  fjournal = {Applied Mathematical Sciences},
  mrclass = {15A18 (33C45)},
  mrnumber = {2300976},
  mrreviewer = {RicardoL.Soto}
}

@book{2006Orthogonal,
     TITLE = {Orthogonal polynomials and special functions},
    SERIES = {Lecture Notes in Mathematics},
    VOLUME = {1883},
    EDITOR = {Marcell\'an, F. and Van Assche, W.},
 PUBLISHER = {Springer-Verlag, Berlin},
      YEAR = {2006},
     PAGES = {xiv+418},
}

@article{reichel.trefethen1992Eigenvalues,
  title = {Eigenvalues and Pseudo-Eigenvalues of {{Toeplitz}} Matrices},
  author = {Reichel, Lothar and Trefethen, Lloyd N.},
  date = {1992},
  journaltitle = {Linear Algebra Appl.},
  volume = {162--164},
  pages = {153--185},
}

@article{mironov.oleinik1997limits,
  title = {On the Limits of Applicability of the Tight Binding Approximation Method for a Complex-Valued Potential},
  author = {Mironov, A. L. and Oleinik, V. L.},
  date = {1997},
  journaltitle = {Ross\'iiskaya Akademiya Nauk},
  shortjournal = {Teoret. Mat. Fiz.},
  volume = {112},
  number = {3},
  pages = {448--466},
  fjournal = {Teoreticheskaya i Matematicheskaya Fizika}
}

@article{mironov.oleinik1994Limits,
  title = {Limits of Applicability of the Tight Binding Approximation},
  author = {Mironov, A. L. and Oleinik, V. L.},
  date = {1994},
  journaltitle = {Theoretical and Mathematical Physics},
  volume = {99},
  number = {1},
  pages = {457--469},
}

@article{fefferman.lee-thorp.ea2017Honeycomb,
  title = {Honeycomb Schr\"odinger Operators in the Strong Binding Regime},
  author = {Fefferman, Charles L. and Lee-Thorp, James P. and Weinstein, Michael I.},
  date = {2017},
  journaltitle = {Comm. Pure Appl. Math.},
  volume = {71},
  number = {6},
  pages = {1178--1270},
}

@article{thouless1974Electrons,
  title = {Electrons in Disordered Systems and the Theory of Localization},
  author = {Thouless, D.J.},
  date = {1974},
  journaltitle = {Phys. Rep.},
  volume = {13},
  number = {3},
  pages = {93--142},
}

@online{ammari.barandun.ea2024Tunable,
  title = {Tunable {{Localisation}} in {{Parity-Time-Symmetric Resonator Arrays}} with {{Imaginary Gauge Potentials}}},
  author = {Ammari, Habib and Barandun, Silvio and Liu, Ping and Uhlmann, Alexander},
    year = {2024},
  eprint = {2405.05002},
  eprinttype = {arXiv},
}

@online{ammari.barandun.ea2024Generalised,
  title = {Generalised {{Brillouin Zone}} for {{Non-Reciprocal Systems}}},
  author = {Ammari, Habib and Barandun, Silvio and Liu, Ping and Uhlmann, Alexander},
  year = {2024},
  eprint = {2408.05073},
  eprinttype = {arXiv},
}

@article {cauchy,
    AUTHOR = {Hwang, Suk-Geun},
     TITLE = {Cauchy's interlace theorem for eigenvalues of {H}ermitian
              matrices},
   JOURNAL = {Amer. Math. Monthly},
  FJOURNAL = {American Mathematical Monthly},
    VOLUME = {111},
      YEAR = {2004},
    NUMBER = {2},
     PAGES = {157--159},
}

@article{lyra.mayboroda.ea2015Dual,
  title = {Dual Landscapes in {{Anderson}} Localization on Discrete Lattices},
  author = {Lyra, M. L. and Mayboroda, S. and Filoche, M.},
  date = {2015},
  journaltitle = {EPL (Europhysics Lett.)},
  volume = {109},
  number = {4},
  pages = {47001},
}

@article{murphy.wortis.ea2011Generalized,
  title = {Generalized Inverse Participation Ratio as a Possible Measure of Localization for Interacting Systems},
  author = {Murphy, N. C. and Wortis, R. and Atkinson, W. A.},
  date = {2011},
  journaltitle = {Phys. Rev. B},
  volume = {83},
  number = {18},
}

\end{document}